%% file: MT_arxiv.tex
\documentclass[11pt,a4paper]{amsart}
\usepackage[foot]{amsaddr}
\usepackage{ifxetex}
\ifxetex
  \usepackage[no-math]{fontspec}
\else
\fi
\usepackage{amsmath}
\usepackage{amsfonts}
\usepackage{amssymb}
\usepackage{amsthm}
\usepackage{fullpage}
\usepackage{microtype}
\usepackage[libertine]{newtxmath}
\usepackage[tt=false]{libertine} %% tt is ugly 
\usepackage{caption}
\usepackage{bbm}
\usepackage{hyperref, color}
\hypersetup{colorlinks=true,citecolor=blue, linkcolor=blue, urlcolor=blue}
\usepackage[linesnumbered,boxed,ruled,vlined]{algorithm2e}
\usepackage{bm}
\usepackage{bbm}
\usepackage[numbers]{natbib}
\usepackage{xcolor}
\usepackage{enumerate} 
\usepackage{enumitem}
\usepackage{tabularx}
\usepackage{array}
\usepackage{cleveref}
\newcolumntype{L}[1]{>{\raggedright\arraybackslash}p{#1}}
\newcolumntype{C}[1]{>{\centering\arraybackslash}m{#1}}
\newcolumntype{R}[1]{>{\raggedleft\arraybackslash}p{#1}}

\usepackage{makecell}

\usepackage{cleveref}
\usepackage{color}
\usepackage{graphicx}
\usepackage{float}
\usepackage{bbm}

\usepackage{footnote}
\makesavenoteenv{tabular}

\newtheorem{thm}{\textbf{Theorem}}
\newenvironment{thmbis}[1]
  {\renewcommand{\thethm}{\textbf{\ref{#1}}}%
   \addtocounter{thm}{-1}%
   \begin{thm}}
  {\end{thm}}

\newtheorem{thms}{\textbf{Proposition}}

\newtheorem{theorem}{Theorem}[section]
\newtheorem{claim}[theorem]{Claim}

\newtheorem{lemma}[theorem]{Lemma}
\newtheorem{definition}[theorem]{Definition}

\newtheorem{remark}[theorem]{Remark}
\newtheorem{proposition}[theorem]{Proposition}
\newtheorem{condition}[theorem]{Condition}

\newtheorem{fact}[theorem]{Fact}

\newcommand{\Neighbor}{\mathcal{N}}

\newcommand{\DependencyGraph}{G_D}

\newcommand{\BipartiteGraph}{G_B}
\newcommand{\EventSet}{\mathcal{A}}

\newcommand{\Event}{A}

\newcommand{\Neg}[1]{\overline{#1}}
\renewcommand{\vec}[1]{\mathbf{#1}}
\renewcommand{\Pr}{\mathbb{P}}

\newcommand{\Interior}{\mathcal{I}}

\newcommand{\dist}{\mathrm{dist}}

\newcommand{\E}{\mathbb{E}}
\newcommand{\N}{\mathcal{N}}

\newcommand{\eps}{\varepsilon}

\newcommand{\A}{\mathcal{A}}
\newcommand{\M}{\mathcal{M}}

\renewcommand{\d}{\mathrm{d}}

\newcommand{\mdl}{\Delta}

\newcommand{\vbrl}{\mathrm{vbl}}

\newcommand{\vecdelta}{\boldsymbol{\delta}}

\renewcommand{\vec}[1]{\boldsymbol{#1}}

\newcommand{\lowerBB}{\digamma}
\newcommand{\Location}{\mathcal{L}}

\newcommand{\vbl}{\mathrm{vbl}}
\newcommand{\D}{\mathcal{D}}

\renewcommand{\d}{\mathrm{d}}

\newcommand{\Rem}{\mathscr{S}}
\newcommand{\Mat}{\mathscr{M}}

\newcommand{\locg}{K}

\newcommand{\warn}[1]{{\color{red}{#1}}}

\newcommand{\lqc}[1]{#1}

\title{Moser-Tardos Algorithm: Beyond Shearer's Bound}
\date{}

\author{Kun He}
\author{Qian Li}
\author{Xiaoming Sun}

\address[Kun He, Qian Li, and Xiaoming Sun]{Institute of Computing Technology, Chinese Academy of Sciences, Beijing, China. University of Chinese Academy of Sciences. Beijing, China. \textnormal{E-mail: \url{hekun.threebody@foxmail.com}, \url{liqian@ict.ac.cn}, and \url{sunxiaoming@ict.ac.cn}}.}

\begin{document}

\maketitle

\begin{abstract}

In a seminal paper (Moser and Tardos, JACM’10), Moser and Tardos developed a simple and powerful algorithm to find solutions to combinatorial problems in the \emph{variable} Lov{\'a}sz Local Lemma (LLL) setting. Kolipaka and Szegedy (Kolipaka and Szegedy, STOC’11) proved that the Moser-Tardos algorithm is efficient up to the tight condition of the \emph{abstract} Lov{\'a}sz Local Lemma, known as Shearer's bound.
A fundamental problem 
around LLL is whether the efficient region of the Moser-Tardos algorithm can be further extended.

In this paper, we give a positive answer to this problem. We show that the efficient region of the Moser-Tardos algorithm goes beyond the Shearer's bound of the underlying dependency graph, if the graph is not chordal. Otherwise, the dependency graph is chordal, and it has been shown that Shearer's bound exactly characterizes the efficient region for such graphs (Kolipaka and Szegedy, STOC’11; He, Li, Liu, Wang and Xia, FOCS’17).

Moreover, we demonstrate that the efficient region can exceed Shearer's bound by a constant by explicitly calculating the gaps on several infinite lattices. %can be constant by exp t} %\emph{significantly} 
%exceed Shearer's bound by a constant explicitly calculating the gaps on several lattices.??

The core of our proof is a new criterion on the efficiency of the Moser-Tardos algorithm which takes the intersection between dependent events into consideration. Our criterion is strictly better than Shearer's bound whenever the intersection exists between dependent events. Meanwhile, if any two dependent events are mutually exclusive, our criterion becomes the Shearer's bound, which is known to be tight in this situation for the Moser-Tardos algorithm (Kolipaka and Szegedy, STOC’11; Guo, Jerrum and Liu, JACM’19).

\end{abstract}
%\newpage

\section{Introduction}
Suppose \lqc{$\EventSet=\{\Event_1,\cdots,\Event_m\}$} %$\EventSet=(\Event_1,\cdots,\Event_m)$
is a set of bad events. If the events are mutually independent, then we can avoid all of these events %all the bad events 
simultaneously whenever no event has probability 1. Lov{\'a}sz Local Lemma (LLL)~\cite{erdos1975problems}, one of the most important probabilistic methods, allows for limited dependency among the events, but still concludes that all the events can be avoided simultaneously if each individual event has a bounded probability.
In the most general setting (a.k.a. abstract LLL), the dependency among $\EventSet$ is characterized by an undirected graph $G_D=([m],E_D)$, called a
\emph{dependency graph} of $\EventSet$, which satisfies that for any vertex $i$, $\Event_i$ is independent of $\{\Event_j: j\notin \Neighbor_{G_D}(i) \cup \{i\}\}$. Here $\Neighbor_{G}(i)$ stands for the set of neighbors of vertex $i$ in a given graph $G$. 

We use $\EventSet\sim(G_D,\vec{p})$ to denote that (i) $G_D$ is a dependency graph of $\EventSet$ and (ii) the probability vector of $\EventSet$ is $\vec{p}$. 
Given a graph $G_D$, define the \emph{abstract interior} $\Interior_a(G_D)$ to be the set consisting of all vectors $\vec{p}$ such that $\Pr\left(\cap_{\Event\in \EventSet} \Neg{\Event} \right)>0$ for any $\EventSet\sim(G_D,\vec{p})$. In this context, the most frequently used abstract LLL can be stated as follows:
\begin{theorem}[\cite{spencer1977asymptotic}]\label{thm:asymmetric}
Given any graph $\DependencyGraph=([m],E_D)$ and any probability vector $\vec{p}\in (0,1]^m$, if there exist real numbers $x_1,...,x_m \in (0,1)$ such that $p_i \leq x_i \prod_{j\in  \Neighbor_{G_D}(i)} (1-x_j)$ for any $i\in[m]$, then $\vec{p}\in\Interior_a(G_D)$.%$\Pr\left(\cap_{\Event\in \EventSet} \Neg{\Event} \right)>0$ for any $\EventSet\sim(G_D,\vec{p})$.
\end{theorem}

Shearer~\cite{shearer1985problem} obtained the strongest possible condition for abstract LLL.  %the general setting LLL, known as Shearer's bound.
Let $\text{Ind}(G_D)$ be the set of all independent sets of an undirected graph $G_D=([m],E_D)$ and $\vec{p}=(p_1,\cdots,p_m)\in (0,1]^m$. For each $I\in \text{Ind}(G_D)$, define the quantity
\[
q_{I}(G_D,\vec{p})=\sum_{J\in \text{Ind}(G_D),I\subseteq J}(-1)^{|J|-|I|}\prod_{i\in J}p_i.
\]
%\begin{definition}[Shearer's bound]\label{def:shearerbound}
$\vec{p}$ is called \emph{in Shearer's bound} of $G_D$ if $q_{I}(G_D,\vec{p})>0$ for any $I\in \text{Ind}(G_D)$. Otherwise we say $\vec{p}$ is \emph {beyond Shearer's bound} of $G_D$. 
Shearer's result can be stated as follows.

%{\color{red} on Shearer's bound}
%\end{definition}
\begin{theorem}[\cite{shearer1985problem}]\label{thm:shearer1985problem}
For any %dependency
graph $G_D =([m],E_D)$ and any probability vector $\vec{p}\in (0,1]^{m}$, $\vec{p}\in \Interior_a(G_D)$ if and only if $\vec{p}$ is in Shearer's bound of $G_D$.%the following conditions are equivalent:
%\begin{enumerate}
%  \item $\vec{p}$ is in Shearer's bound of $G_D$.
%  \item for any probability space $\Omega$ and events $\{A_i\subseteq\Omega:i\in [m]\}$  having $G_D$ as dependency graph and satisfying $\Pr(A_i)=p_i$, we have $\Pr(\cap_{i\in [m]}\overline{A_i})\geq q_{\emptyset}(G_D,\vec{p})>0$.
%\end{enumerate}
\end{theorem}

%That is, given any undirected graph $G_D=([m],E_D)$ and $\vec{p}\in(0,1]^m$, there is a set $\EventSet\sim(G_D,\vec{p})$ of event such that $\Pr(\cap_{i\in[m]}\Neg{\Event_i})=0$ if and only if $\vec{p}$ is not in Shearer's bound of $G_D$.
\vspace{1.5ex}
\noindent\underline{Variable Lov{\'a}sz Local Lemma.} Variable Lov{\'a}sz Local Lemma (VLLL) is another quite general and common setting of LLL,
 which applies to variable-generated event systems. % exploits richer structures of the events.
 In this setting, there is a set of underlying mutually independent random variables $\{X_1,\cdots,X_n\}$, and each event $\Event_i$ can be fully determined by some variables $\vbrl(A_i)$ of them. The dependency between events and variables can be naturally characterized by a bipartite graph $G_B=([m],[n],E_B)$, known as the event-variable graph, such that edge $(i,j) \in [m]\times [n]$ exists if and only if $X_j \in \vbrl(A_i)$. %an edge $(i,j) \in [m]\times [n]$ exists if $X_j \in \vbrl(A_i)$.

The variable setting is important, mainly because most applications of LLL have natural underlying independent variables, such as the satisfiability of CNF formulas\cite{gebauer2009lovasz,gebauer2016local,moitra2019approximate,feng2020fast}, hypergraph coloring \cite{mcdiarmid1997hypergraph,guo2019counting},
and Ramsey numbers\cite{spencer1975ramsey,spencer1977asymptotic,harris2016lopsidependency}. %, etc \cite{giotis2017acyclic}.
In particular, the groundbreaking result by Moser and Tardos \cite{moser2010constructive} on constructive LLL applies in the variable setting. 

There is a natural choice for the dependency graph of variable-generated systems, called the \emph{canonical dependency graph}: two events are adjacent if they share some common variables. Formally, given a bipartite graph $\BipartiteGraph=(U,V,E_B)$, its \emph{base graph} is defined as the graph  $\DependencyGraph(\BipartiteGraph)=(U,E_D)$
such that for any two vertices $u_i, u_j \in U$, $(u_i, u_j) \in E_D$ if and only if $u_i$ and $u_j$ share common neighbors in $\BipartiteGraph$. %there is an edge $(u_i, u_j) \in E_D$ if and only if $u_i$ and $u_j$ share some common neighbor in $\BipartiteGraph$. 
If $\BipartiteGraph$ is the event-variable graph of a variable-generated system $\EventSet$, then $\DependencyGraph(\BipartiteGraph)$
is the canonical dependency graph of $\EventSet$.

Given a graph $G_D$, define the variable interior $\Interior_v(G_D)$ to be the set consisting of all vectors $\vec{p}$ such that $\Pr\left(\cap_{\Event\in \EventSet} \Neg{\Event} \right)>0$ for any \emph{variable-generated} event system $\EventSet\sim(G_D,\vec{p})$. 
Obviously, $\Interior_v(G_D)\supseteq\Interior_a(G_D)$ for any $G_D$.
In contrast with the abstract LLL, the Shearer's bound (of the canonical dependency graph) turns out to be not tight for variable-generated systems \cite{he2017variable}: the containment is proper %, i.e., $\Interior_v(G_D) \supsetneq \Interior_a(G_D)$ 
if and only if $G_D$ is not chordal\footnote{A graph is chordal if it has no induced cycle of length at least four.}.

\vspace{1.5ex}
\noindent\underline{Constructive (variable) Lov{\'a}sz Local Lemma and Moser-Tardos algorithm.}
The abstract LLL and the variable LLL mentioned above are not constructive in that they do not indicate how to efficiently find an object avoiding all the bad events.
In a seminal paper \cite{moser2010constructive}, Moser and Tardos developed an amazingly simple efficient algorithm for variable-generated systems, depicted in Algorithm \ref{alg:mt}\footnote{Throughout the paper, the Moser-Tardos algorithm is allowed to follow arbitrary selection rules.}, and showed that this algorithm terminates quickly under the condition in Theorem~\ref{thm:asymmetric}. 
Following the Moser-Tardos algorithm (or MT algorithm for short), a large amount of effort devoted to constructive LLL, including the remarkable works which extend the MT techniques beyond the variable setting
%On the one hand, the MT techniques have been extended beyond the variable setting
~\cite{HarrisS14a,AchlioptasIV17,AchlioptasIK19,AchlioptasIS19,IliopoulosS20,HarveyV20}. %On the other hand, 
The MT algorithm has been applied to many important problems, including $k$-SAT~\cite{gebauer2016local}, hypergraph coloring~\cite{harris2016lopsidependency}, Hamiltonian cycle~\cite{harris2016lopsidependency}, and their counting and sampling %the counting and uniform sampling of $k$-SAT
~\cite{guo2019uniform,moitra2019approximate,feng2020fast,feng2021sampling,Vishesh21sampling,he2021perfect}. 
%\vspace{-2ex}

\begin{algorithm}[htb]
 	\caption{Moser-Tardos Algorithm}
 	\label{alg:mt}
  	%\textbf{Parameter}: $\vectheta^\ast$
 	
    Assign random values to $X_1,\cdots,X_n$\;
   \While{$\exists i~\in[m]$ such that $A_i$ holds}{
 	  Arbitrarily select %choose
 	  one such $i$ and resample all variables $X_j$ in $\vbrl(A_i)$;
    }
   Return the current assignment;
\end{algorithm}
%\vspace{-2ex}

Mainly because such a simple algorithm is so powerful and general-purpose, %general-purposed,
it is one of the most intriguing and fundamental problems on constructive LLL how powerful the MT algorithm is. Given a graph $G_D$, define the \emph{Moser-Tardos interior} $\Interior_{MT}(G_D)$ to be the set consisting of all vectors $\vec{p}$ such that the MT algorithm is efficient for any \emph{variable-generated} event system  %systems
$\EventSet\sim(G_D,\vec{p})$. Clearly, $\Interior_{MT}(G_D)\subseteq\Interior_v(G_D)$ for any $G_D$.
A major line of follow-up works explores $\Interior_{MT}(G_D)$
~\cite{CLLL,pegden2014extension,kolipaka,catarata2017moser}. 
The best known criterion is obtained by Kolipaka and Szegedy~\cite{kolipaka}.
They extended the MT interior to the Shearer's bound. That is, they showed that $\Interior_{MT}(G_D)\supseteq \Interior_a(G_D)$.
%, i.e., $\Interior_{MT}(G_D)\supseteq \Interior_a(G_D)$. 
As mentioned above, if $G_D$ is not chordal, $\Interior_a(G_D)$ is properly contained in $\Interior_v(G_D)$, so it is possible to further push $\Interior_{MT}(G_D)$ beyond Shearer's bound.

\vspace{1ex}
In this paper, we concentrate on the following open problem:

\vspace{1ex}
\noindent\textit{\textbf{Problem 1}: does $\Interior_{MT}(G_D)$ \emph{properly} contain $\Interior_a(G_D)$ for some $G_D$? If so, for what kind of graph $G_D$? 
}
\vspace{1ex}

Rather than potential applications, our main motivations are the following fundamental problems around LLL itself:

\begin{itemize}[leftmargin=20pt]
    \item \emph{The limitation of the constructive LLL in the variable setting.}
    In the most fascinating problems around LLL, 
    a mysterious conjecture says that
    there is an algorithm which is efficient for all variable-generated systems $\A$ if 
    $\A \sim (G_D,\vec{p})$ for some $G_D$ and $\vec{p}\in \Interior_{v}(G_D)$ \cite{szegedy2013lovasz}. 
    It would be a small miracle if the conjecture is true, since if so, one can always \emph{construct} a solution efficiently in the variable setting if %the
    solutions are guaranteed to \emph{exist} 
    by the LLL condition.
    Towards this conjecture%To solve the problem
    , a good start is to show that $\Interior_{MT}(G_D)\supsetneqq \Interior_a(G_D)$ for some $G_D$, as %since
    $\Interior_v(G_D)\supsetneqq \Interior_a(G_D)$ for $G_D$ which is not chordal.
    
    \item \emph{The limitation of the MT algorithm.} The MT algorithm is one of the most intriguing topics in modern algorithm researches, not only because it is very simple and with magic power, 
    but also because it is closely related to the famous Walksat algorithm for random $k$-SAT. 
    A mysterious problem about the MT algorithm is where is its true limitation~\cite{szegedy2013lovasz,catarata2017moser}. It is conjectured that $\Interior_{MT}(G_D) = \Interior_{v}(G_D)$ for any $G_D$~\cite{szegedy2013lovasz}.
    To prove this %the
     conjecture, the first step is to give a positive answer to Problem 1.
    Moreover, due to the connection between Shearer's bound and the Repulsive Lattice Gas model, it is conjectured that \emph{essential connection exists between statistical mechanics and the MT algorithm}~\cite{szegedy2013lovasz}.
    Whether $\Interior_{MT}(G_D) = \Interior_{a}(G_D)$ for each $G_D$ is critical to this conjecture.
\end{itemize}

%\lqc{%In fact, the phenomenon that MT is still efficient beyond Shearer's bound was known to exist, in \emph{some specific applications} only (e.g., \cite{harris2016lopsidependency,catarata2017moser}). Here, instead of \emph{how far} beyond Shearer's bound in specific applications, the primary concern of this work is the  \emph{universality} of this phenomenon. In specific, this work is mainly motivated by the following open question:

%We care about this question not only because of the potential improvements of specific applications but also because it lead to deeper understanding about the nature of constructive LLL and the dynamics of MT algorithm. In specific, this question is closely related to the following two fundamental problems about constructive LLL and MT algorithm.
%}

\begin{remark}
To explore the power of the MT algorithm in specific applications, one may employ %the
special structures of the %in these 
applications, such as the way the variables interact, to obtain sharp bounds rather than in terms of the canonical  %abstract 
dependency graph only.
Nevertheless, characterizing the power of the MT algorithm in terms of the canonical %abstract 
dependency graph is a very fundamental problem and also the focus of the major line of researches~\cite{moser2010constructive,pegden2014extension,bissacot2011improvement,kolipaka}.
Moreover, a major difficulty to strengthen %strength 
the guarantees of the MT algorithm is that the analysis should be valid for all possible variable-generated event systems.
It is not quite surprising to obtain better bounds 
if the event system has further restrictions.
To substantially improve %really expand 
the guarantees of the MT algorithm and provide deep insight about its dynamics,
%\lqc{we focus on the general variable LLL setting instead of specific applications with special structures.} 
we would rather focus on the general variable LLL setting than employ the special structures in the applications.
\end{remark}

We should emphasize that Problem 1 is still quite open! 
As mentioned before, it has been proved that 
the Shearer's bound is not tight for variable-generated systems \cite{he2017variable}.
However, this only says that there is some probability vector $\vec{p}$ beyond the Shearer's bound such that all variable-generated event systems $\EventSet \sim (G_D,\vec{p})$ must have a satisfying assignment. It is unclear whether the MT algorithm can construct such an assignment efficiently.

It also has been proved that the MT algorithm can still be efficient even beyond the Shearer’s bound \emph{for some specific applications}~\cite{harris2016lopsidependency}. 
Despite its novel contribution, this result does 
not provide an answer to Problem 1. The result in~\cite{harris2016lopsidependency} focuses on the event systems with special structures. %in the applications.
Thus, it only implies that there is a probability vector $\vec{p}$ beyond the Shearer's bound such that the MT algorithm is efficient for \textbf{some restricted} variable-generated event systems $\EventSet\sim (G_D,\vec{p})$.
However, to show $\Interior_{MT}(G_D)\supsetneqq \Interior_a(G_D)$,
one must prove that the MT algorithm is efficient for \textbf{all possible} event systems,
and this is one \lqc{major difficulty to resolve}  %of the major difficulties of 
Problem 1.

\subsection{Results and contributions}
We provide a complete answer to Problem 1 \lqc{(Theorem \ref{thm-dep-graph-beyond})}: if $G_D$ is not chordal, then $\Interior_{MT}(G_D)\supsetneqq\Interior_a(G_D)$, i.e., the efficient region of the MT algorithm goes beyond Shearer's bound. Otherwise, $\Interior_{MT}(G_D)=\Interior_a(G_D)$, \lqc{because} %since
$\Interior_a(G_D) \subseteq \Interior_{MT}(G_D)\subseteq \Interior_v(G_D)$ and $\Interior_v(G_D)=\Interior_a(G_D)$ for chordal graphs $G_D$  ~\cite{he2017variable}.

\lqc{The core of the proof of Theorem \ref{thm-dep-graph-beyond} is a new convergence criterion for the MT algorithm \lqc{(Theorem \ref{thm:main})}, which may be of independent interest. This new criterion takes the intersection between dependent events into consideration, and is strictly better than Shearer’s bound when there exists a pair of dependent events which are not mutually exclusive. 
}

%To solve Problem 1, we also prove a new convergence criterion for the MT algorithm \lqc{(Theorem \ref{thm:main})}, which takes the intersections between dependent events into consideration.
%This criterion goes beyond the Shearer’s bound if \lqc{not all pairs of dependent events are  mutually exclusive}. %any two dependent events are not mutually exclusive.

%\lqc{We show that the phenomenon that MT is still efficient beyond Shearer's bound is \emph{universal}, thus give an affirmative answer to Question 1! Specifically, we show that: If the canonical dependency graph is not chordal\footnote{A graph is chordal if it has no induced cycle of length at least four.}, then the MT algorithm can still be efficient on general instances even  beyond Shearer's bound. Meanwhile, 
%if the canonical dependency graph is chordal, then Shearer's bound is tight for non-constructive VLLL~\cite{he2017variable}, hence also tight for the MT algorithm.
%}

%Comparing to the new LLL conditions with important applications \cite{HarrisS14a,AchlioptasI16,harris2016lopsidependency,HarveyV20,AchlioptasIS19},
%it seems that our results have no immediate applications in theoretic computer science.
%However, the techniques proposed in our analysis may be further used to 
%the applications of LLL and the MT algorithm.

 \subsubsection{Moser-Tardos algorithm: beyond Shearer's bound}

Given a dependency graph $G_D=([m],E_D)$ and a probability vector $\vec{p}=(p_1,p_2,\cdots,p_m)\in(0,1)^m$,
we say \lqc{that} $\vec{p}$ is on the Shearer's boundary of $G_D$ if $(1 - \eps)\vec{p}$ is in Shearer's bound
and $(1 + \eps)\vec{p}$ is not for any $\eps>0$. A chordless cycle in a graph $G_D$ is an induced cycle of length at least 4. A chordal graph is a graph without chordless cycles.

%Given a vector $\vec{p}$ on the Shearer's boundary of $G_D$, we say $\vec{p}$ is not tight\lqfoot{Where is this "tight"  used} for $G_D$ in the variable setting if $\Pr\left(\cap_{\Event\in \EventSet} \Neg{\Event} \right)>0$ for any variable-generated event system $\EventSet\sim(G_D,\vec{p})$.

Given two vectors $\vec{p}$ and $\vec{q}$, we say $\vec{p}\leq \vec{q}$ if the inequality holds entry-wise. Additionally, if the
inequality is strict on at least one entry, we say that $\vec{p}< \vec{q}$.
%In this paper, we assume $\sup \{\emptyset\} = -1$.

\begin{definition}[Maximum $L_1$-gap to the Shearer's bound]\label{def-distance}
Given a dependency graph $G_D$ and a probability vector $\vec{p}$ beyond the Shearer's bound of $G_D$, 
define the maximum $L_1$-gap from $\vec{p}$ to the Shearer's bound of $G_D$ as
\[
d(\vec{p},G_D) \triangleq \arg \sup_{||\vec{q}||_1}\{\vec{p} - \vec{q} \not\in \Interior_a(G_D):\vec{q}\leq \vec{p}\}.
\]
For \lqc{convenience, }%simplify, 
we let $d(\vec{p},G_D) = -1$ if $\vec{p}$ is in the Shearer's bound of $G_D$.
\end{definition}
Intuitively, $d(\vec{p},G_D)$ \lqc{measures }%characterizes 
how far $\vec{p}$ is from the Shearer's bound of $G_D$.
One can verify that $d(\vec{p},G_D)<0$ if 
$\vec{p}$ is in the Shearer's bound, $d(\vec{p},G_D)=0$ if $\vec{p}$ is on the Shearer's boundary, and $d(\vec{p},G_D)>0$ if $\vec{p}$ is beyond Shearer's bound but not on the Shearer's boundary. %, we have $d(\vec{p},G_D)> 0$.
\lqc{Now, we are ready to state our main result.}%Then we have the following theorem.

 \iffalse
\lqc{\begin{theorem}\label{thm-dep-graph-beyond}
 For any chordal graph $G_D$, $\Interior_{MT}(G_D)=\Interior_a(G_D)$, i.e., $\vec{p}\in \Interior_{MT}(G_D)$ iff $d(\vec{p},G_D)<0$.
 
 For any graph $G_D$ which is not chordal,  $\vec{p}\in \Interior_{MT}(G_D)$ if %the MT algorithm is efficient for all vectors $\vec{p}$ if
 \[d(\vec{p},G_D) < \frac{1}{545}\cdot\sum_{i\leq \ell} |C_i|\big(\min_{j\in C_i} p_j\big)^4\cdot\left(\max\left\{\frac{2\sum_{j\in C_i} \sqrt{p_j}}{|C_i|}  - 1,0\right\}\right)^2
\]
for some disjoint chordless cycles  $C_1,C_2,\cdots,C_\ell$ in $G_D$.
In particular, for each chordless cycle $C$, there is a probability vector $\vec{p}$ such that 
\[
2^{-20} |C|^{-3}\leq d(\vec{p},G_D) < \frac{1}{545}\cdot |C|\big(\min_{j\in C} p_j\big)^4\cdot\left(\max\left\{\frac{2\sum_{j\in C} \sqrt{p_j}}{|C|}  - 1,0\right\}\right)^2.
\]
This implies that $\Interior_{MT}(G_D)$ contains a probability vector $\vec{p}$ with $d(\vec{p},G_D)\geq 2^{-20} L^{-3}$,
where $L$ denotes the length of the shortest chordless cycle in $G_D$.
 \end{theorem}
 }
 \fi

 \begin{theorem}\label{thm-dep-graph-beyond}
 For any chordal graph $G_D$, $\Interior_{MT}(G_D)=\Interior_a(G_D)$, i.e., $\vec{p}\in \Interior_{MT}(G_D)$ iff $d(\vec{p},G_D)<0$.
 
 For any graph $G_D$ which is not chordal,  $\vec{p}\in \Interior_{MT}(G_D)$ if %$\vec{p}$ satisfies the following condition
\begin{equation*}\label{eq-main}
 d(\vec{p},G_D) < \frac{1}{545}\cdot\sum_{i\leq \ell} |C_i|\big(\min_{j\in C_i} p_j\big)^4\cdot\left(\max\left\{\frac{2\sum_{j\in C_i} \sqrt{p_j}}{|C_i|}  - 1,0\right\}\right)^2
 \end{equation*}
for some disjoint chordless cycles  $C_1,C_2,\cdots,C_\ell$ in $G_D$.
In particular,
there is a probability vector $\vec{p}$ with  
$d(\vec{p},G_D)\geq 2^{-20} \locg^{-3}$ satisfying the above condition,  %(\ref{eq-main}) for the shortest chordless cycle, 
where $\locg$ is the length the shortest chordless cycle.
This implies that $\Interior_{MT}(G_D)$ contains a probability vector $\vec{p}$ with $d(\vec{p},G_D)\geq 2^{-20} \locg^{-3}$.
 \end{theorem}

%\vspace{1ex}

\iffalse
\warn{The intuition of Theorem~\ref{thm-dep-graph-beyond} is as follows. 
Theorem~\ref{thm-dep-graph-beyond} characterizes the efficient region of the MT algorithm with $d(\vec{p},G_D)$. %It guarantees the efficiency of the MT algorithm if $d(\vec{p},G_D)$ is upper bounded by a nonnegative term
It shows that if $d(\vec{p},G_D)$ is upper bounded by a \emph{non-negative} quantity with respect to the chordless cycles in $G_D$,
then the MT algorithm is efficient. 
We should emphasize that for each $G_D$ which is not chordal, there is some $\vec{p}$ with $d(\vec{p},G_D)>0$ such that the quantity is a \emph{strictly positive} upper bound of 
$d(\vec{p},G_D)$. 
Specifically, the quantity can be very large such that 
$2^{-20}\kappa^{-3}<d(\vec{p},G_D)$ is upper bounded by the quantity.
%Obviously, our characterization is at least as good as Shearer's bound, since $I_{}$
%where the upper bound is related the chordless cycles in $G_D$ and the vector $\vec{p}$.
%We should emphasize that our upper bound in the theorem is nonnegative,
%Moreover, we shows that our characterization of $\I_{MT}(G_D)$ 
%and for each chordless $G_D$, there is some $\vec{p}$ such that the upper bound is strictly positive. 
Note that $\vec{p}\in \Interior_a(G_D)$ if and only if $d(\vec{p},G_D)<0$,
thus, our characterization of $\Interior_{MT}(G_D)$ is strictly better than the Shearer's bound. 
In summary, Theorem~\ref{thm-dep-graph-beyond} implies that \emph{chordless cycles in $G_D$ enhance the power of the MT algorithm}. }
\fi

The intuition of Theorem~\ref{thm-dep-graph-beyond} is as follows. 
The theorem characterizes the efficient region of the MT algorithm with $d(\vec{p},G_D)$. %It guarantees the efficiency of the MT algorithm if $d(\vec{p},G_D)$ is upper bounded by a nonnegative term
It shows that if $d(\vec{p},G_D)$ is upper bounded by a \emph{non-negative} quantity related to %with respect to 
the chordless cycles in $G_D$,
then the MT algorithm is efficient. Since $\Interior_a(G_D)$ is the set of $\vec{p}$ where $d(\vec{p},G_D)<0$, our criterion is at least as good as Shearer's bound. Moreover, for each $G_D$ which is not chordal, our criterion is strictly better: there exists some $\vec{p}$ with $d(\vec{p},G_D)\geq  2^{-20} \locg^{-3}$ satisfying our criterion. %such that the quantity is a \emph{strictly positive} upper bound of $d(\vec{p},G_D)$. 
%Thus, $\vec{p}$ is in $\Interior_{MT}(G_D)$.
Intuitively, Theorem~\ref{thm-dep-graph-beyond} implies that \emph{chordless cycles in $G_D$ enhance the power of the MT algorithm}.

We emphasize that Theorem~\ref{thm-dep-graph-beyond} provides a \emph{complete} answer to Problem 1: $\Interior_{MT}(G_D)$ properly contains $\Interior_a(G_D)$ if and only if $G_D$ is not chordal.

%Theorem~\ref{thm-dep-graph-beyond} implies that \emph{chordless cycles in $G_D$ enhance the power of the MT algorithm}. 
%We should emphasize that Theorem~\ref{thm-dep-graph-beyond} provides a \emph{complete} answer to Problem 1: if $G_D$ is not chordal, there exists some $\vec{p} \in \Interior_{MT}(G_D)\setminus \Interior_a(G_D)$,
%thus $\Interior_{MT}(G_D)\supsetneqq \Interior_a(G_D)$;
%otherwise $\Interior_{MT}(G_D) = \Interior_a(G_D)$.

%\begin{theorem}[Qualitative version]\label{thm-dep-graph-beyond-quality}
% 	Given a dependency graph $G_D$ and a vector $\vec{p}$ on the Shearer's boundary of $G_D$, if $\vec{p}$ is not tight for $G_D$ in the variable setting, then the MT algorithm is efficient at all $\vec{q}\leq \vec{p}'$ for some $\vec{p}' >\vec{p}$.
% \end{theorem}
 
%Theorem~\ref{thm-dep-graph-beyond-quality} implies that if the \emph{existence} is beyond the Shearer's bound at some direction of probability vector, then the \emph{construction} is also beyond the Shearer's bound at that direction. 
%It can be viewed as a complement of Theorem~\ref{thm-dep-graph-beyond},
%in which we only prove that the MT algorithm is efficient beyond the Shearer's bound at some particular directions.

 \subsubsection{A new constructive LLL for non-extremal instances}\label{sec:intro_illl}
Given \lqc{a set $\EventSet$ of events} %a set of events $\EventSet$ 
with dependency graph $G_D$,
$\EventSet$ is called \emph{extremal} if all pairs of  
dependent events are mutually exclusive, and \emph{non-extremal} otherwise.
%otherwise, it is called  \emph{non-extremal}.
\lqc{Kolipaka and Szegedy \cite{kolipaka} showed that the MT algorithm is efficient up to the Shearer's bound. In particular, Shearer's bound is the tight convergence criterion for extremal instances \cite{kolipaka,guo2019uniform}.}
%In the elegant proof of that the MT algorithm is efficient up to the Shearer's bound, Kolipaka and Szegedy \cite{kolipaka} obtained an upper bound on the complexity of the MT algorithm which becomes infinite when the probability vector goes beyond Shearer's bound. This upper bound \lqc{was} %is
%shown to be tight for \lqc{every} extremal instance %\cite{kolipaka,guo2019uniform}.
\lqc{Here, we provide a new convergence criterion (Theorem \ref{thm:main}) which is a strict improvement of Kolipaka and Szegedy's result: this criterion is strictly better than Shearer's bound when the instance is non-extremal, and becomes Shearer's bound when the instance is extremal.}
%As a complement to their results, 
%we provides a convergence criterion which goes beyond the Shearer's bound for non-extremal instances.
%When applied to extremal instances,
%our new criterion becomes the Shearer's bound and thus is tight.
This criterion, named \emph{intersection LLL}, is the core of our proof of Theorem~\ref{thm-dep-graph-beyond}.

Let $G_D=([m],E_D)$ be a canonical dependency graph and $\vec{p}=(p_1,\cdots,p_m)\in (0,1)^m$ be a probability vector. Let $\M=\{(i_1,i_1'),(i_2,i_2'),\cdots\} \subseteq E_D$ be a matching of $G_D$, and  $\vecdelta=(\delta_{i_1,i_1'},\delta_{i_2,i_2'},\cdots)\in(0,1)^{|\M|}$ be another probability vector.
 We say that an event set $\A$ is of the setting $(G_D,\vec{p},\M,\vecdelta)$, and write $\A\sim(G_D,\vec{p},\M,\vecdelta)$, if $\A\sim(G_D,\vec{p})$
 and $\Pr(A_i\cap A_{i'})\geq \delta_{i,i'}$ for each pair $(i,i')\in\M$.
 Given $(G_D,\vec{p},\M,\vecdelta)$, define $\vec{p}^{-}\in(0,1)^m$ as follows:
 \begin{align*}
 \forall i\in [m]:\quad
 p_{i}^-=\begin{cases}
 p_i-\frac{1}{17}\cdot \delta_{i,i'}^2, & \text{if } (i,i')\in\M \text{ for some } i';\\
 p_i, & \text{otherwise}.
 \end{cases}
 \end{align*}
 
 \begin{theorem}[intersection LLL (informal)]\label{thm:main}
 	For any $\A\sim(G_D,\vec{p},\M,\vecdelta)$, MT algorithm terminates quickly %in linear time
 	if $\vec{p}^-$ is in the Shearer's bound of $G_D$.
 \end{theorem}
 
The intuition of Theorem~\ref{thm:main} is as follows.
For any matching $\M$ in $G_D$, if the intersection of events on each edge $(i,i')$ \lqc{in} %of 
$\M$ has a lower bound $\delta_{i,i'}$,
\lqc{then} one can subtract $\frac{1}{17}\cdot \delta^2_{i,i'}$ from the probabilities of endpoints $i$ and $i'$, and 
the MT algorithm is \lqc{guaranteed to be }%still 
efficient \lqc{whenever }%if 
the reduced probability vector is in the Shearer's bound.

\begin{remark}
In many applications of LLL~\cite{mcdiarmid1997hypergraph,gebauer2016local,gebauer2009lovasz,moitra2019approximate,giotis2017acyclic},
the dependent %(or even lopsidependent)
bad events naturally intersect with each other.
For instance, \lqc{in a CNF formula, if the common variables in two clauses are both either positive or negative, then} %for any CNF formula
%where the common variables in two of the clauses are both positive or negative,
the bad events corresponding to these two clauses are dependent and
%they 
intersect with each other.
Thus our intersection LLL may be capable of improving bounds for these applications.
However, currently the improvement is weak because only the intersections between the matched events are considered in Theorem~\ref{thm:main}.

Nevertheless, the primary motivation of this work is to explore the power of 
the MT algorithm in the general \lqc{variable} LLL setting. 
This basic problem is very important in itself, besides its potential applications.
\end{remark}

\subsubsection{Application to lattices}
To illustrate the application of Theorem~\ref{thm-dep-graph-beyond}, we estimate the  efficient region of the MT algorithm on some lattices explicitly.
For simplicity, we focus on symmetric probabilities, i.e., $\vec{p}=(p,p,\cdots,p)$. Our lower bounds on the gaps between the efficient region of the MT algorithm and the Shearer's bound
are summarized in Table~\ref{table1}. For example, when the canonical dependency graph is the square lattice, the vector $(0.1193,0.1193,\cdots)$ is on the Shearer's boundary, and the MT algorithm is provably efficient whenever the probability of each event is at most $0.1193+1.858\times 10^{-22}$.  %\warn{Our results performs well on constant-degree graph.}
%For simplicity, we focus on symmetric probabilities, where $\vec{p}=(p,p,\cdots,p)$. 
%Given a bipartite graph $G_B$ and a vector $\vec{p}$, we say $\vec{p}$ is on Shearer's boundary
%of $G_B$, if $(1 - \eps)\vec{p}$ is in Shearer's bound
%and $(1 + \eps)\vec{p}$ is not for any $\eps>0$. Theorem \ref{thm:gapbetweenlattices} is an application of Theorem~\ref{thm-MT-beyond-Shearer} to periodic Euclidean graphs. Based on  Theorem \ref{thm:gapbetweenlattices}, gaps between our criterion and Shearer's bound on some lattices are calculated explicitly and summarized in Table~\ref{table1}.

\vspace{-2ex}
\begin{table}[H]
	\caption{Summary of lower bounds on the gaps}\label{table1}
	\vspace{-2ex}
	\begin{center}
		\begin{tabular}{ ccc }
			\hline
			Lattice & Shearer's bound & lower bound on the gaps\\
			\hline  \hline
			%Triangular & $\frac{5\sqrt{5} - 11}{2}$~\cite{gaunt1967hard,baxter1980hard,SYNGE1999TRANSFER} & $5.650\times 10^{-20}$\\ 
			Square & 0.1193 ~\cite{gaunt1965hard,SYNGE1999TRANSFER} & $1.858\times 10^{-22}$\\
			Hexagonal & 0.1547 ~\cite{SYNGE1999TRANSFER}& $2.597\times 10^{-25}$\\
			Simple Cubic & 0.0744 ~\cite{gaunt1967hard} & $7.445\times 10^{-23}$\\
			\hline
		\end{tabular}
	\end{center}
\end{table}

\subsection{Technique overview}\label{sec:technique-overview}

As mentioned before, the Shearer's bound is the tight criterion for MT algorithm on extremal instances. Thus in order to show that MT algorithm goes beyond Shearer's bound, we need to take advantage of the intersection between dependent events. Specifically, Theorem \ref{thm-dep-graph-beyond} immediately follows from two results about non-extremal instances. One is the intersection LLL criterion (Theorem \ref{thm:main}), which goes beyond Shearer's bound whenever there are intersections between dependent events. The other result is a lower bound on the amount of intersection between dependent events for general instances (Theorem \ref{thm:lower-bound-intersection}).

\subsubsection{Proof overview of Theorem \ref{thm:main}}
 Let us first remember Kolipaka and Szegedy's argument \cite{kolipaka}, which shows that the MT algorithm is efficient up to the Shearer's bound. We assume that $\{A_i\}_{i=1}^m$ is a fixed set of events with dependency graph $G_D=([m],E_D)$ and probabilities $\vec{p}=(p_1,\cdots,p_m)$. The notion of a witness DAG\footnote{In the paper \cite{kolipaka}, the role of witness DAGs was played by ``stable set sequences", but the concepts are essentially the same: there is a natural bijection between stable set sequences and wdags.} (abbreviated wdag) %, see Definition \ref{def:wdag}) %and \textit{resampling tables} (see Section \ref{sec:31})  
 is central to their argument. A wdag is a DAG whose each node $v$ has a label $L(v)$ from $[m]$ and in which two nodes $v$ and $v'$ are connected by an arc if and only if $L(v)=L(v')$ or $(L(v),L(v'))\in E_D$.
 With a resampling sequence $\vec{s}=s_1,s_2,\cdots,s_T$ (i.e., MT algorithm picks the events $A_{s_1},A_{s_2},\cdots,A_{s_T}$ for resampling in this order), we associate a wdag $D_{\vec{s}}$ on node set $\{v_1,\cdots,v_T\}$ as follows: (a) $L(v_k)=s_k$ and (b) there is an arc from $v_k$  to $v_{\ell}$ with $k<\ell$ if and only if either $s_k=s_{\ell}$ or $(s_k,s_{\ell})\in E_D$ (see an example in Figure \ref{example_wdags}). We say that a wdag $D$ occurs in the resampling sequence $\vec{s}$ if there is subset $U$ of nodes in $D_{\vec{s}}$ such that $D$ is a subgraph of $D_{\vec{s}}$ induced by the nodes that have a directed path to $U$ (Figure \ref{example_wdags} (d) is an example, where $U=\{v_4\}$). An useful observation is that $\E[T]=\sum_{D\in\D}\Pr_{\vec{s}}[D\mbox{ occurs in }\vec{s}]$. Here, $\D$ denotes the set of all single-sink wdags (a.k.a. proper wdags) of $G_D$.    %Intuitively, a wdag is a DAG whose each nod ae is labelled with an event and each arc points from a node to a later node where these two nodes are dependent (see Definition \ref{def:wdag}); the resampling table records all randomness used by the algorithm (see Section \ref{sec-resampling-table}).
 
 We define the weight of a wdag $D$ to be $\Pi_{v\in D}p_{L(v)}$. 
 The crucial lemma in Kolipaka and Szegedy's argument (the idea is from Moser-Tardos analysis) is that the probability of occurrence of a certain wdag $D$ is upper bounded by its weight. %$\Pr[D\mbox{ occurs}]\leq\Pi_{v\in D}\Pr[A_{L(v)}]$. 
 The idea is that we can assume (only for
the analysis) that the MT algorithm has a preprocessing step where it prepares an infinite number of independent samples for each variable. These independent samples create a table $\vec{X}$, called \emph{the resampling table}  (see Figure \ref{figure:resample} in Section \ref{sec:31} for an example). When the MT algorithm decides to resample variable $X_j$, it picks a new sample of $X_j$ from the resampling table. Suppose a certain wdag $D$ occurs, then for each of its
events we can determine a particular set of samples in the resampling table that must satisfy the event, where we say that $D$ is consistent with the resampling table $\vec{X}$ and denote it by $D\sim \vec{X}$. Hence, $\Pr_{\vec{s}}[D\mbox{ occurs in }\vec{s}]\leq \Pr_{\vec{X}}[D\sim \vec{X}]=\Pi_{v\in D}p_{L(v)}$.
 
 Finally, they solved beautifully the summation of weights of proper wdags, i.e., $\sum_{D\in\D}\Pi_{v\in D}p_{L(v)}$, which turns out to converge if and only if $\vec{p}$ is in the Shearer's bound of $G_D$.

%, and $\vec{X}$ denote the resampling table. Kolipaka and Szegedy's argument consists of three steps.
%\begin{itemize}[leftmargin=16pt]
%\item First, they observed that the expected number of resample steps performed by MT algorithm is exactly the summation of probabilities of occurrence of each proper wdag. That is, $\E[T]=\sum_{D\in\D}\Pr[D\mbox{ occurs}]$.
%\item Then, they showed that a proper wdag $D$ should be consistent with the resampling table, denoted by $D\sim \vec{X}$, if $D$ occurs. So $\Pr[D\mbox{ occurs}]\leq \Pr[D\sim \vec{X}]=\Pi_{v\in D}\Pr[A_{L(v)}]$. 
%\item Finally, they solved beautifully the summation of weights of proper wdags, i.e., $\sum_{D\in\D}\Pi_{v\in D}\Pr[A_{L(v)}]$, which turns out to converge if and only if the probability vector $\vec{p}$ is in the Shearer's bound of $G_D$.
%\end{itemize}

\begin{figure}[t]
\centering
\includegraphics[scale=0.61]{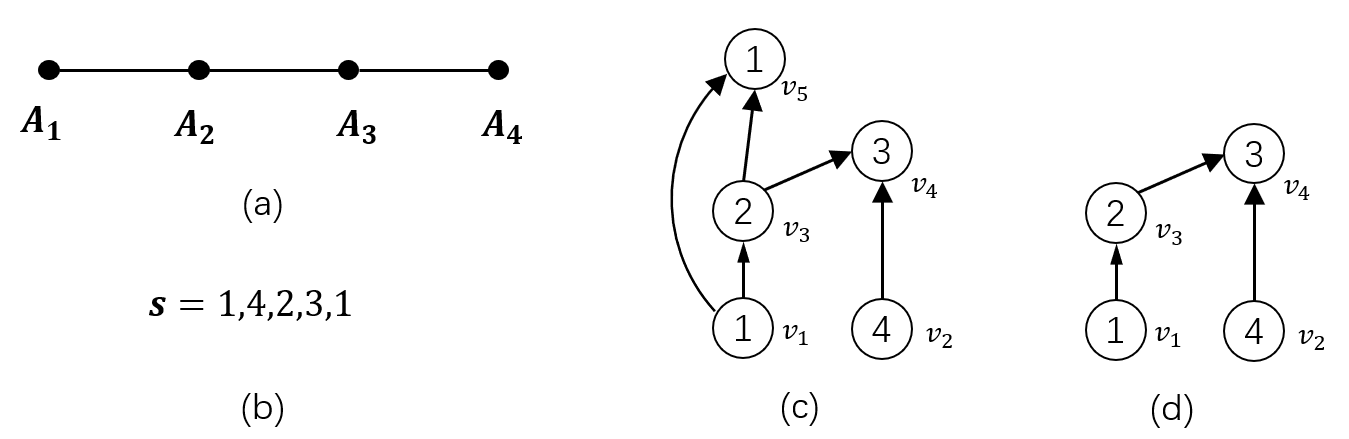}
\caption{(a) a dependency graph $G_D$; (b) a resample sequence; (c) the $D_s$; (d) a wdag occurring in $\vec{s}$.}
\label{example_wdags}
\end{figure}

Viewing Theorem \ref{thm:main} as an improvement of Kolipaka and Szegedy's result, we begin by providing a tighter upper bound on $\sum_{D\in\D}\Pr_{\vec{s}}[D\mbox{ occurs in }\vec{s}]$ when the instance is non-extremal (Theorem \ref{thm-ub-R}). %Since our argument should hold no matter what selection rule MT follows, 
First, note that for each wdag $D$, there exist selection rules to make $\Pr_{\vec{s}}[D \mbox{ occurs in }\vec{s}]=\Pi_{v\in D}p_{L(v)}$, so it is \emph{impossible} to give a better upper bound on $\Pr_{\vec{s}}[D \mbox{ occurs in }\vec{s}]$ which holds for all selection rules. Our idea is to group proper wdags, and consider the sum of $\Pr_{\vec{s}}[D \mbox{ occurs in }\vec{s}]$ over a group. %partiton $\D$ into groups  Instead of showing better upper bound on $\sum_{D\in\D}\Pr[D \mbox{ occurs}]$ %To illustrate why a tighter upper boundpossible, 
%To see the possibility of a tighter upper bound on the summation of $\Pr[D\mbox{ occurs}]$ over a group which holds for all selection rules, 
For example, suppose that $A_1$ and $A_2$ are dependent and $\Pr[A_1\cap A_2]\geq \delta_{1,2}$. Let $D_1$ denote the proper wdag which consists of only one arc $A_1\rightarrow A_2$, and $D_2$ denote the proper wdag consisting of only $A_2\rightarrow A_1$. $D_1$ and $D_2$ cannot both occur, but they may be both consistent with a given resampling table. So the total weights of $D_1$ and $D_2$ is an overestimate of the probability that $D_1$ or $D_2$ occurs. Formally,
\begin{align*}
\Pr_{\vec{s}}[D_1\mbox{ occurs in }\vec{s}]+\Pr_{\vec{s}}[D_2\mbox{ occurs in }\vec{s}]=&\Pr_{\vec{s}}[(D_1\mbox{ occurs in }\vec{s}) \vee (D_2 \mbox{ occurs in }\vec{s})]\\
\leq &\Pr_{\vec{X}}[(D_1\sim \vec{X})\vee(D_2\sim \vec{X})]\\
=&\Pr_{\vec{X}}[D_1\sim \vec{X}]+\Pr_{\vec{X}}[(D_2\sim \vec{X})\wedge(D_1\not\sim \vec{X})]\\
\leq & p_1 p_2+p_1 p_2-\delta_{1,2}^2,
\end{align*}
where the last inequality is according to the Cauchy–Schwarz inequality (see Proposition \ref{prop-prob}). Importantly, the upper bound holds for all selection rules.

It is crucial as well as the difficulty that our improvement over the weight of wdags should be ``exponential": since the quantity $\sum_{D\in\D} \Pi_{v\in D}p_{L(v)}^-$ converges if and only if $\vec{p}^-$ is in the Shearer's bound, constant factor or even sub-exponential improvements over $\sum_{D\in\D} \Pi_{v\in D}p_{L(v)}$ do not help to show the desired convergence criterion. Our exponential improvement relies on a delicate grouping and a tricky random partition of the union of $D\sim \vec{X}$ across wdags.

We first state how we group proper wdags: 
%Our exponential improvement relies on (a) a tricky group a tricky random partition of the union of $D\sim \vec{X}$ across wdags. %where the factor is a roughly exponential function of the number of occurrences of events in $\M$. 
%To achieve the exponential improvement,
%Our key idea  %key in our proof is  
%To obtain an ``exponential" improvement, the key of our idea the difficulty is  analyze the different wdags consistent with the resampling table, which is also the main technical part. Specifically, 
define $\D(i,r)$ to be the set of proper wdags whose unique sink node is labelled with $i$ and in which there are exactly $r$ nodes labelled with $i$. 
Noticing that at most one wdag in $\D(i,r)$ can occur, we have that
\begin{align*}
\sum_{D\in\D(i,r)}\Pr_{\vec{s}}[D\mbox{ occurs in }\vec{s}]=&\Pr_{\vec{X}}\left[\bigvee_{D\in\D(i,r)}(D\mbox{ occurs})\right]\leq \Pr_{\vec{X}}\left[\bigvee_{D\in\D(i,r)}(D\sim \vec{X})\right].%=\Pr[D_1\sim X]+\Pr[D_2\sim X \wedge(D_1\not\sim X)]+\cdots.
\end{align*}
%Notice that a constant factor improvement over the summation of weights of wdags do not help to show a better criterion for MT algorithm to converge. What we need is a somewhat exponential improvement, where the factor is a roughly exponential function of the number of occurrences of events in $\M$. 
%To achieve the exponential improvement, 

Now, we partition the space $\bigvee_{D\in\D(i,r)}(D\sim \vec{X})$ across wdags in $\D(i,r)$. The notions of \emph{reversible arcs} (see Definition \ref{def:reverse}) and a $\emph{auxiliary table}$ (see Section~\ref{sec:31}) are two central concepts here. Specifically, an arc $u\rightarrow v$ in a wdag $D$ is said reversible, if the directed graph obtained from $D$ by reversing the direction of $u\rightarrow v$ is also a wdag. % If $u\rightarrow v$ is a reversible  that $D$ and the resulted wdag cannot both occur, but  may can both be consistent with a given resampling table. %In our tighter upper bound (Theorem \ref{thm-ub-R}), each reversible arc gives a factor improvement. \\
%Moreover, we introduce another key notion, the \emph{auxiliary table} (see Section~\ref{sec-aux-table}) which is a table $\vec{Y}$ of independent fair coins indicating directions of reversible arcs. Alternatively, the auxiliary table induces a random ranking $\prec_{\vec{Y}}$ of wdags in $\D(i,r)$. Formally,
%Here, we introduce the notion of 
The auxiliary table is a table $\vec{Y}$ of independent fair coins corresponding to directions of reversible arcs. We say a wdag $D$ is consistent with $(\vec{X},\vec{Y})$, denoted by $D\sim (\vec{X},\vec{Y})$ if (i) $D\sim \vec{X}$; and (ii) for each reversible arc whose direction is \emph{not} consistent with $\vec{Y}$, the wdag obtained by reversing the arc is \emph{not} consistent with $\vec{X}$. The crucial lemma (Lemma \ref{compatible}) %verifies that a given auxiliary table induces a partition of the space $\bigvee_{D\in\D(i,r)}(D\sim \vec{X})$ across wdags in $\D(i,r)$. Specifically, Lemma \ref{compatible} 
shows that for any certain assignment $\vec{y}$ of the auxiliary table $\vec{Y}$, %the space
$\bigvee_{D\in\D(i,r)}(D\sim \vec{X})=\bigvee_{D\in\D(i,r)}(D\sim (\vec{X},\vec{y}))$. %can be covered by the union of $D\sim (\vec{X},\vec{y})$ over $D\in\D(i,r)$. Intuitively, $(D\sim \vec{X})$'s and $(D\sim (\vec{X},\vec{y}))$'s cover the same space (in which $\vec{X}$ lives), but 
The point is that $(D\sim (\vec{X},\vec{y}))$'s have much less overlap with each other so that they can be viewed as a ``approximate" partition of the space. By applying union bound, we get  %  the directions of reversible arcs in $D$ view The auxiliary table can be viewed as a random ranking $\prec_{\vec{Y}}$ of wdags in $\D(i,r)$. We say $D$ Then 
\begin{align*}
\Pr_{\vec{X}}\left[\bigvee_{D\in\D(i,r)}(D\sim \vec{X})\right]=&\E_{\vec{Y}}\Pr_{\vec{X}}\left[\bigvee_{D\in\D(i,r)}(D\sim\vec{X})\right]= \E_{\vec{Y}}\Pr_{\vec{X}}\left[\bigvee_{D\in\D(i,r)}(D\sim(\vec{X},\vec{Y})\right]\\
\leq&\E_{\vec{Y}}\sum_{D\in\D(i,r)}\Pr_{\vec{X}}\left[D\sim( \vec{X},\vec{Y})\right]\\
=& \sum_{D\in\D(i,r)}\E_{\vec{Y}}\Pr_{\vec{X}}\left[D\sim( \vec{X},\vec{Y})\right].
\end{align*}
Then we are able to provide an upper bound on $\E_{\vec{Y}}\Pr_{\vec{X}}\left[D\sim( \vec{X},\vec{Y})\right]$ which is ``exponentially" smaller than $\Pi_{v\in D}p_{L(v)}$ (Lemma \ref{lemma-prob-pwdag}), and then complete the proof of Theorem \ref{thm-ub-R}.

The next step is to show that the tighter upper bound converges when $\vec{p}^-$ is in the Shearer's bound. %Here, we propose a general framework to bound some quantity of proper wdags. 
For each vertex $i$ in the matching $\M$, we ``split"  vertex $i$ into two new connected vertices $i^\uparrow$ and $i^\downarrow$.
Let $G^{\M}$ be the resulted dependency graph (see an example in Figure \ref{example_homograph}).
Define $p^{\M}_{i^\uparrow} = p'_i$ and $p^{\M}_{i^\downarrow} = p^-_i - p'_i$ (see the definition of $p_i'$ in Section \ref{sec:23}).
One can see that $(G_D,\vec{p}^-)$ and $(G^\M,\vec{p}^\M)$ are essentially the same: suppose $\A\sim(G_D,\vec{p}^-)$, then for each $i\in \M$, we view $A_{i}$ as the union of two mutually exclusive events $A_{i^{\uparrow}}$ and $A_{i^{\downarrow}}$ whose probabilities are $p_i'$ and $p^--p_i'$ respectively. Such a representation of $\A$ is of the setting $(G^{\M},\vec{p}^{\M})$.
Thus, the sum of weights of proper wdags in the setting $(G_D,\vec{p}^-)$ is equal to that in the setting $(G^\M,\vec{p}^\M)$ (Proposition~\ref{obs-mapping-dag}). So it suffices to show that our tighter upper bound is upper bounded by the sum of weights of proper wdags in the setting $(G^\M,\vec{p}^\M)$ (Theorem \ref{thm-dag-mapping}). Our idea is to construct a mapping  which maps each $D\in\D(G_D)$ to a subset of $\D(G^\M)$ and satisfies that: 
\begin{itemize}[leftmargin=16pt]
	\item[(a)] distinct proper wdags of $G_D$ are mapped to disjoint subsets of $\D(G^\M)$; and 
	\item[(b)] for each $D\in\D(G_D)$, the bound in Lemma \ref{lemma-prob-pwdag} is upper bounded by the sum of weights of proper wdags over the subset that $D$ is mapped to.
\end{itemize} 
We present such a mapping in Definition \ref{def:h}. Conditions (a) and (b) are verified in Theorem \ref{thm-injection} and Theorem \ref{thm-dag-mapping} respectively. 

The idea of constructing a mapping between wdags of two dependency graphs may be of independent interest, and may be applied elsewhere when we wish to show some properties about Shearer's bound.

\subsubsection{Proof overview of Theorem \ref{thm:lower-bound-intersection}}

The proof of Theorem~\ref{thm:lower-bound-intersection} mainly consists of two parts. First, we show that there is an elementary event set which approximately achieves the minimum amount of the intersection between dependent events (Lemma \ref{lem:general2elementary}). Here, we call an event $A_i\in \A$ elementary, if there is a subset $S_j^i$ of the domain of variable $X_j$ for each variable in $\vbrl(A)$ such that $A$ happens if and only if 
$X_j \in S_j^i$ for all variables in $\vbrl(A)$.
We call a set $\A$ of events elementary if every $A_i\in\A$ is elementary. Then, for elementary event sets, by applying AM-GM inequality, we obtain a lower bound on the total amount of overlap on common variables, which further implies a lower bound on the amount of intersection between dependent events (Lemma \ref{lemma:elementary}).

\subsection{Related works}
Beck proposed the first constructive LLL, 
which provides efficient algorithms for finding the perfect object 
avoiding all ``bad" events~\cite{beck1991algorithmic}.
His methods were refined and improved by a long line of research~\cite{alon1991parallel,molloy1998further,czumaj2000coloring,haeupler2011new}. 
In a groundbreaking work, 
Moser and Tardos proposed a new algorithm, i.e., Algorithm~\ref{alg:mt},
and proved that it finds such a perfect object under the condition in Theorem~\ref{thm:asymmetric}
in the variable setting~\cite{moser2010constructive}.
Pegden \cite{pegden2014extension} proved that the MT algorithm efficiently converges even under the condition of the cluster expansion local lemma~\cite{bissacot2011improvement}.
Kolipaka and Szegedy ~\cite{kolipaka} pushed the efficient region to Shearer's bound. The phenomenon that the MT algorithm can still be efficient beyond Shearer's bound was known to exist \emph{for sporadic and toy examples} \cite{harris2016lopsidependency}.
However, such result employs the special structures in the examples and only applies to %shows that the MT algorithm is efficient for 
\textbf{some restricted} variable-generated event systems $\EventSet\sim (G_D,\vec{p})$.
By contrast, the results in this work applies to \textbf{all} variable-generated event systems.
%Several works (e.g. \cite{harris2016lopsidependency,catarata2017moser}) proved that MT algorithm is still efficent beyond Shearer's bound \emph{in some special cases}. %However, these results require additional assumptions and cannot apply to general instances. 
%For example, the criterion in \cite{harris2016lopsidependency} depends in a fundamental way on the decomposition of bad events into variables.}

Besides the line of research exploring the efficient region of the MT algorithm, there is a large amount of effort devoted to derandomizing or parallelizing the MT algorithm \cite{moser2010constructive,determ1,determ2,BrandtFHKLRSU16,Ghaffari16,ChungPS17,HaeuplerH17,Harris18} and to extending the Moser-Tardos techniques beyond the variable setting  \cite{HarrisS14a,AchlioptasI16,HarveyV20,AchlioptasIV17,AchlioptasIK19,Molloy19,IliopoulosS20,AchlioptasIS19}.

%Most convergence criterions of these algorithms are contained in or incomparable to Shearer's bound.
%In a remarkable work, Harris \cite{harris2016lopsidependency} proved a new criterion for the MT algorithm to converge which goes beyond Shearer's bound in some cases.
%We mention that Harris' criterion dependes on the notion
%of orderability and thus cannot be applied to general variable-generated systems.

There is a line of works studying the gap between non-constructive VLLL and Shearer's bound \cite{kolipaka,he2017variable,Andras2017variable,stoc19}. Kolipaka and Szegedy \cite{kolipaka} obtained the first example of gap existence where the canonical dependency graph is a cycle of length 4.  The paper \cite{he2017variable} showed that Shearer's bound is not tight for VLLL. More precisely, Shearer's bound is tight for non-constructive VLLL if and only if the canonical dependency graph is chordal. The first paper to study quantitatively the gaps systematically is \cite{stoc19}, which provides lower bounds on the gap when the canonical dependency graph containing many chordless cycles. %, including various lattices. %are obtained for  Recently, He et.al. \cite{??} obtained the first known lower bound on gaps for general graphs.

Erd{\"{o}}s and Spencer \cite{ErdosS91} introduced the lopsided-LLL, which extends the results in \cite{erdos1975problems} to lopsidependency graphs. 
Lopsided LLL has many interesting applications in combinatorics and theoretical computer science,
such as the $k$-SAT~\cite{gebauer2016local}, random permutations~\cite{lu2007using}, Hamiltonian cycles~\cite{albert1995multicoloured}, and matchings on the complete graph~\cite{lu2009new}.
Shearer's bound is also the tight condition for the lopsided LLL~\cite{shearer1985problem}.

LLL has a strong connection to sampling. 
Guo, Jerrum and Liu~\cite{guo2019uniform} proved that the MT algorithm indeed uniformly samples a perfect object if the instance is extremal. 
For extremal instances, they developed an algorithm called ``partial rejection sampling''
which resamples in a parallel fashion, 
since the occurring bad events form an
independent set in the dependency graph.
Actually, a series of sampling algorithms for specific problems
are the parallel resampling algorithm
running in the extremal case~\cite{guo2019uniform,guo2019polynomial,guo2020tight,guo2018approximately}.
In a celebrated work, Moitra~\cite{Moi19}
introduced a novel approach that utilizes LLL to sample $k$-CNF solutions.
This approach was then extended by several works ~\cite{guo2019counting,galanis2019counting,feng2020fast,feng2020sampling,Vishesh20towards,Vishesh21sampling}.

\subsection{Organization of the paper.} In Section \ref{sec:preliminaries}, we recall and introduce some definitions and notations. In Section \ref{sec:proof}, we prove Theorem~\ref{thm:main}. 
Section \ref{sec:lower-bound-intersection} is about the proof of Theorem~\ref{thm:lower-bound-intersection}, which gives a lower bound on the amount of the intersection between dependent events. 
In Section \ref{sec:beyond-shearer}, we prove Theorem \ref{thm-dep-graph-beyond}.  %and~\ref{thm:lower-bound-intersection}, 
%we show that Shearer's bound the efficient region of the MT algorithm goes beyond Shearer's bound.
%
In Section \ref{sec:lattice}, we provide a explicit lower bound for the gaps between the efficient region of MT algorithm and Shearer's bound on periodic Euclidean graphs.
%
%Section \ref{sec4} applies intersection LLL to provide a lower bound on the gap between VLLL and Shearer's bound.

\input{sec-preliminaries}
\input{sec2}
\input{sec3}

\input{sec4}
\input{sec5}

%\input{sec6}

%%%%%%%%%%%%%%%%%%%%%%%%%%%%%%%%%%%%%%%%%%%%%%%%%%%%%%%%%%%%%%%%%%%%%%%%%%%%%%%%%%%%%%%%%%%%%%%%%%%%%%%%%%%%%%%%%%%%%%%%%%%%%
\bibliographystyle{alpha}
\bibliography{MT_arxiv}
%%%%%%%%%%%%%%%%%%%%%%%%%%%%%%%%%%%%%%%%%%%%%%%%%%%%%%%%%%%%%%%%%%%%%%%%%%%%%%%%%%
\input{sec_appendix}

\end{document}

%% file: sec-preliminaries.tex
\section{Preliminaries}\label{sec:preliminaries}
%Given a dependency graph $G_D=([m],E_D)$, a matching $\M$ of $G_D$, and two probability vectors $\vec{p}\in (0,1)^m$ and $\vecdelta\in(0,1)^{|\M|}$,
%we say that a set $\A$ of events is of the setting $(G_D,\vec{p},\M,\vecdelta)$, denoted by $\A\sim(G_D,\vec{p},\M,\vecdelta)$, if $\A\sim(G_D,\vec{p})$
%and $\Pr(A_i\cap A_j)\geq \delta_{i,j}$ for each pair $(i,j)\in\M$.
	
Let $\mathbb{N}=\{0,1,2,\cdots\}$ denote the set of non-negative integers. Let $\mathbb{N}^+=\{1,2,\cdots\}$ denote the set positive integers. For $m\in \mathbb{N}^+$,  we define $[m]=\{1,\cdots,m\}$.
Throughout this section, we fix a canonical dependency graph $G_D=([m],E_D)$.

\subsection{Witness DAG}

If for a given run, MT algorithm picks the events 
	$A_{s_1}, A_{s_2} ,...,A_{s_T}$
	for resampling in this order, we say that $\vec{s} = s_1,s_2 ...,s_T$
	is \emph{a resample sequence}. If the algorithm never finishes, the
	resample sequence is infinite, and in this case we set $T = \infty$.

\begin{definition}[Witness DAG]\label{def:wdag}
	We define a witness DAG (abbreviated wdag) of $G_D$ to be a DAG $D$, in which each node $v$ has a label $L(v)$ from $[m]$, and which satisfies the additional condition that for all distinct nodes $v,v'\in D$ there is an arc between $v$ and
	$v'$ (in either direction) if and only if $L(v)=L(v')$ or $\big(L(v),L(v')\big)\in E_D$.
	
	We say $D$ is a \emph{proper} wdag (abbreviated pwdag) if $D$ has only one sink node. Let $\mathcal{D}(G_D)$ denote the set of pwdags of $G_D$.
\end{definition}

Given a resampling sequence $\vec{s} = s_1, s_2,...,s_T$, we associate a wdag $D_{\vec{s}}$ on the node set $\{v_1,...,v_T\}$ such that (i) $L(v_k) = s_k$ and (ii) $v_k \rightarrow v_\ell$ with $k<\ell$ is as an arc of $D_{\vec{s}}$ if and only if either $s_k = s_\ell$ or $(s_k,s_\ell) \in E_D$. See Figure \ref{example_wdags} for an example of $D_{\vec{s}}$.

Given a wdag $D$ and a set $U$ of nodes of $D$, we define $D(U)$ to be the induced subgraph on all nodes which has a directed
path to some $u\in U$. Note that $D(U)$ is also a wdag. We say that $H$ is a prefix of $D$, denoted by
$H\unlhd D$, if $H = D(U)$ for some node set $U$. 

\begin{definition}
	We say a wdag $D$ occurs in a resampling sequence $\vec{s}$ if $D\unlhd D_{\vec{s}}$.  Let $\chi_D$ be the indicator variable of the event that $D$ occurs in $\vec{s}$.
\end{definition}

Similar to Lemma 12 in \cite{kolipaka}, we have that $T = \sum_{D \in \mathcal{D}(G_D)} \chi_D$. For $i\in[m]$ and $r\in\mathbb{N}^+$, define $\D(i,r)$ to be the set of pwdags whose unique sink node is labelled with $i$ and in which there are exactly $r$ nodes labelled with $i$. Let $\chi_{\D(i,r)}$ be the indicator variable
of the event that there is a $D\in \D(i,r)$ occurring in $\vec{s}$. It is easy to see that only one pwdag in $\D(i,r)$ can occur in $\vec{s}$. Thus $\chi_{\D(i,r)}=\sum_{D\in \D(i,r)}\chi_D$, which further implies that 
\begin{fact}\label{lem:T-sum}
	$T = \sum_{i\in[m]}\sum_{r\in\mathbb{N}^+}\chi_{\D(i,r)}$.
\end{fact}

\subsection{Reversible arc}
In the rest of this section, we fix a matching $\M\subseteq E_D$ of $G_D$. Given $i\in [m]$, with a slight abuse of notation, we sometimes say $i\in \M$ if there is some $i'\in [m]$ such that $(i,i')\in \M$.
\begin{definition}[Reversibility]\label{def:reverse}
We say that an arc $u\rightarrow v$ is reversible in a wdag $D$ if the directed graph obtained from $D$ by reversing the direction of the arc is still a DAG.

Furthermore, we say that $u\rightarrow v$ is $\M$-reversible in $D$ if $u\rightarrow v$ is reversible in $D$ and $(L(u),L(v)) \in \M$.
\end{definition}
By definition, we have the following two observations. %We prove them in the appendix for completeness.
\begin{fact}\label{claim-exchangeable}
$u\rightarrow v$ is reversible in $D$ if and only if it is the unique path from $u$ to $v$ in $D$.
\end{fact}
\begin{fact}\label{claim-exchangeable2}
If $u\rightarrow v$ is reversible in a wdag $D$ of $G_D$, then the directed graph obtained from $D$ by reversing the direction of $u\rightarrow v$ is also a wdag of $G_D$.
\end{fact}

Given a pwdag $D = (V,E,L)$, define 
\[
\mathcal{V}(D) \triangleq \{v: \exists u\in V 
\text{ such that $u\rightarrow v$ or $v\rightarrow u$ is $\M$-reversible in $D$}\}
\]
to be the set of nodes participating in reversible arcs, and $\overline{\mathcal{V}} (D) \triangleq V \setminus \mathcal{V}(D)$. For $i\in[m]$, define $\mathcal{V}(D,i)\triangleq \mathcal{V}(D)\cap\{v:L(v)=i\}$.

\subsection{Other notations}\label{sec:23}
Let $\vec{p}=(p_1,\cdots,p_m)\in (0,1]^m$ and $\vecdelta\in(0,1)^{\M}$ be two probability vectors. Recall that $\vec{p}^-=(p_1^-,\cdots,p_m^-)$ is defined as
\begin{align}
\forall i\in [m]:\quad
p_{i}^-=\begin{cases}
p_i-\frac{\delta_{i,i'}^2}{17} & \text{if } (i,i')\in\M \text{ for some } i',\\
p_i & \text{otherwise}.
\end{cases}
\label{eq-pminus}
\end{align}
For each $i\in [m]$ where $(i,i')\in\M$ for some $i'\in [m]$, define
\[
c_i\triangleq \frac{\delta_{i,i'}^2}{8p_ip_{i'}}\quad \quad \quad \text{ and } \quad \quad  \quad p_i'\triangleq p_i(1-c_i)=p_i -\frac{\delta_{i,i'}^2}{8p_{i'}}.
\]
\begin{fact}\label{fact:210}
$p^{-}_{i} + p^{-}_{i'}(p^{-}_{i} - p'_{i})\geq p_i$ for each $(i,i')\in\M$.
\end{fact}
\newpage

%% file: sec2.tex
\section{Proof of Theorem \ref{thm:main}}\label{sec:proof}
The proof of Theorem \ref{thm:main} consists of two parts. First, we provide a tighter upper bound on the complexity of MT algorithm (Section \ref{sec:31}). Then, we show that the tighter upper bound converges if $\vec{p}^-$ is in the Shearer's bound of $G_D$ (Section \ref{sec-mapping}).

\subsection{A tighter upper bound on the complexity of MT algorithm}\label{sec:31}
In this subsection, we prove Theorem \ref{thm-ub-R}, which follows from Lemma \ref{compatible} and Lemma \ref{lemma-prob-pwdag} immediately. We first recall and introduce some concepts and notations.

%\subsection{Resampling Table}\label{sec-resampling-table}
\vspace{1ex}
\noindent\textit{Resampling Table}.  One key analytical technique of Moser and Tardos \cite{moser2010constructive} is to precompute the randomness in a resampling table $\vec{X}$.
%One the key techniques to analyze the MT algorithm~\cite{moser2010constructive}. 
Specifically, we can assume (only for
the analysis) that MT algorithm has a preprocessing step where
it draws an infinite number of independent samples $X_{j}^{1}, X_{j}^{2},\cdots$ for each variable $X_j$. These independent samples create a table $\vec{X}=(X_{j}^k)_{j\in[m],k\in\mathbb{N}^+}$, called the resampling table 
(see Figure \ref{figure:resample}). 
%The rest of MT algorithm can be interpreted as the following
% process: 
MT algorithm takes that first column as the initial assignments of $X_1,\cdots,X_n$. Then, when $X_j$ is to be resampled, MT algorithm goes right in the row corresponding to $X_j$ and picks the sample. %We emphasize that we only use this idea to analyze the algorithm rather than to really create the table in the execution. 

\vspace{1ex}
\noindent\textit{Consistency with the resampling table}. 
%There is a critical connection between the resampling table and wdags.
 For a wdag $D$, a node $v$, and a variable $X_j\in \vbrl(A_{L(v)})$, we define
\[
\Location(D,v,j) \triangleq |\{u: \text{ there is a directed path from $u$ to $v$ in $D$ and } X_j \in \vbl(A_{L(u)})\}| + 1.
\]
Moreover, let $\vec{X}_{D,v} \triangleq \{X_{j}^{\Location(D,v,j)}:X_j\in A_{L(v)}\}$. We say that $D$ is \emph{consistent} with $\vec{X}$, denoted by $D\sim \vec{X}$, if for each node $v$ in $D$, the event $A_{L(v)}$ holds on $\vec{X}_{D,v}$. 
Intuitively,
suppose $D$ occurs, then $\vec{X}_{D,v}$ are the assignments of $\vbl(A_{L(v)})$ just before the time that the MT algorithm picks the event corresponding to $v$ to resample, hence $A_{L(v)}$ must hold on $\vec{X}_{D,v}$. %On the other hand, if $D \sim \vec{X}$, then there is a selection strategy for MT algorithm to make $D$ occur.
We sometimes use $\Location(v,j)$ and $\vec{X}_v$ instead of $\Location(D,v,j)$ and $\vec{X}_{D,v}$ respectively if $D$ is clear from the context. 
Besides, we use $\D(i,r)\sim \vec{X}$ to denote that there is some $D\in \D(i,r)$ such that $D\sim \vec{X}$.
%\begin{fact}
%	Suppose $u\rightarrow v$ is reversible in $D$. For any $X_j\in\vbl\big(A_{L(u)}\big)\cap \vbl\big(A_{L(u)}\big)$, we have $\Location(u,j)+1=\Location(v,j)$.
%\end{fact}

\begin{figure}[!ht]
  \centering
%      \begin{subfigure}[b]{0.3\textwidth}
        \includegraphics[scale=0.8]{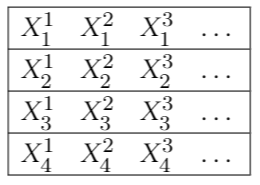}
 %       \caption{A resampling table}
%      \end{subfigure}
        \hspace{1in}
%       \begin{subfigure}%[b]{0.3\textwidth}
        \includegraphics[scale=0.9]{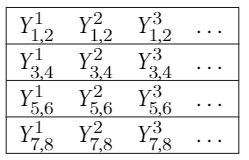}
%        \caption{An auxiliary table where $|\M|=3$}
%      \end{subfigure}
 \caption{The left is a resampling table where there are four variables $X_1,\cdots,X_4$. The right is an auxiliary table where $\M=\{(1,2),(3,4),(5,6),(7,8)\}.$}
 \label{figure:resample}
\end{figure}

\vspace{1ex}
\noindent\textit{Auxiliary Table}. We introduce
another central concept in the proof of Theorem \ref{thm-ub-R}, called the auxiliary table, which is a table of independent fair coins.  Specifically, for each pair $(i,i')\in \M$, we draw an infinite number of independent fair coins $Y_{i,i'}^{1}, Y_{i,i'}^{2},\cdots$, where $\Pr(Y_{i,i'}^k=i)=\Pr(Y_{i,i'}^k=i')=1/2$. These independent coins form the auxiliary table $\vec{Y}=(Y_{i,i'}^k)_{(i,i')\in\M,k\in\mathbb{N}^+}$ (see Figure \ref{figure:resample}). 
%We emphasize that the auxiliary table is not really used in the MT algorithm.
%It is just for analysis.
The auxiliary table is used to encode directions of $\M$-reversible arcs, according to which  %in the wdags when there are multiple choices from the dependent events.
we partition the space $\bigvee_{D\in\D(i,r)}(D\sim \vec{X})$.

\vspace{1ex}
\noindent\textit{Consistency with the resampling table and the auxiliary table}. We need some notations about reversible arcs. 
 Suppose $D$ has a unique sink node $w$ and $u\rightarrow v$ is reversible in $D$. Let $D'$ be the DAG obtained from $D$ by reversing the direction of $u\rightarrow v$. We define $\varphi(D,u,v)\triangleq D'(\{w\})$.
 In other words, $\varphi(D,u,v)$ is the prefix of $D'$ with a unique sink node $w$. 
Given $(i,i')\in\M$ and a pwdag $D$, let  $\mathrm{List}(D,i,i')$ denote the sequence listing all nodes in $D$ with labels $i$ or $i'$ in a topological order of $G_D$\footnote{It is easy to see that $\mathrm{List}(D,i,i')$ is well defined. That is, all topological orderings of $D$ induce the same $\mathrm{List}(D,i,i')$.}. 
%Note that for any two nodes with labels $i$ or $i'$, there is arc between them. Therefore $\mathrm{List}(D,i,i')$ is well-defined.
Given a node $v$ in $D$, if $(L(v),i)\in \M$\footnote{Because $\M$ is a matching, there is at most one such $i$.}, we define
\[
\lambda(v,D)\triangleq |\{u:(u\rightarrow v \text{ is in } D) \land (L(u)\in \{i,L(v)\})\}|+1
\]
to be the order of $v$ in $\mathrm{List}(D,L(v),i)$.
For simplicity of notations, we will use $\lambda(v)$ instead of $\lambda(v,D)$ if $D$ is clear from the context.

%Formally, we have the following definition.

%\begin{definition}[Consistency]\label{def-cons}
	Given a wdag $D$, we say an $\M$-reversible arc $u\rightarrow v$ is {inconsistent} with the auxiliary table $\vec{Y}$ if
	 $y_{L(u),L(v)}^{\lambda(u)} = L(v)$.	
	We say $D$ is consistent with $(\vec{X},\vec{Y})$, denoted by $D\sim(\vec{X},\vec{Y})$, if (i) $D\sim \vec{X}$ and (ii) for any $\M$-reversible arc $u\rightarrow v$ inconsistent with $\vec{y}$, $\varphi(D,u,v)\not\sim \vec{X}$.  
	We say $\D(i,r)\sim (\vec{X},\vec{Y})$ if there is some $D\in\D(i,r)$ such that $D\sim (\vec{X},\vec{Y})$.
%\end{definition}

The intuition of the notion ``consistency" is as follows. Suppose  $u\rightarrow v$ in a $\M$-reversible arc in $D$, and both $D$ and $\varphi(D,u,v)$ are consistent with the resampling table. But $D$ and $\varphi(D,u,v)$ cannot both occur. It is according to the auxiliary table to which one of $D$ and $\varphi(D,u,v)$ we assign $(D\sim \vec{X})\wedge (\varphi(D,u,v)\sim \vec{X})$.

%Given an assignment $\vec{x}$ of the resampling table, it may be interpreted as different wdags given different event-selection strategies.
%If $\vec{x}\sim D$ and $\vec{x}\sim \varphi(D,u,v)$ for some $\M$-reversible arc $u\rightarrow v$ in $D$, we employ the assignment $\vec{y}$ of the auxiliary table
%to determine that which of $\{u\rightarrow v, v\rightarrow u\}$ is
%the ``consistent direction" of the reversible arc 
%and which of $\{D,\varphi(D,u,v)\}$ is 
%the ``consistent wdag". 
%One can verify that exactly one of $\{D,\varphi(D,u,v)\}$ is consistent with $(\vec{x},\vec{y})$.

%\begin{fact}\label{fact:26}
%	Suppose $u\rightarrow v$ is $\M$-reversible in $D$. Then $\lambda(u,D)+1=\lambda(v,D)$, i.e., $v$ is the next node of $v$ in the sequence $\mathrm{List}(D,L(u),L(v))$.
	
%	Let $D'$ be the wdag obtained from $D$ by reversing the direction of $u\rightarrow v$. Then $\lambda(u,D')=\lambda(v,D)=\lambda(u,D)+1$ and $\lambda(v,D')=\lambda(u,D)=\lambda(v,D)-1$.
%\end{fact}
\vspace{2ex}
\begin{lemma}\label{compatible}
	For each $i\in [m]$ and $r\in \mathbb{N}^+$, $\Pr_{\vec{X}}[ \D(i,r)\sim\vec{X}]
	= \Pr_{\vec{X},\vec{Y}}[\D(i,r)\sim (\vec{X},\vec{Y})]$.
\end{lemma}
\begin{proof}
	Fix an arbitrary assignment $\vec{x}$ of $\vec{X}$ and an arbitrary assignment $\vec{y}$ of $\vec{Y}$. 
	Suppose $\D(i,r)\sim\vec{x}$, i.e, $\exists D_0\in \D(i,r)$ such that $D\sim \vec{x}$. We will show that there must exist some $D\in\D(i,r)$ such that $D\sim(\vec{x},\vec{y})$. This will imply the conclusion immediately.
	
	We apply the following procedure to find such a pwdag $D\in \D(i,r)$.
%\vspace{-2.5ex}
\begin{algorithm}[h]
	Initially, $k=0$\;
	\While{$\exists$ an $\M$-reversible arc $u_k\rightarrow v_k$ in $D_k$ inconsistent with $\vec{y}$ such that $\varphi(D_k,u_k,v_k)\sim \vec{x}$}{
		let $D_{k+1}:=\varphi(D_k,u_k,v_k)$ and $k:=k+1$\;
	}
	Return $D_k$;
\end{algorithm}
%\vspace{-2.5ex}

By induction on $k$, it is easy to check that $D_k\sim \vec{x}$ and $D_k\in \D(i,r)$ for each $k$. Furthermore, if the procedure terminates, then in the final wdag $D$, for every $\M$-reversible arc $u\rightarrow v$ inconsistent with $\vec{y}$, we have that $\varphi(D,u,v)\not\sim \vec{x}$. So  $D\sim(\vec{x},\vec{y})$. In the following, we will show that the procedure always terminates, which finishes the proof.

Note that each $D_k$ has no more nodes than $D_0$ and that there are finite number of wdags in $\D(i,r)$ with no more nodes than $D_0$, so it suffices to prove that each wdag appears at most once in the procedure.

By contradiction, assume $D_{j}=D_{k}$ for some $j\leq k$. Recall that $u_{j}\rightarrow v_j$ is reversible in $D_j$ and inconsistent with $\vec{y}$. So $y_{L(u_j),L(v_j)}^{\lambda(v_j,D_j)-1}=y_{L(u_j),L(v_j)}^{\lambda(u_j,D_j)}=L(v_j)$. 

	Let $D_{\ell}$ be the last wdag in $D_{j+1},\cdots,D_k$
		such that 
		$\lambda(v_j,D_\ell) < \lambda(v_j,D_j)$.
		%According to Fact \ref{fact:26}, 
		Observing that $\lambda(v_j,D_{j+1}) = \lambda(v_j,D_{j}) -1 $,
		we have such $D_{\ell}$ must exist.
		By $\lambda(v_j,D_k) = \lambda(v_j,D_j)$,
		we have $\lambda(v_j,D_\ell) = \lambda(v_j,D_{j}) -1$, 
		$\lambda(v_j,D_{\ell + 1}) = \lambda(v_j,D_{j})$.
		Therefore, 
		$\lambda(v_j,D_{\ell + 1}) = \lambda(v_j,D_{\ell}) + 1$.
		Combining with that
		$u_{\ell} \rightarrow v_{\ell}$
		is the inconsistent arc in $D_{\ell}$ which is reversed in $D_{\ell + 1}$, 
		we have $u_{\ell} = v_{j}$, 
		$(L(u_j),L(v_j)) = (L(u_\ell),L(v_\ell))\in \M$
		and $y_{L(u_{\ell}),L(v_{\ell})}^{\lambda(u_\ell,D_\ell)} = L(v_{\ell})$.
		Thus we have $L(v_{\ell}) = L(u_j)$ and 
		$y_{L(u_{\ell}),L(v_{\ell})}^{\lambda(u_\ell,D_\ell)} = L(u_j)$.
		Note that $\lambda(v_\ell,D_\ell) = 1 + \lambda(u_\ell,D_\ell) = 1 + \lambda(v_j,D_\ell)$.
		Combining with $\lambda(u_j,D_j) = \lambda(v_j,D_{j}) -1$,
		we have 
		$\lambda(u_\ell,D_\ell) = \lambda(u_j,D_{j})$.
		Combining with $y_{L(u_{\ell}),L(v_{\ell})}^{\lambda(u_\ell,D_\ell)} = L(u_j)$,
		we have 
		$y_{L(u_{\ell}),L(v_{\ell})}^{\lambda(u_j,D_{j})} = L(u_j)$. This is contradicted with $y_{L(u_j),L(v_j)}^{\lambda(u_j,D_j)}=L(v_j)$.

\end{proof}
\vspace{2ex}
%By Fact \ref{fact:26}, it is not hard to see the following proposition. For the sake of completeness, we present the proof in the appendix.
The following two propositions will be used in the proof of Lemma \ref{lemma-prob-pwdag}. The first proposition is an easy observation, and the second one is a direct application of the Cauchy-Schwarz inequality. For the sake of completeness, we present their proof in the appendix.
\begin{proposition}\label{prop:27}
	Given any wdag $D$, there exists a set $\mathcal{P}$ of disjoint $\M$-reversible arcs\footnote{We say two arc $u\rightarrow v$ and $u'\rightarrow v'$ are disjoint if their node sets are disjoint, i.e. $\{u,v\}\cap\{u',v'\}=\emptyset$.} such that: for each $i\in \M$,
	\[
	|\{v:\exists u \text{ such that }u\rightarrow v \text{ or }v\rightarrow u \text{ is in }\mathcal{P}\}\cap \{v:L(v)=i\}|\geq \frac{1}{2}\cdot \mathcal{V}(D,i).
	\]
\end{proposition}

%The following proposition will be used in the proof of Lemma \ref{lemma-prob-pwdag}.
\begin{proposition}\label{prop-prob}
Suppose $X,Y$ and $Z$ are three independent random variables, $A$ is an event determined by $\{X,Y\}$, and $A'$ is an event determined by $\{Y,Z\}$. Let $X_1,Y_1,Y_2,Z_1$ be four independent samples of $X,Y,Y,Z$, respectively. 
Then the following holds with probability at most $\Pr(A)\Pr(A')-\Pr(A\cap A')^2$:
	\begin{itemize}
		\item $A$ is true on $(X_1,Y_1)$, $A'$ is true on $(Y_2,Z_1)$, and
		\item either $A$ is false on $(X_1,Y_2)$ or $A'$ is false on $(Y_1,Z_1)$.
	\end{itemize}
\end{proposition}
Now, we are ready to show Lemma \ref{lemma-prob-pwdag}.
\begin{lemma}\label{lemma-prob-pwdag}
	 For each pwdag $D$, 
\[
\Pr[D\sim (\vec{X},\vec{Y})] \leq \left(\prod_{v\in \overline{\mathcal{V}} (D)}p_{L(v)}\right)\left(\prod_{v\in \mathcal{V}(D)}p'_{L(v)}\right).
\]
\end{lemma}
\begin{proof}
Let $\mathcal{P}$ be the set of disjoint $\M$-reversible arcs defined in Proposition \ref{prop:27}. Let $V(\mathcal{P})$ denote the set of nodes which appears in $\mathcal{P}$, and $\overline{V(\mathcal{P})}$ consists of the other nodes. Proposition \ref{prop:27} says that for each $i\in \M$,
\[
V(\mathcal{P})\cap \{v:L(v)=i\}|\geq \frac{1}{2}\cdot \mathcal{V}(D,i).
\]

For each $v\in \overline{V(\mathcal{P})}$, let $B_v$ denote the event that $A_{L(v)}$ holds on $\vec{X}_{v}$. It is easy to see that $\Pr[B_v]=p_{L(v)}$. Besides,
\begin{claim}\label{claim:321}
If $D\sim(\vec{X},\vec{Y})$, then $B_v$ holds for each $v\in\overline{V(\mathcal{P})}$. 
\end{claim}
\begin{proof}
Note that $\vec{X}_{v}$ are the assignments of $\vbl(A_{L(v)})$ just before the time that the MT algorithm picks the event corresponding to $v$ to resample. MT algorithm decides to pick $A_{L(v)}$ only if $A_{L(v)}$ holds. Hence $A_{L(v)}$ must hold on $\vec{X}_{v}$.
\end{proof}

Let $u\rightarrow v$ be an arc in $\mathcal{P}$, where $L(u)=i$ and $L(v)=i'$. 
Then by the definition of $\mathcal{P}$, we have $u\rightarrow v$ is reversible in $D$.
Let $D'$ be the wdag obtained by reversing the direction of $u\rightarrow v$ in $D$.
Recalling the definition of $\vec{X}_{D',v}$,
one can verify that
	$$\vec{X}_{D',u} :=  \left\{X_{j,\Location(v,j)}: X_j \in \vbl\big(A_{i}\big)\cap \vbl\big(A_{i'}\big)\right\}\cup \left\{X_{j,\Location(u,j)}: X_j \in \vbl\big(A_{i}\big)\setminus \vbl\big(A_{i'}\big)\right\}$$
	and 
	$$\vec{X}_{D',v} :=  \left\{X_{j,\Location(u,j)}: X_j \in \vbl\big(A_{i}\big)\cap \vbl\big(A_{i'}\big)\right\}\cup \left\{X_{j,\Location(v,j)}: X_j \in \vbl\big(A_{i'}\big)\setminus \vbl\big(A_{i}\big)\right\}.$$
	For simplicity, let $\lambda:=\lambda(u,D)$. We define $B_{u,v}$ to be the event that the following hold:
	\begin{itemize}
	\item[(a)] $A_{i}$ holds on $X_{u}$, and  $A_{i'}$ holds on $X_{v}$;
	\item[(b)] If $Y_{i,i'}^{\lambda}=i'$, then either $A_i$ is false on $\vec{X}_{D',u}$  or $A_i'$ is false on $\vec{X}_{D',v}$.
	\end{itemize}
Conditioned on that $Y_{i,i'}^{\lambda}=i$, $B_{u,v}$ happens with 
probability $p_{i}p_{i'}$.
Condition on that $Y_{i,i'}^{\lambda}=i'$, by using Proposition~\ref{prop-prob}, $B_{u,v}$ happens with probability at most $p_{i}p_{i'} - \delta^2_{i,i'}$.
Thus,
\begin{align*}\label{eq-buv}
\Pr[B_{u,v}] &\leq 
\Pr[Y_{i,i'}^{\lambda}=i]p_{i}p_{i'} + \Pr		[Y_{i,i'}^{\lambda}=i']\left(p_{i}p_{i'}- \delta^2_{i,i'}\right)
 \leq \frac{1}{2}\cdot p_{i}p_{i'} + \frac{1}{2}\cdot\left(p_{i}p_{i'}- \delta^2_{i,i'}\right)\\
 &\leq p_ip_{i'}(1-2c_i)(1-2c_{i'}).
\end{align*}

\begin{claim}
	If $D\sim(\vec{X},\vec{Y})$, then $B_{u,v}$ holds for each $u\rightarrow v$ in $\mathcal{P}$.
\end{claim}
\begin{proof}
Suppose $D\sim(\vec{X},\vec{Y})$. Similar to the argument in Claim \ref{claim:321}, we can see that Item (a) holds. In the following, we show Item (b) holds.

By contradiction, assume $Y_{i,i'}^{\lambda}=i'$, $A_{i}$ holds on $\vec{X}_{D',u}$, and $A_{i'}$ holds on $\vec{X}_{D',v}$. Then, we have $u\rightarrow v$ in $D$ is inconsistent with $\vec{Y}$ and $D'\sim\vec{X}$.
Thus, $\varphi(D,u,v)\sim\vec{X}$ since $\varphi(D,u,v)$ is a prefix of $D'$. By definition, we have $D\not\sim(\vec{X},\vec{Y})$, a contradiction.
\end{proof}

Since the events $\{B_v:v\in \overline{V(\mathcal{P})}\}$ and $\{B_{u,v}:u\rightarrow v \text{ is in } \mathcal{P}\}$ depend on distinct entries of $\vec{X}$ and $\vec{Y}$, they are mutually independent. 
Therefore,
\begin{align*}
\Pr\left[D\sim(\vec{X},\vec{Y})\right] &\leq \Pr
\left[\left(\bigcap_{w\in \overline{V(\mathcal{P}})}B_w\right)\bigcap \left(\bigcap_{u\rightarrow v \text{ is in }\mathcal{P}}B_{u,v}\right)\right]
=\left(\prod_{w\in \overline{V(\mathcal{P})}}\Pr(B_w)\right)\left(\prod_{u\rightarrow v\text{ is in }\mathcal{P}}\Pr(B_{u,v})\right)\\
&\leq\left(\prod_{w\in \overline{V(\mathcal{P})}}p_{L(w)}\right)\left(\prod_{u\rightarrow v\text{ is in }\mathcal{P}}p_{L(u)}p_{L(v)}\left(1-2c_{L(u)}\right)\left(1-2c_{L(u)}\right)\right)\\
&=\left(\prod_{v\text{ in } D}p_{L(v)}\right)\cdot \left(\prod_{i\in[m]}\left(1-2c_i\right)^{|\mathcal{P}\cap \{v:L(v)=i\}|}\right)\\
&\leq\left(\prod_{v\text{ in } D}p_{L(v)}\right)\cdot \left(\prod_{i\in[m]}\left(1-2c_i\right)^{|\mathcal{V}(i)|/2}\right)\leq\left(\prod_{v\text{ in } D}p_{L(v)}\right)\cdot \left(\prod_{i\in[m]}\left(1-c_i\right)^{|\mathcal{V}(i)|}\right)\\
&=\left(\prod_{v\in \overline{\mathcal{V}} (D)}p_{L(v)}\right)\left(\prod_{v\in \mathcal{V}(D)}p'_{L(v)}\right).
\end{align*}
\end{proof}

Now we are ready to prove the main theorem of this subsection.
\begin{theorem}\label{thm-ub-R}
	$\E[T] \leq \sum_{D\in \mathcal{D}(G_D)}\big(\prod_{v \in \overline{\mathcal{V}} (D)}p_{L(v)}\big)\big(\prod_{v\in \mathcal{V}(D)}p'_{L(v)}\big)$.
\end{theorem}
\begin{proof}
First, according to Lemmas~\ref{compatible} and ~\ref{lemma-prob-pwdag},
\begin{align*}
\Pr[\chi_{\D(i,r)}]&\leq \Pr[\D(i,r)\sim\vec{X}]= \Pr[\D(i,r)\sim(\vec{X},\vec{Y})]\leq \sum_{D\in \D(i,r)}\Pr[D\sim(\vec{X},\vec{Y})]\\
&\leq \sum_{D\in \D(i,r)}\left(\prod_{v \in \overline{\mathcal{V}} (D)}p_{L(v)}\right)\left(\prod_{v\in \mathcal{V}(D)}p'_{L(v)}\right).
\end{align*}
Then, by Fact~\ref{lem:T-sum} and the above inequality, we have
	\begin{equation*}
		\begin{aligned}
	\E[T] &= \sum_{i\in [m]}\sum_{r\in \mathbb{N}^+}\Pr[\chi_{\D(i,r)}] \leq \sum_{i\in [m]}\sum_{r\in \mathbb{N}^+} \sum_{D\in \D(i,r)}\left(\prod_{v \in \overline{\mathcal{V}} (D)}p_{L(v)}\right)\left(\prod_{v\in \mathcal{V}(D)}p'_{L(v)}\right)
	\\& \leq \sum_{D\in \D(G_D)} \left(\prod_{v \in \overline{\mathcal{V}} (D)}p_{L(v)}\right)\left(\prod_{v\in \mathcal{V}(D)}p'_{L(v)}\right).
		\end{aligned}
    \end{equation*}
	
\end{proof}

\subsection{Mapping between wdags}\label{sec-mapping}

In this section, we will prove Theorem~\ref{thm-dag-mapping}, which provides a upper bound of $\E[T]$ in terms of $\vec{p}^-$. 
\begin{definition}[Homomorphic dependency graph]
	Given a dependency graph $G_D=([m],E_D)$ and a matching $\M$ of $G_D$, we define a  graph $G^{\M}=(V^{\M},E^{\M})$ homomorphic to $G_D$ respected to $\M$ as follows. 
	\begin{itemize}
		\item $V^{\M} = [m]\setminus\{i_0,i_1:(i_0,i_1)\in \M\}\cup\{i_0^{\uparrow},i_0^{\downarrow},i_1^{\uparrow},i_1^{\downarrow}:(i_0,i_1)\in \M\}$;
		\item $\forall (i_0,i_1)\in E_D$, each pair of vertices in $\{i_0,i_1,i_0^{\uparrow},i_0^{\downarrow},i_1^{\uparrow},i_1^{\downarrow}\}\cap V^{\M}$ are connected in $G^{\M}$.
	\end{itemize}
Besides, we associate a probability vector $\vec{p}^{\M}$ with $G^{\M}$ as follows:
\begin{equation*}
	\begin{aligned}
		\forall v\in V^{\M}:\quad
		p^{\M}_{v} =\begin{cases}
			p_i' & \text{if } v = i^\uparrow \text{ for some } i\in[m],\\
			p_i^{-} - p_i' & \text{if } v = i^\downarrow \text{ for some } i\in[m],\\
			p^-_i & \text{otherwise, } v = i \text{ for some } i\in[m].
		\end{cases}
	\end{aligned}
\end{equation*}
\end{definition}
\begin{figure}[t]
\centering
\includegraphics[scale=0.7]{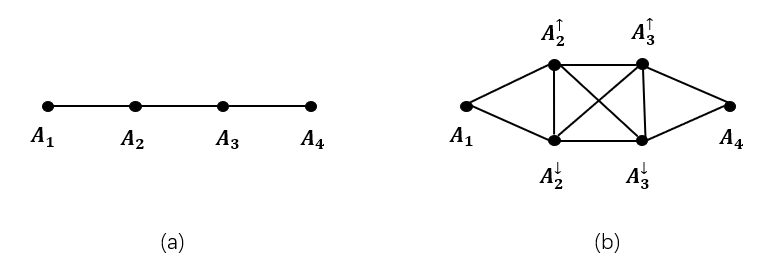}
\caption{(a) a dependency graph $G_D$; (b) the $G^{\M}$ when $\M=\{(2,3)\}$.}
\label{example_homograph}
\end{figure}
In fact, $(G_D,\vec{p}^-)$ and $(G^\M,\vec{p}^\M)$ are essentially the same: suppose $\A\sim(G_D,\vec{p}^-)$, then for each $i\in \M$, we view $A_{i}$ as the union of two mutually exclusive events $A_{i^{\uparrow}}\cup A_{i^{\downarrow}}$ whose probabilities are $p_i'$ and $p^--p_i'$ respectively. Such a representation of $\A$ is of the setting $(G^{\M},\vec{p}^{\M})$.

We have the following proposition, whose proof can be found in the appendix.
\begin{proposition}\label{obs-mapping-dag}
$		\sum_{D'\in \D(G^{\M})}\prod_{v' \text{ in } D'}p^{\M}_{L'(v')} = \sum_{D\in \D(G_D)}\prod_{v \text{ in } D}p^{-}_{L(v)}.$
\end{proposition}

Given a pwdag $D= (V,E,L)$,
recall that
$\mathcal{V}(D)$ is the set of nodes of $\M$-reversible arcs in $D$.  
Define $\Mat(D)\triangleq \{v: L(v) \in \M\}$ to be the set of nodes $v$ in $D$ where $L(v)$ is contained in an edge in $\M$. Obviously, $\mathcal{V}(D)\subseteq\Mat(D)$.
For simplicity of notations, we will omit $D$ from the notations if $D$ is clear from the context.

Given a pwdag $D=(V,E,L)$, we use $\Rem=\{\Rem_{1},\Rem_{2},\Rem_{3},\Rem_{4}\}$ to represent a partition of $\Mat(D)$ where $\mathcal{V} \subseteq  \Rem_{1}$ (some of these four sets are possibly empty). Let $\psi(D)$ denote the set consisting of all such partitions. 
The formal definition is as follows.

\begin{definition}[Partition] Given a pwdag $D= (V,E,L)$ of $G_D$,
define 
$$\psi(D) \triangleq \{ \{\Rem_{1},\Rem_{2},\Rem_{3},\Rem_{4}\}:\mathcal{V} \subseteq  \Rem_{1} \text{ and } \Mat=\Rem_{1}\sqcup\Rem_{2}\sqcup \Rem_{3}\sqcup\Rem_{4}\}. $$
\end{definition}

Given a wdag $D$, there may be two or more topological ordering of $D$. We fix an arbitrary topological ordering, and denote it by $\pi_D$.
 In the following, we define an injection $h$ from $\{(D,\Rem):D\in \mathcal{D}(G_D),\Rem\in \psi(D)\} $ to $\mathcal{D}(G^\M)$.

\begin{definition}\label{def:h}
Given a pwdag $D$ and $\Rem\in \psi(D)$, define $h(D,\Rem)$ to be a directed graph $D'=(V',E',L')$ constructed as follows.	
\vspace{1ex}

\noindent\underline{Constructing $V'$.} $V'=V_1'\sqcup V_2'$ where $|V_1'|=|V|$ and $|V_2'|=|\Rem_3\cup\Rem_4|$. For convenience of presentation, we fix two bijections $f:V\rightarrow V_1'$ and $f^{\ast}:\Rem_3\cup\Rem_4\rightarrow V_2'$ to name nodes in $V'$. In order to distinguish between nodes in $D$ and those in $D'$, we will always use $u,v,w$ to represent the nodes of $D$ and $u',v',w'$ to present the nodes of $D'$. Given $v'\in V'$, we use $g(v')$ to denote the unique node $v\in V$ such that $f(v)=v'$ (if $v'\in V_1'$) or $f^{\ast}(v)=v'$ (if $v'\in V_2'$).

\vspace{1ex}

\noindent\underline{Description of $L'$.}  For each node $v'\in V_1'$, where $v'=f(v)$,
\begin{equation}
	\begin{aligned}
		\label{eq-f_v}
		L'(v') =
		\begin{cases}
			(L(v))^{\uparrow}, & \text{if } v\in \Rem_1,\\
			(L(v))^{\downarrow},& \text{if } v\in \Rem_2\cup \Rem_3 \cup \Rem_4,\\
			L(v),& \text{otherwise, } v\not \in \Mat.
		\end{cases}	
	\end{aligned}
\end{equation}
For each node $v'\in V_2'$, assuming $v\in\Rem_3\cup\Rem_4$ is the node such that $v'=f^\ast(v)$ and $i\in[m]$ is the node such that $((L(v),i) \in \M$,
\begin{equation}
	\begin{aligned}
		\label{eq-f_ast_v}
		\quad L'(v') =
		\begin{cases}
			i^{\uparrow}, & \text{if }v \in \Rem_3,\\
			i^{\downarrow},& \text{otherwise, } v\in \Rem_4.
		\end{cases}	
	\end{aligned}	
\end{equation}

\noindent\underline{Constructing $E'$.} $E'=E_1'\sqcup E_2'$ where $E'_1 = \{f^\ast(v)\rightarrow f(v):v\in \Rem_3\cup \Rem_4\}$ and
$$E'_2 = \{u' \rightarrow v':\big((L'(u')=L'(v')\big) \lor \big((L'(u'),L'(v')) \in E^{\M})\big)\land (g(u') \prec g(v') \text{ in } \pi_D)
\}.
$$
\end{definition}
%We define a surjective mapping $h:\D(G_D)\times \psi(D)\rightarrow \mathcal{G_D^\M}$ as follows.
\begin{theorem}\label{thm-injection}
$h(\cdot,\cdot)$ is an  injection from $\{(D,\Rem):D\in \mathcal{D}(G_D),\Rem\in \psi(D)\} $ to $\mathcal{D}(G^\M)$.
\end{theorem}
The proof of Theorem~\ref{thm-injection} is in the appendix.
Now we can prove the main theorem of this subsection.
\begin{theorem}\label{thm-dag-mapping}
	$\sum_{D\in \D(G_D)} \left(\prod_{v \in \overline{\mathcal{V}} (D)}p_{L(v)}\right)\left(\prod_{v\in \mathcal{V}(D)}p'_{L(v)}\right) \leq \sum_{D\in \D(G_D) }\prod_{v \text{ in } D}p^{-}_{L(v)}.$
\end{theorem}
\begin{proof}	
For each $i\in [m]$ where $(i,j)\in \M$, let
 	\begin{equation*}
 		\begin{array}{llll}
 			q^1_i\triangleq p'_{i},\quad q^2_i\triangleq p^{-}_{i} - p'_{i},\quad q^3_i\triangleq (p^{-}_{i} - p'_{i})p'_{j},\quad \text{and }\quad q^4_i\triangleq (p^{-}_{i} - p'_{i})(p^{-}_{j} - p'_{j}).
 		\end{array} 
	\end{equation*}
According to Fact \ref{fact:210}, $q^1_i + q^2_i + q^3_i + q^4_i = p^{-}_{i} + p^{-}_{j}(p^{-}_{i} - p'_{i})\geq p_i$.

Given $D=(V,E,L)\in\D(G_D)$ and $\Rem\in\psi(D)$, let $D'=h(D,\Rem)$.
For each $v$ in $D$ where $(L(v),j)\in \M$ for some $j\in[m]$, according to the definition of $\vec{p}^{\M}$, (\ref{eq-f_v}), and (\ref{eq-f_ast_v}), we have that 
\begin{itemize}
\item if $v\in \Rem_1$, then $p^{\M}_{L'(f(v))} =p'_{L(v)}=q^1_{L(v)}$;
\item if $v\in \Rem_2$, then $p^{\M}_{L'(f(v))} =p^{-}_{L(v)} - p'_{L(v)}=q^2_{L(v)}$;
\item if $v\in \Rem_3$, then $p^{\M}_{L'(f(v))} \cdot p^{\M}_{L'(f^\ast(v))}=(p^{-}_{L(v)} - p'_{L(v)})p'_{j}= q^3_{L(v)}$;
\item if $v\in \Rem_4$, then $p^{\M}_{L'(f(v))}\cdot p^{\M}_{L'(f^\ast(v))} =(p^{-}_{L(v)} - p'_{L(v)})(p^{-}_{j} - p'_{j})= q^4_{L(v)}$.
\end{itemize}
Moreover, for each $u\in V\setminus \Mat(D) = \overline{\mathcal{V}(D)} \setminus \Mat(D)$, we have
$p^{\M}_{L'(f(v))} = p_{L(v)}$.
Thus, for each $\Rem \in \psi(D)$,
\begin{equation*}
	\begin{aligned}
		\prod_{v' \text{ in } h(D,\Rem)}p^{\M}_{L'(v')}
		&= \prod_{v \in \overline{\mathcal{V}} \setminus \Mat }p_{L(v)}\prod_{v \in \Rem_1}q^1_{L(v)}
		 \prod_{v \in \Rem_2}q^2_{L(v)}	
		 \prod_{v \in \Rem_3}q^3_{L(v)}		
		 \prod_{v \in \Rem_4}q^4_{L(v)}
		\\&= \prod_{v \in \overline{\mathcal{V}} \setminus \Mat }p_{L(v)}\prod_{v \in \mathcal{V}}p'_{L(v)}
		\prod_{v \in \Rem_1\setminus \mathcal{V}}q^1_{L(v)}	
		 \prod_{v \in \Rem_2}q^2_{L(v)}	
		 \prod_{v \in \Rem_3}q^3_{L(v)}		
		 \prod_{v \in \Rem_4}q^4_{L(v)}.
	\end{aligned}
\end{equation*}
So
\begin{equation*}
	\begin{aligned}
	\sum_{\Rem\in \psi(D) }\prod_{v' \text{ in } h(D,\Rem)}p^{\M}_{L'(v')}
	&= \sum_{\Rem\in \psi(D) }\prod_{v \in \overline{\mathcal{V}} \setminus \Mat }p_{L(v)}\prod_{v \in \mathcal{V}}p'_{L(v)}
	\prod_{v \in \Rem_1\setminus \mathcal{V}}q^1_{L(v)}	
	\prod_{v \in \Rem_2}q^2_{L(v)}	
	\prod_{v \in \Rem_3}q^3_{L(v)}		
	\prod_{v \in \Rem_4}q^4_{L(v)}
	\\&=\prod_{v \in \overline{\mathcal{V}} \setminus \Mat }p_{L(v)}\prod_{v \in \mathcal{V}}p'_{L(v)}\sum_{\Rem\in \psi(D) }\prod_{v \in \Rem_1\setminus \mathcal{V}}q^1_{L(v)}	
	\prod_{v \in \Rem_2}q^2_{L(v)}	
	\prod_{v \in \Rem_3}q^3_{L(v)}		
	\prod_{v \in \Rem_4}q^4_{L(v)}
	\\&=\prod_{v \in \overline{\mathcal{V}} \setminus \Mat }p_{L(v)}\prod_{v \in \mathcal{V}}p'_{L(v)}\prod_{v \in \Mat\setminus \mathcal{V} }\left(q^1_{L(v)} + q^2_{L(v)}
	+q^3_{L(v)}  +q^4_{L(v)}\right)
	\\&\geq \prod_{v \in \overline{\mathcal{V}} \setminus \Mat }p_{L(v)}\prod_{v\in \mathcal{V}}p'_{L(v)}\prod_{v \in \Mat\setminus \mathcal{V} }p_{L(v)} 
	\\&= \prod_{v\in  \overline{\mathcal{V}}}p_{L(v)}\prod_{v\in \mathcal{V}}p'_{L(v)},
	\end{aligned}
\end{equation*}
where the third equality is according to the definition of $\psi(D)$. Finally,
	\begin{equation*}
		\begin{aligned}
			\sum_{D\in \D(G_D) }\left(\prod_{v \in \overline{\mathcal{V}} (D)}p_{L(v)}\right)\left(\prod_{v\in \mathcal{V}(D)}p'_{L(v)}\right)
			&\leq \sum_{D\in \D(G_D) }\sum_{\Rem\in \psi(D) }\prod_{v' \text{ in } h(D,\Rem)}p^{\M}_{L'(v')} \leq \sum_{D'\in \D(G^{\M})} \prod_{v \text{ in } D'}p^{\M}_{L'(v')}\\
           &= \sum_{D\in \D(G_D) }\prod_{v \text{ in } D}p^{-}_{L(v)},
		\end{aligned}
	\end{equation*}
where the second inequality is due to  Theorem~\ref{thm-injection} and the equality is by Proposition~\ref{obs-mapping-dag}.
\end{proof}

\subsection{Putting all things together}
The following lemma is implicitly proved in \cite{kolipaka}.
\begin{lemma}[\cite{kolipaka}]\label{lem:szegedy2}
For any undirected graph $G_D=([m],E_D)$ and probability vector $\vec{p}\in \Interior_a(G_D)/(1+\eps)$, $\sum_{i\in [m]}\frac{q_{\{i\}}(G_D,\vec{p})}{q_{\emptyset}(G_D,\vec{p})}\leq m/\eps$.
\end{lemma}

\begin{thmbis}{thm:main}[restated]
For any $\A\sim(G_D,\vec{p},\M,\vec{\delta})$, if $(1+\eps)\cdot \vec{p}^-\in \Interior_a(G_D)$, then the expected number of resampling steps performed by MT algorithm is most $m/\eps$, where $m$ is the number of events in $\A$.
\end{thmbis}
\begin{proof}
Fix any such $\A$. We have that 
\[
\E[T]\leq \sum_{D\in \D(G_D) }\prod_{v \text{ in } D}p^{-}_{L(v)} \leq \frac{\sum_{i\in [m]}q_{\{i\}}(G_D,\vec{p}^{-})}{q_{\emptyset}(G_D,\vec{p}^{-})}\leq \frac{m}{\eps},
\]
where the first inequality is by Theorems~\ref{thm-ub-R} and~\ref{thm-dag-mapping}, the second inequality is due to Theorem 4 in~\cite{kolipaka}, and the last inequality is according to Lemma \ref{lem:szegedy2}.
\end{proof}

%% file: sec3.tex
\section{Lower bound on the amount of intersection}\label{sec:lower-bound-intersection}

In order to explore how far beyond Shearer's bound MT algorithm is still efficient in general, we provide a lower bound on the amount of intersection between dependent events for general instances  (Theorem \ref{thm:lower-bound-intersection}).

We first introduce some notations. Given a bipartite graph $G_B=([m],[n],E_B)$, we call the vertex $i\in[m]$ left vertex and the vertex $j\in [n]$ right vertex. 
We call $G_B$ \emph{linear}\footnote{The notion is not arbitrary. The bipartite graph $G_B$ can be represented by a hypergraph in a natural way: each right vertex $j$ is represented by a node $v_j$ in the hypergraph, each left vertex $i$ is represented by a hyperedge $e_i$, and $v_j$ is in $e_i$ if and only if $(i,j)\in E_B$. A hypergraph is called linear if any two hyperedges share at most one node.}
if any two left vertices in $[m]$ share at most one common neighbor in $[n]$. Let $\Delta_{D}(G_B)$ denote the maximum degree of $G_D(G_B)$, and $\Delta_{B}(G_B)$ denote the maximum degree of the left vertices in $G_B$. If $G_B$ is clear from the context, we may omit $G_B$ from these notations. In addition, for a bipartite graph $G=(L\subset[m],R,E)$ and a probability vector $\vec{p}\in (0,1)^m$, we define\footnote{It is possible that $\lowerBB(G,\vec{p})<0$.}
\[
\lowerBB(G,\vec{p}) \triangleq \frac{\big(\min_{i\in L} p_i\big)^2\cdot\left(-|\cup_{i\in L}\Neighbor_{G}(i)|+\sum_{i\in L} |\mathcal{N}_G(i)|\cdot p_i^{1/|\mathcal{N}_G(i)|} \right)}{\sqrt{|L|}\cdot\Delta_D(G)\cdot\Delta_B(G)^2}.
\]
and $\lowerBB^+(G,\vec{p}) \triangleq \max\{\lowerBB(G,\vec{p}),0\}$.

We use $\EventSet\sim(G_B,\vec{p})$ to denote that (i) $G_B$ is an event-variable graph of $\EventSet$ and (ii) the probability vector of $\EventSet$ is $\vec{p}$. 
Let $\M=\{(i_1,i_1'),(i_2,i_2'),\cdots\}$ be a matching of $G_D(G_B)$, and  $\vecdelta=(\delta_{i_1,i_1'},\delta_{i_2,i_2'},\cdots)\in(0,1)^{|\M|}$ be another probability vector.
We say that an event set $\A$ is of the setting $(G_B,\vec{p},\M,\vecdelta)$, and write $\A\sim(G_B,\vec{p},\M,\vecdelta)$, if $\A\sim(G_B,\vec{p})$
and $\Pr(A_i\cap A_{i'})\geq \delta_{i,i'}$ for each pair $(i,i')\in\M$.

We call an event $A$ \emph{elementary}, if $A$ can be written as 
$(X_{i_1} \in S_{i_1})\land (X_{i_2} \in S_{i_2}) \land\cdots \land (X_{i_k} \in S_{i_k})$
where $S_{i_1},\cdots,S_{i_k}$ are subsets of the domains of variables.
We call an event set $\A$ \emph{elementary} if all events in $\A$ are elementary. 
\begin{theorem}\label{thm:lower-bound-intersection}
	Let $G_B=([m],[n],E_B)$ be a bipartite graph, $\vec{p}\in(0,1]^m$ be a probability vector, and $L_1,L_2,\cdots,L_t$ be a collection of disjoint subsets of $[m]$. For each $k\in [t]$,
	let $G_k$ denote the induced subgraph on $L_k\cup \left(\cup_{i\in L_k}\Neighbor_{G_B}(i)\right)$ and
	$E_k$ denote the edge set of $G_D(G_k)$.
	If all $G_k$'s are linear, then the following holds.
	
	If $\A\sim(G_B,\vec{p})$, then there is a matching $\M$ of $G_D(G_B)$ 
	satisfying that $\sum_{(i,i')\in \M\cap E_k}\Pr(A_i\cap A_{i'})^2 \geq \left(\lowerBB^+(G_k,\vec{p})\right)^2$ for any $k$.
\end{theorem}

The proof of Theorem~\ref{thm:lower-bound-intersection} mainly consists of two parts. First, we show that there is an elementary event set which approximately achieves the minimum amount of intersection between dependent events (Lemma \ref{lem:general2elementary}). Then, for elementary event sets, by applying AM-GM inequality, we obtain a lower bound on the total amount of overlap on common variables, which further implies a lower bound on the amount of intersection between dependent events (Lemma \ref{lemma:elementary}).

\begin{lemma}\label{lem:general2elementary}
Let $G_B=([m],[n],E_B)$ be a linear bipartite graph, 
$E_D$ be the edge set of $G_D(G_B)$, and $\vec{p}\in(0,1]^m$ is a probability vector. Let $\gamma$ denote the minimum $\sum_{(i_0,i_1)\in E_D}\Pr[A_{i_0}\cap A_{i_1}]$ among all event sets $\A=(A_1,\cdots,A_m)\sim(G_B,\vec{p})$. Then there is an elementary event set $\A'$ such that
 %\begin{itemize}
 %\item[(a)] each event $A\in\A$ is elementary; and
 $\sum_{(i_0,i_1)\in E_D}\Pr[A_{i_0}'\cap A_{i_1}']\leq (\mdl_B(G_B))^2\cdot\gamma$.
 %\end{itemize}
\end{lemma}
\begin{proof}
For simplicity, we let $\Delta\triangleq\Delta_B(G_B)$.
Without loss of generality, we assume that each random variable $X_i$ is uniformly distributed over $[0,1]$.
%, since (a) and (b) every probability measure on $(\mathcal{R},\mathcal{B})$ corresponds to some random variable defined on the standard probability space $(0,1)$ with Lebesgue measure.
Let $\A\sim(G_B,\vec{p})$ be an event set where $\sum_{(i_0,i_1)\in E_D}\Pr[A_i\cap A_j]=\gamma$.  We will replace $A_i$ with an elementary $A_i'$ one by one for each $i=1,2,\cdots,m$, so that the resulted event set $\A'$ satisfies $\sum_{(i_0,i_1)\in E_D}\Pr[A_{i_0}'\cap A_{i_1}']\leq \mdl^2 \cdot \sum_{(i_0,i_1)\in E_D}\Pr[A_{i_0}\cap A_{i_1}]=\mdl^2 \cdot \gamma$. % Then we obtain an event set $\A\sim (G_B,\vec{p})$ where each event $A_i\in\A$ is elementary, and we will argue that $\sum_{(i,j)\in E_D}\Pr[A_i\cap A_j]\leq \Delta \cdot \gamma$.

More precisely, fix $i\in[m]$ and suppose $A_1,\cdots,A_{i-1}$ have been replaced with elementary events $A_1',\cdots,A_{i-1}'$ respectively. For simplicity of notations, for any pair $i_0<i_1$, we abbreviate $\Pr[A_{i_0}\cap A_{i_1}]$, $\Pr[A_{i_0}'\cap A_{i_1}]$ and $\Pr[A_{i_0}'\cap A_{i_1}']$ to $p_{i_0,i_1}$, $p_{i_0,i_1}'$ and $p_{i_0,i_1}''$ respectively. Without loss of generality, we assume $A_i$ depends on variables $X_1,X_2,\cdots,X_k$. For every $j\in[k]$, we define
\[
P_j(x_j):=\sum_{i_0<i, i_0\in \Neighbor_{G_B}(j)}\frac{1}{\Delta}\cdot \Pr[A_{i_0}'\mid X_j=x_j]+\sum_{i_0>i, i_0\in \Neighbor_{G_B}(j)}\Pr[A_{i_0}\mid X_j=x_j].
\]
for $x_j\in[0,1]$. Without loss of generality, we assume $P_j(\cdot)$ is non-decreasing. Let $\mu:[0,1]^k\rightarrow\{0,1\}$ be the indicator of $A_i$, then
\begin{align*}
\int_{x_1,\cdots,x_k}\mu(x_1,\cdots,x_k)\d x_1\cdots \d x_k=\Pr[A_i],
\end{align*}
For each $j\in[k]$, let 
\[
\mu_j(x_j):=\Pr[A_i\mid X_j=x_j]=\int_{x_1,\cdots,x_{j-1},x_{j+1},\cdots,x_k}\mu(x_1,\cdots,x_k)\d x_1\cdots\d x_{j-1}\d x_{j+1}\cdots\d x_k.
\]
Noticing that $G_B$ is linear (i.e., any two events share at most one common variable), we have
\begin{align}\label{eq:sec3_1}
\int_{x_j} P_j(x_j)\mu_j(x_j)\d x_j
=\sum_{i_0<i, i_0\in \Neighbor_{G_B}(j)}\frac{p'_{i_0,i}}{\mdl}+\sum_{i_0>i, i_0\in \Neighbor_{G_B}(j)}p_{i_0,i}.
\end{align}

Let $A_i'$ be an elementary event such that it happens if and only if $(x_1,\cdots,x_k)\in [0,q_1]\times\cdots\times [0,q_k]$. Here $q_1,\cdots,q_k$ is a set of positive real numbers satisfying that
 \begin{itemize}
 \item[(i)] $\Pi_{j=1}^kq_j=\Pr[A_i]$. That is, $\Pr[A_i']=\Pr[A_i]$;
% \item[(ii)] $\Pr[A_i'\cap (X_j\geq q_j)]=\Pr[A_i\setminus A_i']/k$ for each $j\in[k]$.
 \item[(ii)] $\int_{x_1\geq q_1}\mu_1(x_1)\d x_1=\int_{x_2\geq q_2}\mu_2(x_2)\d x_2\cdots=\int_{x_k\geq q_k}\mu_k(x_k)\d x_k$.
 \end{itemize}
 \begin{claim}
Such $\{q_1,\cdots,q_k\}$ exists. Thus so does $A_i'$.
 \end{claim}
 \begin{proof}
We prove a generalized statement in which $\Pi_{j=1}^kq_j$ can be required to be an arbitrary number in $[0,1]$. Our proof is by induction on $k$. The base case when $k=1$ is trivial. Now we assume that for any preset $q'\in(0,1]$, there exist $\{q_1,\cdots,q_{k-1}\}$ satisfying that 
\begin{itemize}
\item[(i)] $\Pi_{j=1}^{k-1}q_j=q'$ and
\item[(ii)] $\int_{x_1\geq q_1}\mu_1(x_1)\d x_1=\cdots=\int_{x_{k-1}\geq q_{k-1}}\mu_{k-1}(x_{k-1})\d x_{k-1}$.
\end{itemize}
Let $f(q')$ denote the minimum $\int_{x_1\geq q_1}\mu_1(x_1)\d x_1$ among all such $\{q_1,\cdots,q_{k-1}\}$'s. It is easy to see that $f(1)=0$ and $f$ is continuous and non-increasing.

Fix an arbitrary $q\in[0,1]$. We define $g(q''):=\int_{x_k\geq q/q'}\mu_k(x_k)\d x_k$ for $q''\in[q,1]$. Obviously, $g(q)=0$ and $g$ is continuous and non-decreasing. So there must exist a $q^\ast\in[q,1]$ such that $g(q^\ast)=f(q^\ast)$. Then let $\{q_1^\ast,\cdots,q_{k-1}^\ast\}$ be a set of positive real numbers where 
\begin{itemize}
\item[(i)] $\Pi_{j=1}^{k-1}q_j^\ast=q^\ast$ and
\item[(ii)] $f(q^\ast)=\int_{x_1\geq q_1^\ast}\mu_1(x_1)\d x_1=\cdots=\int_{x_{k-1}\geq q_{k-1}^\ast}\mu_{k-1}(x_{k-1})\d x_{k-1}$.
\end{itemize}
Let $q_k^\ast=q/q^\ast$. It is obvious that $\Pi_{j=1}^kq_j^\ast=q$ and  $f(q^\ast)=g(q^\ast)=\int_{x_k\geq q_k^\ast}\mu_k(x_k)\d x_k$. This completes the induction step.
\end{proof}

\begin{claim}\label{claim:sec3_2}
For every $j\in[k]$, we have
\[
\sum_{i_0<i, i_0\in \Neighbor_{G_B}(j)}\frac{p_{i_0,i}''}{\mdl}+\sum_{i_0>i, i_0\in \Neighbor_{G_B}(j)}p_{i_0,i}'\leq \sum_{i_0<i, i_0\in \Neighbor_{G_B}(j)}p'_{i_0,i}+\mdl\cdot\sum_{i_0>i, i_0\in \Neighbor_{G_B}(j)}p_{i_0,i}.
\]
\end{claim}
\begin{proof}
Let $\mu_{i'},\mu_{i\cap i'}$, and $\mu_{i'\setminus i}$ denote the indicator functions of the events $A_i'$, $A_i'\cap A_i$, and $A_i'\setminus A_i$ respectively. Since $\Pr[A_i']=\Pr[A_i]$,
\[
\int_{x_1\geq q_1}\mu_1(x_1)\d x_1+\cdots+\int_{x_k\geq q_k}\mu_k(x_k)\d x_k\geq \Pr[A_i\setminus A_i']=\Pr[A_i'\setminus A_i]=\int_{x_1,\cdots,x_k}\mu_{i'\setminus i}(x_1,\cdots,x_k)\d x_1\cdots\d x_k.
\]
Fix $j\in[k]$, then
\[
\int_{x_j\geq q_j}\mu_j(x_j)\d x_j\geq \frac{1}{k}\cdot \int_{x_1,x_2,\cdots,x_k}\mu_{i'\setminus i}(x_1,\cdots,x_k)\d x_1\cdots\d x_k.
\]
Since $P_j(x_j)$ is non-decreasing and $k\leq \mdl$, we have 
\[
\int_{x_j\geq q_j}P_j(x_j)\mu_j(x_j)\d x_j\geq \frac{1}{\mdl}\cdot\int_{x_1,x_2,\cdots,x_k}P_j(x_j)\mu_{i'\setminus i}(x_1,\cdots,x_k)\d x_1\cdots\d x_k.
\]
According to Equation \ref{eq:sec3_1},
\begin{align*}
&\sum_{i_0<i, i_0\in \Neighbor_{G_B}(j)}p'_{i_0,i}+\mdl\cdot\sum_{i_0>i, i_0\in \Neighbor_{G_B}(j)}p_{i_0,i}
=\mdl\cdot \int_{x_j}P_j(x_j)\mu_j(x_j)\d x_j\\
=& \mdl\cdot\int_{x_j\geq q_j}P_j(x_j)\mu_j(x_j)\d x_j+\mdl\cdot\int_{x_j< q_j}P_j(x_j)\mu_j(x_j)\d x_j\\
\geq& \mdl\cdot\int_{x_j\geq q_j}P_j(x_j)\mu_j(x_j)\d x_j+\mdl\cdot\int_{x_1,\cdots,x_k}P_j(x_j)\mu_{i\cap i'}(x_1,\cdots,x_k)\d x_1\cdots\d x_k\\
\geq& \int_{x_1,\cdots,x_k}P_j(x_j)\mu_{i'\setminus i}(x_1,\cdots,x_k)\d x_1\cdots\d x_k+\int_{x_1,\cdots,x_k}P_j(x_j)\mu_{i\cap i'}(x_1,\cdots,x_k)\d x_1\cdots\d x_k\\
=& \int_{x_1,\cdots,x_k}P_j(x_j)\mu_{i'}(x_1,\cdots,x_k)\d x_1\cdots\d x_k\\
=& \sum_{i_0<i, i_0\in \Neighbor_{G_B}(j)}\frac{\Pr[A_{i_0}'\cap A_i']}{\mdl}+\sum_{i_0>i, i_0\in \Neighbor_{G_B}(j)}\Pr[A_{i_0}\cap A_i'].
\end{align*}
This completes the proof.
\end{proof}

From Claim \ref{claim:sec3_2}, we have 
\begin{align}\label{eq:sec3_2}
\sum_{i_0<i, i_0\in \Neighbor_{G_D}(i)}\frac{p_{i_0,i}''}{\Delta}+\sum_{i_0>i, i_0\in \Neighbor_{G_D}(i)}p_{i_0,i}'\leq \sum_{i_0<i, i_0\in \Neighbor_{G_D}(i)}p'_{i_0,i}+\Delta\cdot\sum_{i_0>i, i_0\in \Neighbor_{G_D}(i)}p_{i_0,i},
\end{align}
By summation over all $i\in[m]$, we finish the proof:
\[
\sum_{(i_0,i)\in E_D}\frac{p_{i_0,i}''}{\Delta}\leq \Delta\cdot \sum_{(i_0,i)\in E_D}p_{i_0,i}.
\]
\end{proof}

%In the following, we provide a lower bound of $\sum_{(i,j)\in E_B}\Pr[A_i\cap A_j]$ for elementary event set $\A$, which then implies a lower bound for general event sets according to Lemma \ref{lem:general2elementary}. Fix any elementary event set $\A$ of the setting $(G_B,\vec{p})$. For each $i\in[m]$ and $j\in[n]$ where $(i,j)\in E_B$, let $z_{i,j}$ be the length on $X_j$ occupied by $A_i$. Since $\Pi_{j\in\Neighbor_{G_B}(i)}z_{i,j}=p_i$, we have
% \[
% \sum_{j\in\Neighbor_{G_B}(i)}z_{i,j}\geq |\Neighbor_{G_B}(i)|\cdot p_i^{1/|\Neighbor_{G_B}(i)|}
% \]
% by AM-GM Inequality. Then
%  \begin{align}\label{eq:square1}
% \sum_{(i,j)\in E_B}z_{i,j}\geq \sum_{i\in[m]}|\Neighbor_{G_B}(i)|\cdot p_i^{1/|\Neighbor_{G_B}(i)|}
% \end{align}
% If R.H.S of Inequality \ref{eq:square1} is strictly greater than $|E_B|$, then we can conclude that there must exist intersections between adjacent pairs of events. More precisely, for each $j\in[n]$, let $\zeta_j:=\sum_{i\in\Neighbor_{G_B}(j)}z_{i,j}-1$. Then $\sum_{i_1,i_2\in\Neighbor_{G_B}(j)}\Pr[A_i\cap A_j]\geq \zeta_j\cdot p_{\min}^2$, thus $\sum_{i_1,i_2\in\Neighbor_{G_B}(j)}\Pr[A_i\cap A_j]^2\geq \frac{\zeta_j\cdot p_{\min}^2}{|\Neighbor_{G_B}(j)|^2}$. Finally, since $|\Neighbor_{G_B}(j)|$ is at most the maximum degree of $G_D$, by summation over all $j\in[n]$, we can conclude that
% \[
%\sum_{(i_1,i_2)\in E_D}\Pr[A_i\cap A_j]^2=\sum_{j\in[n]}\sum_{i_1,i_2\in\Neighbor_{G_B}(j)}\Pr[A_i\cap A_j]^2\geq\sum_{j\in[n]}\frac{\zeta_j\cdot p_{\min}^2}{\Delta(G_D)^2}
% \]

\begin{lemma}\label{lemma:elementary}
Let $G_B=([m],[n],E_B)$ be a linear bipartite graph and $\vec{p}$ be a probability vector. Then for any elementary $\A=(A_1,\cdots,A_m)\sim(G_B,\vec{p})$,
\[
\sum_{(i_0, i_1)\in E_D }\Pr\left(A_{i_0}\cap A_{i_1}\right) \geq \sqrt{m}\cdot\Delta_D(G_B)\cdot\Delta_B(G_B)^2\cdot \lowerBB(G_B,\vec{p}),
\]
 where $E_D$ is the edge set of $G_D(G_B)$ 
\end{lemma}

\begin{proof}
For simplicity of notation, we let $\mathcal{N}$ stand for $\mathcal{N}_{G_B}$.
Without loss of generality, we assume that each variable $X_i$ is uniformly distributed over $[0,1]$. %that each random variable is
%the canonical uniform random variable, which takes on all values in its domain $[0, 1)$ 
%with equal probability. 
As $\A$ is elementary, each $A_i$ can be written as 
$\bigwedge_{j\in\mathcal{N}(i)}[X_{j} \in S_{i}^j]$
where $S_{i}^j\subset[0,1]$.
%\warn{the support of $A_i$ in the domain $[0, 1)$ of random variable $X_j$.}
%Let $d_1,\cdots,d_n$ be the degrees of $A_1,\cdots,A_n$ in the bipartite graph.
%For any $A_i$, the $d_i$ sides of the hypercube 
Let $\mu$ be the Lebesgue measure.

On one hand, according to the AM–GM inequality, 
\begin{align}\label{eq:sec_341}
\sum_{i\in[m]}\sum_{j \in \mathcal{N}(i)} \mu(S^j_i)\geq \sum_{i\in[m]}|\mathcal{N}(i)|\cdot\big(\Pi_{j\in\mathcal{N}(i)}\mu(S^j_i)\big)^{1/|\mathcal{N}(i)|}=\sum_{i\in[m]}|\mathcal{N}(i)|\cdot p_i^{1/|\mathcal{N}(i)|}.
\end{align}
On the other hand, 
\begin{align}\label{eq:sec_342}
\sum_{i \in [m]}\sum_{j \in \mathcal{N}(i)} \mu(S^j_i) = \sum_{j \in [n]}\sum_{i \in \mathcal{N}(i)} \mu(S^j_i)\leq  n + \sum_{j\in [n]}\sum_{i_0\neq i_1\in \mathcal{N}(j)}\mu\left(S^j_{i_0}\cap S^{j}_{i_1}\right)
\end{align}
By Inequalities \ref{eq:sec_341} and \ref{eq:sec_342} and noticing $G_B$ is linear, we have that
\begin{align}\label{eq-inter-side}
\sum_{(i_0,i_1)\in E_D}\sum_{j\in\mathcal{N}(i_0)\cap \mathcal{N}(i_1) }\mu\left(S^j_{i_0}\cap S^{j}_{i_1}\right) =
\sum_{j\in [n]}\sum_{i_0\neq i_1\in \mathcal{N}(j)}\mu\left(S^j_{i_0}\cap S^{j}_{i_1}\right) \geq \left(\sum_{i\in [m]} |\mathcal{N}(i)|\cdot p_i^{1/|\mathcal{N}(i)|}\right)-n.
\end{align}
Moreover, given any $(i_0,i_1)\in E_D$, where $\{j\}= \mathcal{N}(i)\cap \mathcal{N}(i')$, we have that
\begin{equation}
\begin{aligned}\label{eq-inter-event}
\Pr(A_{i_0}\cap A_{i_1}) &\geq \mu\left(S^j_{i_0}\cap S^{j}_{i_1}\right) \cdot\left(\prod_{k \in \mathcal{N}(i_0)\setminus \{j\}}
\mu(S^k_{i_0})\right)\cdot \left(\prod_{k'\in\mathcal{N}(i_1)\setminus  \{j\}}\mu(S^{k'}_{i_1})\right)
\\&\geq \mu\left(S^j_{i_0}\cap S^{j}_{i_1}\right)\cdot p_{i_0}\cdot p_{i_1}.
\end{aligned}
\end{equation}
Finally, combining (\ref{eq-inter-side}) with (\ref{eq-inter-event}), we concludes that 
\begin{align*}
\sum_{(i_0, i_1)\in E_D }\Pr(A_{i_0}\cap A_{i_1}) &\geq
\sum_{(i_0,i_1)\in E_D}\sum_{j\in\mathcal{N}(i_0)\cap \mathcal{N}(i_1)}\mu\left(S^j_{i_0}\cap S^{j}_{i_1}\right) \cdot p_{i_0}\cdot p_{i_1}
\\ &\geq  \left(\min_{i\in [m]} p_i\right)^2\left(\sum_{i} |\mathcal{N}(i)|\cdot p_i^{1/|\mathcal{N}(i)|} - n\right)\\
&= \sqrt{m}\cdot\Delta_D(G_B)\cdot\Delta_B(G_B)^2\cdot \lowerBB(G_B,\vec{p}).
\end{align*}
\end{proof}

The following lemma is a special case of Theorem \ref{thm:lower-bound-intersection} where $t=1$ and $L_1=[m]$. In fact, Theorem \ref{thm:lower-bound-intersection} is proved by applying Lemma \ref{lemma-indep-pairs} to each $G_k$ separately. 
\begin{lemma}\label{lemma-indep-pairs}
Let $G_B=([m],[n],E_B)$ be a linear bipartite graph and $\vec{p}$ be a probability vector. If $\A\sim(G_B,\vec{p})$, then $\A\sim(G_B,\vec{p},\M,\vec{\delta})$ for some matching $\M$ of $G_D(G_B)$ and some $\vecdelta\in(0,1)^{|\M|}$
satisfying that $\sum_{(i,i')\in \M}\delta^2_{i,i'} \geq \big(\lowerBB^+(G_B,\vec{p})\big)^2$.
\end{lemma}

\begin{proof}
Given an instance $\A\sim(G_B,\vec{p})$, we construct such a $\M$ greedily as follows.

We maintain two sets $E$ and $\M$,
which are initialized as $E_D$ and $\emptyset$ respectively. We do the following iteratively until $E$ becomes empty:
select a edge $(i_0,i_1)$ with maximum $\Pr(A_{i_0} \cap A_{i_1})$ from $E$, add $(i_0,i_1)$ to $\M$, and delete all edges connecting $i_0$ or $i_1$ from $E$ (including $(i_0,i_1)$). 

Let $\Delta_D$ and $\Delta_B$ denote $\Delta_D(G_B)$ and $\Delta_B(G_B)$ respectively. In each iteration, at most $2\Delta_D$ edges are deleted from $E$ and for each deleted edge $(i,i')$, $\Pr(A_{i} \cap A_{i'})^2 \leq \Pr(A_{i_0} \cap A_{i_1})^2$. Based on this observation, it is easy to see that
\begin{align}\label{eq:361}
\sum_{(i_0,i_1)\in \M} \Pr(A_{i_0} \cap A_{i_1})^2 \geq \frac{1}{2\Delta_D}\sum_{(i, i')\in E_D}  \Pr(A_i \cap A_{i'})^2.
\end{align}
Moreover, according to Lemma~\ref{lem:general2elementary} and ~\ref{lemma:elementary}, it has that
\begin{align}\label{eq:362}
\sum_{(i,i')\in E_D}\Pr(A_i\cap A_{i'})^2 \geq \frac{1}{|E_D|}\cdot \left(\sum_{(i,i')\in E_D }\Pr(A_i\cap A_{i'})\right)^2
\geq \frac{m\cdot\Delta_D^2\cdot \left(\lowerBB^+(G_B,\vec{p})\right)^2}{|E_D|},
\end{align}
By combining Inequality \ref{eq:361} and \ref{eq:362}, setting $\delta_{i,i'}=\Pr(A_i\cap A_{i'})$, and noting  $2|E_D|\leq m\Delta_D$, we finish the proof.
\end{proof}

\begin{proof}[Proof of Theorem~\ref{thm:lower-bound-intersection}] For each $k\in[t]$, by applying Lemma \ref{lemma-indep-pairs} to $G_k$, we have that $\A\sim(G_B,\vec{p},\M_k,\vec{\delta}_k)$ for some matching $\M_k\subseteq E_k$ and some $\vec{\delta}_k$ where $\sum_{(i,i')\in \M_k}\delta^2_{i,i'} \geq \big(\lowerBB^+(G_k,\vec{p})\big)^2$. Note that $E_k$'s are disjoint with each other, so $\M_1\cup\M_2\cup\cdots\cup \M_t$ is still a matching. By letting $\M=\M_1\cup\M_2\cup\cdots\cup \M_t$ and $\vec{\delta}=(\vec{\delta}_1,\cdots,\vec{\delta}_t)$, we conclude the theorem.
\end{proof}
	
\begin{remark}\label{remark:simplified}
Given a bipartite graph $G$, its \emph{simplified graph} is defined to be obtained from $G$ by
deleting all the right nodes which only have one neighbor
and combining all the right nodes with the same neighbor set. Notice that if $G$ is linear, so is its simplified graph.

Theorem~\ref{thm:lower-bound-intersection} can be slightly generalized: it is sufficient that the simplified graph of $G_k$ instead of $G_k$ itself is linear.% It still holds when $G_k$ is not linear
%but its simplified graph is.
%Besides, Theorems~\ref{thm-MT-beyond-Shearer} and~\ref{thm-graph-beyond} can also be generalized in a similar way. That is, the condition that a subgraph is linear can be replaced with the condition that its simplified graph is linear.
\end{remark}

%% file: sec4.tex
\section{\lqc{The Moser-Tardos algorithm is} %MT algorithm is 
 beyond Shearer's bound}\label{sec:beyond-shearer}

In this section, we prove Theorem~\ref{thm-dep-graph-beyond}. Given a dependency graph $G_D$, a vector $\vec{p}$ and a chordless cycle $C$ in $G_D$,
define 
$$r(G_D,\vec{p},C) \triangleq |C|\cdot\big(\min_{j\in C} p_j\big)^4\cdot\left(\frac{2\sum_{j\in C} \sqrt{p_j}}{|C|}  - 1\right)^2.$$
and
$$r^+(G_D,\vec{p},C) \triangleq |C|\cdot\big(\min_{j\in C} p_j\big)^4\cdot\left(\max\left\{\frac{2\sum_{j\in C} \sqrt{p_j}}{|C|}  - 1,0\right\}\right)^2.$$
Then Theorem~\ref{thm-dep-graph-beyond} is obvious by Lemmas~\ref{lemma-beyond} and~\ref{lemma-beyond-exist}.

\begin{lemma}\label{lemma-beyond}
    Given $G_D$, $\vec{p}$ and $\eps>0$, let $C_1,C_2,\cdots,C_\ell$ be any disjoint chordless cycles in $G_D$.
    If 
    \[d((1+\eps)\vec{p},G_D) < \frac{1}{544}\sum_{i\leq \ell}r^+(G_D,\vec{p},C_i),
    \]
    then for any variable-generated event system $\A\sim(G_D,\vec{p})$, the expected number of resampling steps performed by MT algorithm is most $m/\eps$.
\end{lemma}
\begin{proof}
  Fix such an instance $\A$. Define $\delta_{i,i'}:=\Pr(A_i\cap A_{i'})$. Let $G_B$ denote the event-variable graph of $\A$. 
  Let $G_k$ denote the induced subgraph of $G_B$ on $C_k\cup \left(\cup_{i\in C_k}\Neighbor_{G_B}(i)\right)$. According to Remark \ref{remark:simplified}, it is lossless to assume $G_k$ is a cycle of length $2|C_k|$. 
  Thus we have 
  \begin{equation}\label{eq-F}
  \lowerBB^+(G_k,\vec{p}) \geq \frac{\big(\min_{i\in C_k} p_i\big)^2\cdot\left(- |C_k| + \sum_{i\in L} 2\sqrt{p_i} \right)}{8\sqrt{|C_k|}}.
 \end{equation}
  According to Theorem~\ref{thm:lower-bound-intersection}, there is a matching $\M$ of $G_D$ such that
	$\sum_{(i,i')\in \M}\delta_{i,i'}^2 \geq \sum_{k\leq \ell}\left(\lowerBB^+(G_k,\vec{p})\right)^2$.
	Define $\vec{p}^-$ as (\ref{eq-pminus}).
	We have $(1+\eps)\vec{p}^-\leq (1+\eps)\vec{p}$
	and
	\begin{equation*}
	\begin{aligned}
	||(1+\eps)\vec{p} - (1+\eps)\vec{p}^-||_1\geq ||\vec{p} - \vec{p}^-||_1\geq \frac{2}{17}\sum_{(i,i')\in \M}\delta^2_{i,i'} \geq \frac{2}{17}\sum_{k\leq \ell}\left(\lowerBB^+(G_k,\vec{p})\right)^2.
	\end{aligned}
	\end{equation*}
	Combining with (\ref{eq-F}), we have
	\begin{equation*}
	\begin{aligned}
	||(1+\eps)\vec{p} - (1+\eps)\vec{p}^- ||_1\geq \frac{1}{544}\sum_{i\leq \ell}r^+(G_D,\vec{p},C_i)> d((1+\eps)\vec{p},G_D),	
	\end{aligned}
	\end{equation*}
	where the last inequality is by the condition of the lemma.
	Thus by Definition~\ref{def-distance}, 
    we have $ (1+\eps)\vec{p}^-$ is in the Shearer's bound of $G_D$.
	Combining with Theorem~\ref{thm:main}, we have the expected number of resampling steps performed by the Moser-Tardos algorithm is most $m/\eps$.
%	As $G_D$ and $\vec{p}$ satisfy Condition~\ref{con:dependency},
%	$G_B$ and $\vec{p}$ satisfy Condition~\ref{con:generalized_l1norm}.
%	By applying Theorem~\ref{thm-MT-beyond-Shearer}, we conclude the theorem.
\end{proof}

\begin{lemma}\label{lemma-beyond-exist}
 Given $G_D$ and any chordless cycle $C$ in $G_D$,
 there is some probability vector $\vec{p}$ beyond the Shearer's bound of $G_D$ and with
\[d(\vec{p},G_D)\geq \frac{1}{545}\cdot r(G_D,\vec{p},C)>2^{-20} \ell^{-3}
\]
such that for any variable-generated event system $\A\sim(G_D,\vec{p})$, the expected number of resampling steps performed by MT algorithm is most $2^{29}\cdot m^2 \cdot |C|^3$.
\end{lemma}

The following two lemmas will be used in the proof of Lemma~\ref{lemma-beyond-exist}.

\begin{lemma}\cite{shearer1985problem}\label{lem-vol-shearer}
$q_{\emptyset}(G_D,\vec{p}) = 1 - \Pr(\bigcup_{A\in \A}A)$ holds for any extremal instance $\A \sim (G_D,\vec{p})$.
\end{lemma}

\begin{lemma}\cite{shearer1985problem}\label{lem-diff-shearer}
Suppose $\vec{p}$ is the Shearer's bound of $G_D=([m],E_D)$. Then for $i\in[m]$, 
\begin{align*}
\frac{\partial{q_{\emptyset}(G_D,\vec{p})}}{\partial{p_i} } = -\Pr\left(\bigcap_{j\notin \mathcal{N}_{G_D}(i)\cup\{i\}}\overline{A_j}\right)
\end{align*}
holds for any $\A\sim(G_D,\vec{p})$ satisfying that $A_{i'}\cap A_{i''}=\emptyset$ for any $(i',i'') \in E_D$ where $i',i'' \neq i$.
\end{lemma}

\begin{proof}[Proof of Lemma~\ref{lemma-beyond-exist}]
Let $\ell = |C|$ and $\vec{\lambda} = \left(\frac{1}{4},\cdots,\frac{1}{4},\frac{1}{4}\right)$.
Let $\A\sim(C,\vec{\lambda})$ be an extremal instance defined as follows:
$\A = (A_1,\cdots,A_{\ell})$ is a variable-generated event system fully determined
a set of underlying mutually independent random variables $\{X_1,\cdots,X_{\ell}\}$.
Moreover, $A_i=[X_i<1/2]\land [X_{i+1}\geq 1/2]$ for each $i\in [\ell - 1]$, and  $A_{\ell}=[X_{\ell}<1/2]\land [X_{1}\geq 1/2]$. According to Lemma ~\ref{lem-vol-shearer},
\[
q_{\emptyset}(C,\vec{\lambda}) =\Pr\left(\bigcup_{i\in [\ell]}A_i\right)= \frac{1}{2^{\ell-1}}.
\]
Besides, according to Lemma \ref{lem-diff-shearer}, for any $\vec{\lambda}' =(\frac{1}{4},\cdots,\frac{1}{4},\frac{1}{4} + \eps)$ in the Shearer's bound of $C$, 
\begin{align*}
\frac{\partial{q_{\emptyset}(C,\vec{\lambda}')}}{\partial{\lambda'_{\ell}} } = - \Pr\left(\bigcap_{i\in[2,\ell-2]}\overline{A_i}\right) = - \frac{\ell- 2}{2^{\ell-3}}.
\end{align*}
Thus, for any $\vec{\lambda}\leq \vec{\lambda'}\leq \vec{\lambda}'':= \left(\frac{1}{4},\cdots,\frac{1}{4},\frac{1}{4} + \frac{1}{4(\ell-1)}\right)$, we have that
\begin{align*}
q_{\emptyset}(C,\vec{\lambda}'') = q_{\emptyset}(C,\vec{\lambda}) + \int_{\frac{1}{4}}^{\lambda''_\ell}\frac{\partial{q_{\emptyset}(C,\vec{\lambda}')}}{\partial{\lambda_{\ell}'} } d\lambda_\ell'
>\frac{1}{2^{\ell-1}} - \frac{\ell - 2}{\ell - 1}\cdot\frac{1}{2^{\ell-1}}
=\frac{1}{\ell - 1}\cdot\frac{1}{2^{\ell-1}}.
\end{align*}
Hence $\vec{\lambda}''$ is in the Shearer's bound of $C$. 
Thus, there exists $q>0$ such that $\vec{q}$ defined as follows is on the Shearer's boundary of $G_D$:
\begin{align*}
\forall i\in [m]:\quad
q_{i}=\begin{cases}
\frac{1}{4}& \text{if } i \in [\ell-1],\\
\frac{1}{4} + \frac{1}{4(\ell-1)}& \text{if } i = \ell,\\
q & \text{otherwise}.
\end{cases}
\end{align*}
One can verify that 
\begin{equation}
\begin{aligned}
\label{eq-lb-r}
r^+(G_D,\vec{q},C) = r(G_D,\vec{q},C)
> \ell \cdot \frac{1}{4^4}\cdot\left(\frac{1}{2\ell^2}\right)^2>\frac{1}{2^{10}\cdot \ell^{3}} .
\end{aligned}
\end{equation}
\vspace{2ex}

Define 
$$f(\delta) =  545\cdot d((1 + \delta)\vec{q},G_D) - r^+(G_D,(1 + \delta)\vec{q},C).$$
One can verify that $f(0) < 0$ because $d(\vec{q},G_D) = 0$
and $r^+(G_D,\vec{q},C)>0$.
Moreover, let $\delta'$ be large enough such that 
$(1 + \delta')\vec{q}\not \in \Interior_v(G_D)$.
One can verify that such $\delta'$ must exist.
We have $f(\delta')\geq 0$.
This is because otherwise $f(\delta')<0$ and then
$$ d((1 + \delta')\vec{q},G_D) < \frac{1}{545}\cdot r^+(G_D,(1 + \delta')\vec{q},C).$$ 
By following the proof of Lemma~\ref{lemma-beyond},
we have the MT algorithm terminates at $(1 + \delta')\vec{q}$,
which is contradictory with $(1 + \delta')\vec{q}\not \in \Interior_v(G_D)$.

%Obviously, $d(\vec{p},G_D)$ increases as $\delta$ increases.
%Similarly, we also have $\frac{1}{544}\cdot |C| \cdot \big(\min_{j\in C} p'_j\big)^4\cdot\left(\frac{2\sum_{j\in C} \sqrt{p'_j}}{|C|}  - 1,0\right)^2$ increases as $\delta$ increases.
Moreover, $f(\delta)$ is a continuous function of $\delta$, because $d((1 + \delta)\vec{q},G_D)$ and $r^+(G_D,(1 + \delta)\vec{q},C)$
are both continuous functions of $\delta$.
Combining with $f(0)<0$ and $f(\delta')>0$, we have 
there must be a $0 \leq \delta \leq \delta'$ such that $f(\delta)=0$.
Let $\vec{p} = (1+\delta)\vec{q}$.
By $f(\delta)=0$,
we have 
\begin{equation}\label{eq-d-ub}
d(\vec{p},G_D) = \frac{1}{545} \cdot r^+(G_D,\vec{p},C).
\end{equation}
Combining with $r^+(G_D,\vec{p},C) = r(G_D,\vec{p},C) > r(G_D,\vec{q},C)$ and (\ref{eq-lb-r}),
we have $d(\vec{p},G_D) >2^{-20} \ell^{-3}$.

\vspace{2ex}

Fix a variable-generated event system $\A\sim(G_D,\vec{p})$. Define $\delta_{i,i'}:=\Pr(A_i\cap A_{i'})$. Let $G_B$ denote the event-variable graph of $\A$. 
  Let $G$ denote the induced subgraph of $G_B$ on $C\cup \left(\cup_{i\in C}\Neighbor_{G_B}(i)\right)$. According to Remark \ref{remark:simplified}, it is lossless to assume that $G$ is a cycle of length $2|C|$. 
  Thus we have 
  \begin{equation}\label{eq-F-one-cycle}
  \lowerBB^+(G,\vec{p}) \geq \frac{\big(\min_{i\in C} p_i\big)^2\cdot\left(- |C| + \sum_{i\in L} 2\sqrt{p_i} \right)}{8\sqrt{|C|}}.
 \end{equation}
According to Theorem~\ref{thm:lower-bound-intersection}, there is a matching $\M$ of $G_D$ such that
$\sum_{(i,i')\in \M}\delta_{i,i'}^2 \geq \left(\lowerBB^+(G,\vec{p})\right)^2$.
Define $\vec{p}^-$ as (\ref{eq-pminus}).
We have 
\begin{equation*}
\begin{aligned}
||\vec{p} - \vec{p}^-||_1\geq \frac{2}{17}\sum_{(i,i')\in \M}\delta^2_{i,i'} \geq \frac{2}{17}\sum_{k\leq \ell}\left(\lowerBB^+(G_k,\vec{p})\right)^2.
\end{aligned}
\end{equation*}
Combining with (\ref{eq-F-one-cycle}), we have
\begin{equation*}
\begin{aligned}
||\vec{p} - \vec{p}^-||_1\geq \frac{1}{544}\cdot r^+(G_D,\vec{p},C).
\end{aligned}
\end{equation*}
Let 
$$\eps \triangleq \frac{1}{ 2^{29}\cdot \ell^3   \cdot m}.$$
By (\ref{eq-lb-r})
we have 
\begin{equation*}
\begin{aligned}
m\eps \leq \frac{1}{545 \cdot 544  \cdot 2^{10}\cdot \ell^3}  \leq \left(\frac{1}{544} -  \frac{1}{545}\right)  r^+(G_D,\vec{q},C) \leq \left(\frac{1}{544} -  \frac{1}{545}\right)  r^+(G_D,\vec{p},C).
\end{aligned}
\end{equation*}
Thus we have
\begin{equation*}
\begin{aligned}
||\vec{p} - (1+\eps)\vec{p}^-||_1> ||\vec{p} - \vec{p}^-||_1 - m\eps
\geq \frac{r^+(G_D,\vec{p},C)}{544} - m\eps \geq \frac{r^+(G_D,\vec{p},C)}{545} \geq d(\vec{p},G_D),
\end{aligned}
\end{equation*}
where the last inequality is by (\ref{eq-d-ub}).
Thus by Definition~\ref{def-distance}, 
we have $ (1+\eps)\vec{p}^-$ is in the Shearer's bound of $G_D$.
Combining with Theorem~\ref{thm:main}, we have the expected number of resampling steps performed by the MT algorithm is most $m/\eps$.
\end{proof}

%% file: sec5.tex
\section{Application to periodic Euclidean graphs}\label{sec:lattice}
In this section, %we show that the efficient region of MT algorithm can \emph{significantly} exceed Shearer's bound by 
we explicitly calculate the gaps between our new criterion and Shearer's bound on periodic Euclidean graphs, including several lattices that have been studied extensively in physics. It turns out the efficient region of MT algorithm can exceed \emph{significantly} beyond Shearer's bound.  

A periodic Euclidean graph $G_D$ is a graph that is embedded into a Euclidean space naturally and has a \emph{translational unit} $G_U$ in the sense that $G_D$ can be viewed as the union of periodic translations of $G_U$. For example, a cycle of length 4 is a translational unit of the square lattice. %periodic graphs. Thus, Theorem \ref{thm:gapbetweenlattices}

Given a dependency graph $\DependencyGraph$, it naturally defines a bipartite graph $\BipartiteGraph(\DependencyGraph)$ as follows. Regard each edge of $\DependencyGraph$ as a variable and each vertex as an event. An event $A$ depends on a variable $X$ if and only if the vertex corresponding to $A$ is an endpoint of the edge corresponding to $X$.

For simplicity, we only focus on symmetric probabilities, where $\vec{p}=(p,p,\cdots,p)$. 
Given a dependency graph $G_B$ and a vector $\vec{p}$, remember that $\vec{p}$ is on Shearer's boundary
of $G_D$ if $(1 - \eps)\vec{p}$ is in Shearer's bound
and $(1 + \eps)\vec{p}$ is not for any $\eps>0$.% Theorem \ref{thm:gapbetweenlattices} is an application of Theorem~\ref{thm-MT-beyond-Shearer} to periodic Euclidean graphs. Based on  Theorem \ref{thm:gapbetweenlattices}, gaps between our criterion and Shearer's bound on some lattices are calculated explicitly and summarized in Table~\ref{table1}.

%In this section, we show that the efficient region of MT algorithm can \emph{significantly} exceed Shearer's bound by applying our techniques (Theorem~\ref{thm:main} and Theorem \ref{thm:lower-bound-intersection}) to periodic Euclidean graphs. %Periodic Euclidean graphs have been extensively studied in natural science and engineering, particularly of three-dimensional crystal nets to crystal engineering, crystal prediction (design), and modeling crystal behavior \cite{senechal1990brief,cohen1990recognizing,james2006objective}.
%For simplicity, we focus on \emph{symmetric} probability vectors $\vec{p}=(p,p,\cdots,p)$.

Given a dependency graph $G_D = ([m],E_D)$ and two vertices $i,i'\in [m]$, we use  $\dist(i,i')$ to denote the distance between $i$ and $i'$ in $G_D$.
The following Lemma  will be used.

\begin{lemma}\label{lem:probtransfer}
Suppose $\vec{p}_a = (p_a,p_a,\cdots,p_a)$ is on Shearer's boundary of $\DependencyGraph=([m],E_D)$. 
For any probability vector $\vec{p}$ other than $\vec{p}_a$, it is in the Shearer's bound
if there exist $K,d \in \mathbb{N}^+$, $\mathcal{S}\subseteq 2^{[m]}$ where $\cup_{S\in\mathcal{S}}=[m]$, and
$f:\mathcal{S} \rightarrow 2^{[m]}$ such that the following conditions hold:
\begin{itemize}
\item[(a)] for each $i\in [m]$, there are at most $K$ subsets $S\in \mathcal{S}$ such that $f(S)\ni i$;
\item[(b)] if $f(S)=T$, then $\dist(i,i')\leq d$ for each $i\in S$ and $i'\in T$;
\item[(c)]   if $f(S)=T$, then
\[
\left(\frac{1-p_a}{p_a}\right)^{d-1}\cdot \frac{K}{p_a}\cdot\sum_{i\in S}\max\{p_i-p_a,0\} 
\leq \sum_{i\in T} \max\{p_a-p_i,0\}.
\]
\end{itemize}
\end{lemma}
While Lemma \ref{lem:probtransfer} looks involved, the basic idea is simple: by contradiction, suppose there is such a vector $\vec{p}'$ beyond Shearer's bound; then we apply Lemma \ref{lem:probtransferonpath} repeatedly to transfer probability from one event to another while keeping the probability vector still beyond Shearer's bound; finally, the vector $\vec{p}'$ will be changed to a vector strictly below $\vec{p}$, which makes a contradiction to the assumption that $\vec{p}$ is on the Shearer's boundary.
The involved part is a transferring scheme which changes $\vec{p}'$ to another probability vector strictly below $\vec{p}$. We leave the proof to the appendix.

The main result of this section is as follows. 
\begin{theorem}\label{thm:gapbetweenlattices}
	Let $G_D = (V_D,E_D)$ be a periodic Euclidean graph with maximum degree $\Delta$, and $\vec{p}_a=(p_a,\cdots,p_a)$ be the probability vector on Shearer's boundary of $G_D$.
	Suppose $G_U=(V_U,E_U)$ is a translational unit of $G_D$ with diameter $D$. Let
	\begin{align*}
	q\triangleq \frac{p_a^{D+2}\big(\lowerBB^+(G_B(G_U),\vec{p}_a)\big)^2}{17\cdot(\Delta+1) \cdot|V|^2\cdot(1-p_a)^{D+1}}.
	\end{align*}
	Then for any $\A\sim(G_B(G_D),\vec{p})$ where $(1+\eps)\vec{p} \leq (p_a+q,\cdots,p_a+q)$,
	the expected number of resampling steps performed by the MT algorithm is most $|V_D|/\eps$.
\end{theorem}

\begin{proof}
Fix any $\A\sim(G_B(G_D),\vec{p})$ where $(1+\eps)\vec{p}\leq (p_a+q,\cdots,p_a+q)$. Let $\delta_{v_0,v_1}$ denote $\Pr(A_{v_0}\cap A_{v_1})$ for  $(v_0,v_1)\in E_D$. We construct a matching $\M\subset E_D$ greedily as follows: we maintain two sets $E$ and $\M$,
which are initialized as $E_D$ and $\emptyset$ respectively. We do the following iteratively until $E$ becomes empty:
select a edge $(v_0,v_1)$ with maximum $\delta_{v_0,v_1}$ from $E$, add $(v_0,v_1)$ to $\M$, and delete all edges connecting $v_0$ or $v_1$ from $E$ (including $(v_0,v_1)$). Let $\vec{\delta}=\big(\delta_{v_0,v_1}:(v_0,v_1)\in\M\big)$. Then $\A\sim (G_B(G_D),\vec{p},\M,\vec{\delta})$.

Define $\vec{p}^-$ as~(\ref{eq-pminus}). In the remaining part of the proof, we will show that $(1+\eps)\vec{p}^-$ is in the Shearer's bound. This implies the conclusion immediately by Theorem \ref{thm:main}.

In fact, it is a direct application of Lemma \ref{lem:probtransfer} to show that $(1+\eps)\vec{p}^-$ is in the Shearer's bound.  
To provide more detail, we need some notations. We use $v,v',v_1,v_2,\cdots$ to represent vertices in $G_D$, and use $u,u',u_1,u_2,\cdots$ to represent vertices in $G_U$.
Let $G_U^1,G_U^2,\cdots$ be the periodic translations of $G_U$ in $G_D$. And we use a surjection\footnote{$h$ is possibly not a injection, as these translations are possibly overlapped with each other. } $h:\mathbb{N}^+\times V_U\rightarrow V_D$ to represent how these periodic translations constitute $G_D$: $h(k,u)=v$ if the copy of $u\in V_U$ in $k$-th translation (i.e., $G_U^k$) is $v\in V_D$. In particular, the vertex set of $G_U^k$ , denoted by $V_U^k$, is $\{h(k,u):u\in V\}$, and the edge set of $G_U^k$ , denoted by $E_U^k$, is $\{(h(k,u),h(k,u')):(u,u')\in E_U\}$. Besides,  let $\N^+(v):=\N_{G_D}(v)\cup \{v\}$ for $v\in V_D$. For $V\subset V_D$, let $\N^+(V):=\cup_{v\in V}\N^+(v)$.  Let $T_k:=\{(v_0,v_1)\in \M: v_0,v_1\in \N^+(G_U^k)\}$ stand for the pairs in $\M$ adjacent to $G_U^k$. 
With some abuse of notation, we sometimes use $v\in T_k$ to denote that $(v,v')\in T_k$ for some $v'\in V_D$.

The following claim says that $\vec{p}^-$ is much smaller than $\vec{p}$ even projected on a single translation. Its proof uses a similar idea to Theorem \ref{lemma-indep-pairs} and can be found in the appendix.
\begin{claim}\label{claim:51}
$\sum_{(v_0, v_1)\in T_k }\delta_{v_0,v_1}^2 \geq \big(\lowerBB^+(G_B(G_U),\vec{p})\big)^2$ holds for any $k$.
\end{claim}

To apply Lemma \ref{lem:probtransfer}, let $K:=(\Delta+1)|V_U|$, $d:=D+2$, $\mathcal{S}:=\{V_U^1,V_U^2,\cdots\}$, and $f(V_U^k):=T_k$. Based on Claim \ref{claim:51}, one can check that all the three conditions in Lemma \ref{lem:probtransfer} hold (see the appendix for details). Thus, according to Lemma \ref{lem:probtransfer}, $(1+\eps)\vec{p}^-$ is in Shearer's bound.
\end{proof}

%An important application of Theorem \ref{thm:gapbetweenlattices} is to lower bound the critical thresholds of
%efficiency of the MT algorithm on lattices. 
%The critical thresholds of abstract LLL and quantum LLL of many infinite lattices have been studied extensively in physics, including the research on hard-core singularity in the statistical mechanical literature \cite{heilmann1972theory,SYNGE1999TRANSFER} and that on local Hamiltonians in quantum physics \cite{movassagh2010unfrustrated,pnas}. Many of these lattices studied in the literature are periodic graphs. Thus we can lower bound the thresholds of the efficiency of the MT algorithm on these lattices with Theorem \ref{thm:gapbetweenlattices}. 
We apply Theorem \ref{thm:gapbetweenlattices} to three lattices: square lattice, Hexagonal lattice, and simple cubic lattice.
%we can obtain the lower bounds in Table~\ref{table1} immediately.
For square lattice, we take the $5\times5$ square with 25 vertices as the translational unit. 
%For triangular lattice, we take the triangular with 4 vertices on each side and 10 vertices in total as the translational unit. 
For Hexagonal lattice, we take a graph consisting of $19$ hexagons as the translational unit, in which 
there are 3,4,5,4,3 hexagons in the five columns, respectively. 
For simple cubic lattice, we take the $3\times 3 \times 3$ cube with 27 vertices as the translational unit. 
%By applying Theorem \ref{thm:gapbetweenlattices}, 
%we can obtain the lower bounds in Table~\ref{table1} immediately. 
The explicit gaps are summarized in  Table~\ref{table1}.
Finally, the lower bounds for these three lattices in Table~\ref{table1} hold for all bipartite graphs with the given canonical dependency graph,
because all such bipartite graphs are essentially the same under the reduction rules defined in~\cite{he2017variable}.

%The lower bounds in Table~\ref{table1} are not optimised. 
%They can be improved further by taking the specific structures of lattices into consideration.
%Our main point is that 
%the MT algorithm goes beyond Shearerer's bound on these lattices.

%% file: sec_appendix.tex
\appendix
%%%%%%%%%%%%%%%%%%%%%%%%%%%%%%%%%%%%%%%%%%%%%%%%%%%%%%%%%%%%%%%%%%%%%%%%%%%%%%%
\section{Missing Proofs in Section \ref{sec:preliminaries}}
\iffalse
\begin{proof}[Proof of Fact \ref{claim-exchangeable}]
 Let $D'=(V',E',L')$ denote the directed graph obtained from $D = (V,E,L)$ by reversing  direction of $u\rightarrow v$.

$\Longrightarrow:$ By contradiction, assume except the arc $u\rightarrow v$, there is another path $P$ from $u$ to $v$. Observe that $P$ and $v \rightarrow u$ forms a cycle in $D'$, which is a contradiction to that $u\rightarrow v$ is reversible in $D$.
	
$\Longleftarrow:$	Suppose $u\rightarrow v$ is not reversible in $D$. Then $D'$ is not a DAG and contains  a cycle. Because $D$ is a DAG and $E'\setminus E =\{v\rightarrow u\}$, the arc $v\rightarrow u$ should be on the cycle, which means that there is a path $P$ from $u$ to $v$ in $D'$. As $(u\rightarrow v)\notin E'$, the path $P$ also exist in $D$. In other words, $u\rightarrow v$ is not the unique path from $u$ to $v$ in $D$.
\end{proof}

\begin{proof}[Proof of Fact \ref{claim-exchangeable2}]
Suppose $u\rightarrow v$ is reversible in a wdag $D$ of $G_D=([m],E_D)$. Let $D'=(V',E',L')$ be the directed graph obtained from $D = (V,E,L)$ by reversing  direction of $u\rightarrow v$. By definition, $D'$ is a DAG, and we want to show that $D'$ is indeed a witness DAG of $G_D$. 

First, $D$ and $D'$ have the same nodes and labels. Second, since $D$ is a wdag, for any distinct  nodes $u,v$, there is a edge between $u$ and $v$ (in either direction) $D$ if and only if either $L(v) = L(u)$ or $(L(v),(L(u)) \in E_D$, which also holds for $D'$ by the definition of $E'$. So $D'$ is a wdag.
\end{proof}
\fi

\begin{proof}[Proof of Proposition \ref{prop:27}] The following simple greedy procedure will find such a $\mathcal{P}$.
\vspace{-2.5ex}
	\begin{algorithm}[h]
		Initially, $\mathcal{P}=\emptyset$\;
		\For{ each $(i,i')\in\M$}{
			\For{each $k$ from $1$ to $|\mathrm{List}(D,i,i')|-1$}{
				\If{the $k$-th node and $(k+1)$-th node in $\mathrm{List}(D,i,i')$ form a reversible arc}{add this arc to $\mathcal{P}$, and $k:=k+2$;\\}
				\Else{$k:=k+1$;\\}
				}
		 
		     }
		Return $\mathcal{P}$;
	\end{algorithm}
	\vspace{-2.5ex}

Obviously, for each $(i,i')\in\M$, the procedure contains at least half of all reversible arcs $u\rightarrow v$ where $\{L(u),L(v)\}=\{i,i'\}$, hence at least half of nodes in $\mathcal{V}(D,i)$.
\end{proof}

\section{Proof of Proposition \ref{obs-mapping-dag}}
	Given a pwdag $D=(V,E,L)$ of $G_D$ and a Boolean string $\vec{R}\in \{0,1\}^{\Mat(D)}$,
	define $h(D,\vec{R})$ to be a directed graph $D':=(V',E',L')$ where $V'=V$, $E'=E$, and 
	\begin{equation*}
	\begin{aligned}
	\forall v\in V:\quad L'(v) =
	\begin{cases}
	(L(v))^{\uparrow}, & \text{if }  v\in \Mat \text{ and } R_v = 0;\\
	(L(v))^{\downarrow},& \text{if }  v\in \Mat \text{ and } R_v = 1;\\
	L(v),& \text{otherwise, } v\not \in \Mat.
	\end{cases}	
	\end{aligned}
	\end{equation*} 

	It is easy to verify that $h(D,\vec{R})$ is a pwdag of $G^\M$.
	Moreover, given any $D'\in \D(G^{\M})$, there is one and only one $D\in\mathcal{D}(G_D)$ and $\vec{R}\in\{0,1\}^{\Mat(D)}$ such that $h(D,\vec{R})=D'$. In other words, $h$ is a bijection between $\{(D,\vec{R}):D\in \mathcal{D}(G_D),\vec{R}\in \{0,1\}^{\Mat(D)}\}$ and $\mathcal{D}(G^\M)$.  So 
	\begin{align*}
	\sum_{D'\in \D(G^{\M})}\prod_{v' \text{ in } D'}p^{\M}_{L'(v')}
	&=  \sum_{D\in \D(G_D)}\sum_{\vec{R}\in \{0,1\}^{\Mat(D)}}\prod_{v' \text{ in } h(D,\vec{R})}p^{\M}_{L'(v')}
	\\&= \sum_{D\in \D(G_D)}
	\sum_{\vec{R}\in \{0,1\}^{\Mat(D)}}\prod_{v \text{ in } D}p^{\M}_{L'(v)}
	\\&= \sum_{D\in \D(G_D)}\prod_{v \not \in \Mat(D)}p^{\M}_{L'(v)}
	\left(\sum_{\vec{R}\in \{0,1\}^{\Mat(D)}}\prod_{v  \in \Mat(D)}p^{\M}_{L'(v)}\right)
	\\&= \sum_{D\in \D(G_D)}\prod_{v \not \in \Mat(D)}p^{\M}_{L(v)}
	\prod_{v \in \Mat(D)}\left(p^{\M}_{L(v)^\uparrow}  + p^{\M}_{L(v)^\downarrow} \right)
	\\&= \sum_{D\in \D(G_D)}\prod_{v \not \in \Mat(D)}p^{-}_{L(v)}
	\prod_{v \in \Mat(D)}\left(p'_{L(v)}  + p^-_{L(v)} - p'_{L(v)} \right)
	\\&= \sum_{D\in \D(G_D)}\prod_{v \text{ in }D}p^{-}_{L(v)},
	\end{align*}
	where the second equality is by that $V = V'$, the forth equality is by the definition of $L'$, and the fifth equality is by the definition of $\vec{p}^{\M}$.

%%%%%%%%%%%%%%%%%%%%%%%%%%%%%%%%%%%%%%%%%%%%%%%%%%%%%%%%%%%%%%%%%%%%%%%%%%%%%%%%%%%%%%%%%%%
\section{Proof of Theorem \ref{thm-injection}}
We first verify that the image of $h$ is a subset of $\D(G^{\M})$.
\begin{lemma}\label{lemma-is-dag}
	For any $D\in \D(G_D)$ and $\Rem\in \psi(D)$, $h(D,\Rem)\in\D(G^{\M})$.
\end{lemma}
\begin{proof}
First, we prove that $h(D,\Rem) = (V',E',L')$ is a DAG. Define a total order $\pi'$ over the set $V'$ as follows:
for any two distinct nodes $u',v'\in V'$,
\begin{itemize}
\item if $g(u')\neq g(v')$, then $u' \prec v'$ in $\pi'$ if and only if $g(u')\prec g(v')$ in $\pi_D$;
\item if $g(u') = g(v')$, then $u'  \prec v'$ in $\pi'$ if and only if $u'=f^\ast(g(u'))$ (and then $v'=f(g(u'))$).
\end{itemize}
One can verify that $\pi'$ is a topological order of $h(D,\Rem)$, which means that $h(D,\Rem)$ is a DAG.
	
Secondly, we prove that $h(D,\Rem)$ is a wdag of $G^{\M}$. As $h(D,\Rem)$ has been shown to be a DAG, we only need to verify that: for any two distinct nodes $u',v'$ in $D'$, 
there is a arc between $u'$ and $v'$ (in either direction)	if and only if either $L'(u')=L'(v')$ or $(L'(v'),L'(u')) \in E^{\M}$.

\noindent$\Longrightarrow$: By symmetry, suppose $(u' \rightarrow v') \in E'$. If $(u' \rightarrow v') \in E'_1$, then $u' = f^\ast(w)$ and $v' = f(w)$ for some vertex $w \in \Rem_3\cup\Rem_4$.
	Thus, by (\ref{eq-f_v}) and (\ref{eq-f_ast_v}) we have 		
	$L'(u') \in \{i^{\uparrow},i^{\downarrow}\}$ and 
	$L'(v')=L(w)^{\downarrow}$ where 
	$(L(w),i)\in \M$.
	By $(L(w),i)\in \M$,
	any two vertices in $\{(L(w))^{\uparrow},(L(w))^{\downarrow},i^{\uparrow},i^{\downarrow}\}$ are connected in $G^{\M}$. In particular, $(L'(v'),L'(u')) \in E^{\M}$.
	If $(u' \rightarrow v') \in E'_2$, we have 
	$L'(u')=L'(v')$ or $(L'(u'),L'(v')) \in E^{\M}$
	immediately.

\noindent$\Longleftarrow$: Suppose $u',v'\in V'$ are two distinct nodes where $L'(u')=L'(v')$ or $(L'(u'),L'(v')) \in E^{\M}$. If $g(u')\not = g(v')$, then either 
	$g(u') \prec g(v')$ or $g(v') \prec g(u')$ in $\pi_D$, which implies that
	either $(u'\rightarrow v') \in E'_2$ or $(v'\rightarrow u') \in E'_2$.
	Otherwise, $g(u')=g(v')$. Let $v:=g(u')=g(v')$.
	By (\ref{eq-f_v}) and (\ref{eq-f_ast_v}), we have 
	$v\in \Rem_3 \cup \Rem_4$ and $\{u',v'\} = \{f(v),f^\ast(v)\}$.
	Therefore either $u'\rightarrow v'$ or 
	$v'\rightarrow u'$ is in  $E'_1$.	

Finally, one can check that $f(v)$ where $v$ is the unique sink of $D$ is the unique sink of $D'$. This completes the proof.		  
\end{proof}

In the rest of this section, we show that $h$ is injective.
	Given $D\in\D(G_D)$ and $(i,j) \in \M$,
	recall that $\mathrm{List}(D,i,j)$ is the sequence listing all nodes in $D'$ labelled with $i$ or $j$ in the topological order. %the $\ell$-th node in $\mathrm{List}(D,i,j)$. 
	Similarly, 
	\begin{definition}
		Given $D' = (V',E',L')\in\D(G^\M)$ and $(i,j) \in \M$, we use $\mathrm{List'}(D',i,j)$ to denote the unique sequence listing all nodes in $D'$ with label in $\{i^{\uparrow},i^{\downarrow},j^{\uparrow},j^{\downarrow}\}$ in the topological order.
%	Let $v'_r(D',i,j)$ denote the $r$-th node in $\mathrm{List'}(D',i,j)$.
	
%	For simplicity of notations, we sometimes use $v_\ell$ and $v_r'$ instead of $v_\ell(D,i,j)$ and $v'_r(D',i,j)$ respectively if $D,D',i$, and $j$ are clear from the context.
\end{definition}

Claims \ref{claim-order-v_i} and \ref{claim-hdr-ex} are two properties about $\mathrm{List}'(D',i,j)$, which will be used to show the injectiveness of $h$.
\begin{claim}\label{claim-order-v_i}
	Suppose $D'=h(D,\Rem)$ for some $D\in\D(G_D)$ and $\Rem\in\psi(D)$. Let $(i,j)\in \M$. Then for any node $v'$ in $D'$,
	\begin{itemize}
	\item[(a)]  $v'\in \mathrm{List}'(D',i,j)$ if and only if $g(v')\in\mathrm{List}(D,i,j)$;
	\item[(b)] for any other node $u'$ in $D'$, if $g(u')$ precedes $g(v')$ in $\mathrm{List}(D,i,j)$, then $u'$ precedes $v'$ in $\mathrm{List}'(D',i,j)$;
	\item[(c)] if $v\in\Rem_3\cup\Rem_4$, then $f(v)$ is next to $f^\ast(v)$ in $\mathrm{List}'(D',i,j)$.
	\end{itemize}
%	$\ell \in \mathbb{N}^+$, 
%	we have $f(v_{\ell}) = v'_r$ where $r = r(\ell,\Rem):= \ell + |\{1 \leq k\leq\ell: v_k \in \Rem_3\cup \Rem_4\}|$.
\end{claim}
\begin{proof}
%	Fix $(i,j)\in\M$ and $t\in\mathbb{N}^+$. 
%We use $\mathrm{List}$ and $\mathrm{List}'$ for abbreviation of $\mathrm{List}(D,i,j)$ and $\mathrm{List}(D,i,j)$ respectively. 
Part (a) is immediate by Definition \ref{def:h}.

 Now, we show Part (b). Suppose $g(u')$ precedes $g(v')$ in $\mathrm{List}(D,i,j)$. Then $g(u')\prec g(v')$ in $\pi_D$.
Thus one can check that all the four arcs $f(g(u'))\rightarrow f(g(v'))$, $f^\ast(g(u'))\rightarrow f(g(v'))$, $f(g(u'))\rightarrow f^\ast(g(v'))$, and $f^\ast(g(u'))\rightarrow f^\ast(g(v'))$ are contained in $E_2'$. In particular, $(u'\rightarrow v')\in E'$ as  $u'\in\{f(g(u')),f^\ast(g(u'))\}$ and $v'\in\{f(g(v')),f^\ast(g(v'))\}$. This implies that $u'$ precedes $v'$ in $\mathrm{List}'(D',i,j)$.

Finally, we prove Part (c). According to Part (b), $f(v)$ and  $f^\ast(v)$ are adjacent in $\mathrm{List}'(D',i,j)$. Besides, as there is an arc $f^\ast(v)\rightarrow f(v)$ in $E_1'$, we conclude that $f(v)$ is next to $f^\ast(v)$ in $\mathrm{List}'(D',i,j)$.

% Note that for a node $v'$ in $D'$, $v'\in \mathrm{List}$ if and only if $g(v')\in\mathrm{List}'$. In particular, $f(v_{\ell})$ is in $\mathrm{List}'$. Moreover, we will show that: for a node $v'\in\mathrm{List'}\neq f(v_{\ell})$, $f(v_{\ell})$ precedes $v'$ in $\mathrm{List'}$ if and only if $v_{\ell}$ precedes $g(v')$ in $\mathrm{List}$. This implies the conclusion immediately.	
%	\vspace{1ex}

%	$\Longleftarrow:$  	Suppose $v_{\ell}$ preceds $g(v')$ in $\mathrm{List}$. Then $v_{\ell}\prec g(v')$ in $\pi_D$.
%	Thus can check that both $f(v_{\ell})\rightarrow f(g(v'))$ and $f(v_{\ell})\rightarrow f^\ast(g(v'))$ are contained in $E_2'$. In particular, $(f(v_{\ell})\rightarrow v')\in E'$ as  $v'=f(g(v'))$ or $f^\ast(g(v'))$.
	
%	\vspace{1ex}
%	$\Longleftarrow:$  	Suppose $v_{\ell}$ does not preced $g(v')$ in $\mathrm{List}$. Then either $g(v')$ preceds $v_{\ell}$ in $\mathrm{List}$ or $g(v')=v_{\ell}$. 
%	In the first case, just as shown above, we have that $v'$ precedes $f(v_{\ell})$ in $\mathrm{List}'$. In the latter case, since $v'\neq f(v_{\ell})$, $v'=f^\ast(v_{\ell})$ and $v'$ precedes $f(v_{\ell})$ as well. 
\end{proof}

\begin{definition}
	For a reversible arc $u'\rightarrow v'$ in $D'$, we call it $(*,\downarrow)$-reversible in $D'$ if  $L'(u')\in\{i^\uparrow,i^\downarrow\}$ and $L'(v')=j^\downarrow$ for some $(i,j)\in E_D$.
\end{definition}
\begin{claim}\label{claim-hdr-ex}
	Suppose $D'=h(D,\Rem)$ for some $D\in\D(G_D)$ and $\Rem\in\psi(D)$. Let $(i,j)\in\M$. Let $u',v'$ be two nodes in $\mathrm{List'}(D',i,j)$ where $v'$ is next to $u'$. Then $u'\in V_2'$ if and only if $u'\rightarrow v'$ is $(*,\downarrow)$-reversible in $D'$ and $v'\in V_1'$.
\end{claim}
\begin{proof}
$\Longrightarrow$: Let $u:=g(u')$. Assume $u'\in V_2'$, i.e., $u'=f^\ast(u)$. By Definition \ref{def:h}, $u\in\Rem_3\cup\Rem_4$. According to Part (c) of Claim \ref{claim-order-v_i}, as $v'$ is next to $u'$,  we have $v'=f(u)$ and then $v'\in V_1'$. 

Now we show that $u'\rightarrow v'$ is $(*,\downarrow)$-reversible. First, by Definition \ref{def:h}, either $L'(u')\in\{i^\uparrow,i^\downarrow\}$ and $L'(v')=j^\downarrow$, or  $L'(u')\in\{j^\uparrow,j^\downarrow\}$ and $L'(v')=i^\downarrow$. 
	What remains is to show $u'\rightarrow v'$ is reversible, by Fact~\ref{claim-exchangeable} which is equivalent to show that $f^\ast(u) \rightarrow f(u)$ is the unique path from $u'$ to $v'$ in $D'$. 
By contradiction, assume that
there is a path $f^\ast(u) \rightarrow w'_1 \rightarrow \cdots \rightarrow w'_k
	\rightarrow f(u)$ in $D'$ where $w'_1\neq f(u)$ and $w'_k\neq f^\ast(u)$.
	As $w'_1\neq f(u)$, we have $(f^\ast(u) \rightarrow w'_1)$ is not in $E'_1$ and then should be in $E'_2$, which further implies that $u \prec g(w'_1)$ in $\pi_D$.
	Similarly, we have $g(w'_k) \prec u$ in $\pi_D$. So $g(w'_k) \prec u\prec g(w'_1)$.
	Meanwhile, for each $\ell<k$, if $(w'_\ell \rightarrow w'_{\ell+1}) \in E'_1$, then $g(w'_\ell) = g(w'_{\ell+1})$;
	if $(w'_\ell \rightarrow w'_{\ell+1}) \in E'_2$, then  $g(w'_\ell) \prec g(w'_{\ell+1})$ in $\pi_D$.
	So, it always holds that $g(w'_\ell) \preccurlyeq g(w'_{\ell+1})$ in $\pi_D$ for each $\ell<k$. In particular, $g(w_1')\preccurlyeq g(w'_k)$. A contradiction.
	
\vspace{1ex}
$\Longleftarrow$: Let $u:=g(u')$ and $v:=g(v')$. Assume $u'\notin V_2'$ and $v'\in V_1'$, i.e., $u'=f(u)$ and $v'=f(v)$. Furthermore, assume $L'(v')=j^\downarrow$, then $v\notin\Rem_1$ and $L(v)=j$. 
We will show that $(f(u)\rightarrow f(v))$ is not reversible.

Note that $(f(u)\rightarrow f(v))$ should be in $E'_2$ and then $u \prec v$ in $\pi_D$.
By $L'(u') \in \{i^{\uparrow},i^{\downarrow}\}$,
$u' = f(u)$, and~(\ref{eq-f_v}),
we have
$L(u) = i$.
Thus, $(L(u),L(v)) = (i,j) \in \M \subseteq E_D$.
As $D$ is a wdag and $u \prec v$ in $\pi_D$, the arc $(u\rightarrow v)$ exists in $D$.
Since $v\notin \Rem_1$, $v \notin \mathcal{V}$, which means that $u\rightarrow v$ is not reversible in $D$. %Furthermore, 
%noting that $(L(u),L(v)) \in \M$, we have 
%$u\rightarrow v$ is not even reversible in $D$.
According to Fact~\ref{claim-exchangeable}, there is a path $u = w_1\rightarrow w_2\rightarrow \cdots \rightarrow w_{k} \rightarrow w_{k+1} = v$ from $u$ to $v$ in $D$ other than the arc $u\rightarrow v$, where  $w_{\ell} \prec w_{\ell+1}$ in $\pi_D$ and $\left(L(w_\ell) = L(w_{\ell+1})\right) \lor \left((L(w_\ell),L(w_{\ell+1})) \in E_D\right)$ for each $\ell\in [k]$.

According to the definition of $G^{\M}$ and (\ref{eq-f_v}), one can check that $(f(w_i) \rightarrow f(w_{i+1})) \in E'_2$. Therefore $u' = f(w_1)\rightarrow f(w_2)\rightarrow \cdots \rightarrow f(w_{k}) \rightarrow f(w_{k+1}) = v'$ is a path from $u'$ to $v'$ in $D'$, which implies that $u' \rightarrow f(v)$ is not reversible in $D'$ by Fact~\ref{claim-exchangeable}.
\end{proof}

 	Having Claims \ref{claim-order-v_i} and \ref{claim-hdr-ex}, we are ready to show that $h$ is injective. 
 	\begin{lemma}\label{lemma-unique-D}
 		$h$ is injective.
 		\end{lemma}
 
    \begin{proof}
    Fix a $D=(V,E,L)\in \D(G_D)$ and a $\Rem\in\psi(D)$. Let $D'=(V',E',L')$ denote $h(D,\Rem)$. We show $(D,\Rem)$ can be recovered from $D'$, which implies the injectiveness of $h$.
 
 First, we recover the partition $(V_1',V_2')$. That is, given a node $u'\in V'$, we distinguish whether $u'\in V_1'$ or $u'\in V_2'$. If $L'(u') \in [m]\setminus \M$, then $u'\in V_1'$ according to (\ref{eq-f_v}). Otherwise, we have $L'(u')\in \{i^\uparrow,i^\downarrow\}$ for some $(i,j)\in\M$, hence $u'$ is in $\mathrm{List}'(D',i,j)$.  %We abbreviate $\mathrm{List}'(D',i,j)$ to $\mathrm{List}$. 
 Assume the nodes in $\mathrm{List}'(D,i,j)$ are $v_1'v_2'v_3'\cdots v_{k}'$. According to Claim \ref{claim-hdr-ex}, we can see that the following procedure distinguishes whether $v_\ell'\in V_1'$ or $v_k'\in V_2'$ for all $v_\ell'\in\mathrm{List}'(D',i,j)$, including $u'$.
% \vspace{-2.5ex}
 \begin{algorithm}[h]
 	Initially, mark that $v_{k}'\in V_1'$, and let $\ell:=k-1$;\\
  	\While{ $\ell\geq 1$}
  	{
  	  \eIf{the arc $(v_{\ell}'\rightarrow v_{\ell+1}')$ is $(*,\downarrow)$-reversible and $v_{\ell+1}'\in V_1'$}{
  	  	Mark that $v_{\ell}'\in V_2'$;\\
  	  }{
      Mark that $v_{\ell}'\in V_1'$;\\
      }	
     $\ell:=\ell-1$;
  	}
  \end{algorithm}
%  \vspace{-2.5ex}

Secondly, we can easily recover $D=(V,E,L)$ from $D'$ and $(V_1',V_2')$. Ignoring labels, it is easy to see that $D$ is exactly the induced subgraph of $D'$ on $V_1'$. By the way, we also get the function $f:V\rightarrow V_1'$.   For labels, we simply replace each label $i^\uparrow$ or $i^\downarrow$ with $i$.

Finally, we recover $\Rem$ from $D'$, $D$ and $(V_1',V_2)$. That is, we distinguish which one of $\{\Rem_1,\Rem_2,\Rem_3,\Rem_3\}$ contains a given node $v\in \Mat(D)$. Assume $L(v)=i$ and $(i,j)\in \M$. Let $u'$ be the node previous to $f(v)$ in $\mathrm{List'}(D,i,j)$. According to Part (c) of Claim \ref{claim-order-v_i}, $u'\in V_2'$ if and only if $v\in \Rem_3\cup \Rem_4$. When $v\in \Rem_3\cup \Rem_4$, $v\in \Rem_3$ if $L'(u')=j^\uparrow$, and $v\in \Rem_4$ if $L'(u')=j^\downarrow$. When $v\notin \Rem_3\cup \Rem_4$, $v\in \Rem_1$ if $L'(v')=i^\uparrow$, and $v\in \Rem_2$ if $L'(v')=i^\downarrow$.

\end{proof}

%%%%%%%%%%%%%%%%%%%%%%%%%%%%%%%%%%%%%%%%%%%%%%%%%%%%%%%%%%%%%%%%%%%%%%%%%%%%%%%%%%%%%%%%%%%
\section{Proof of Lemma \ref{lem:probtransfer}}
Let $\vec{e_i}$ denote the vector whose coordinates are all 0 except the $i$-th that equals 1.
The following lemmas will be used in the proof.

\begin{lemma}\cite{stoc19}\label{lem:probtransferonpath}
	Let $\DependencyGraph=([m],E_D)$ be a dependency graph  and  $\vec{p}$ be a probability vector beyond the Shearer's bound. Suppose $i,i_1,i_2,\cdots,i_{k-1},i'$ form a shortest path from $i$ to $i'$ in $G_D$. Then for any $q \leq p_{i'}$, $\vec{p} - q\vec{e_{i'}}+\big( \prod_{\ell \in [k-1]}\frac{1-p_{i_{\ell}}}{p_{i_{\ell}}}\big)\cdot\frac{1-p_i}{p_{i'}}\cdot q\vec{e_i}$ is also beyond the Shearer's bound.
\end{lemma}

	Without loss of generality, we assume that $p_i - p_a$ is rational for each $i\in [m]$. 
	By contradiction, let $\vec{p}$ be such a vector which is beyond Shearer's bound.
	Let $S_{+}:=\{i\in [m]:p_i >  p\}$ and $S_{-}:=\{i\in [m]:p_i <  p\}$.
	Let $\Delta_p$ be a real number such that the following hold:
	\begin{itemize}[leftmargin=16pt]
		\item For each $i\in S_+$, 	$p_i - p_a = \gamma_i\cdot \Delta_p$ for some $\gamma_i\in \mathbb{N}^+$. Intuitively, we cut $p_i-p_a$ into $\gamma_i$  pieces each of size $\Delta_p$. Besides, we call such pieces \emph{positive pieces}.
		\item For each $i\in S_-$, 
			$$p_a - p_i= \tau_i \cdot K \cdot\left(\frac{1-p_a}{p_a}\right)^{d-1}\cdot \frac{\Delta_p}{p_a}$$ for some $\tau_i\in\mathbb{N}^+$.
		Intuitively, we cut $p_a-p_i$ into $\tau_i\cdot K$  pieces each of size $\left(\frac{1-p_a}{p_a}\right)^{d-1}\cdot \frac{\Delta_p}{p_a}$. We call such pieces \emph{negative pieces}.
	\end{itemize}
We use
$\mathcal{R} :=\{(i,r):i\in S_+,r\in [\gamma_i]\}$ and  $\mathcal{T}\triangleq  \{(i',t,k):i'\in S_-,t\in [\tau_{i'}],k\in [K]\}$
to denote the set of positive pieces and negative pieces respectively.

For convenience, let $\gamma_i = 0$ if $i\not \in S_{+}$, and $\tau_i = 0$ if $i\not \in S_{-}$. 
Then Condition (c) can be restated as: for $f(S)=T$, the  positive pieces in $S$ are no more than the negative pieces in $T$, i.e., 
	\begin{align}\label{eq-gamma-tau}
	\sum_{i\in S}\gamma_{i} \leq \sum_{i'\in T} \tau_{i'}.
	\end{align}
	
The basic idea of Lemma \ref{lem:probtransfer} is relatively simple: for each $S\in\mathcal{S}$, we move positive pieces in $S$ to $f(S)$ such that (i) all the positive pieces in $S$ are absorbed by the negative pieces in $f(S)$ and (ii) the resulted probability vector is still beyond Shearer's bound. Finally, all positive pieces will be absorbed, and we will get a vector strictly smaller than $\vec{p}$. By Lemma \ref{lem:probtransferonpath}, this vector is beyond Shearer's bound, which makes a contradiction.

For $i'\in [m]$, remember Condition (a) which says that there are at most $K$ subsets $S\subset\mathcal{S}$ such that $i'\in f(S)$, and we use $S^1_{i'},S^2_{i'},\cdots$ to represent these subsets.
Let $g:\mathcal{R}\rightarrow \mathcal{T}$ be a injection mapping each $(i,r)\in \mathcal{R}$ to some $(i',t,k)\in \mathcal{T}$ satisfying that (i) $i\in S^k_{i'}$ and (ii)
	$$\sum_{i_0\in S^k_{i'},i_0< i}\gamma_{i_0}+r =  \sum_{i_1\in f(S^k_{i'}),i_1< i'}\tau_{i_1}+t.$$
	By (\ref{eq-gamma-tau}), one can verify that such mapping $g$ exists.
In addition, according to Condition (b), if $g(i,r) = (i',t,k)$, then $\dist(i,i')\leq d$.

In the following, we will apply Lemma \ref{lem:probtransferonpath} repeatedly. 
	%Moreover, for any $j\in S_{-}$, $|(i,r):g(i,r) = (j,t,k)|\leq T_j\cdot K_j$.

	%Given the injection $g$, we define another injection $g'$ as follows.
	Let $g_0$ be $g$, $S_0$ be $S_-$ and $\mathcal{R}_0$ be $\mathcal{R}$.
	Given an injection $g_{\kappa}:\mathcal{R}\rightarrow \mathcal{T}$, $S_{\kappa}$ and $\mathcal{R}_{\kappa}$
	where $\text{dis}(i,j)\leq d$ if $g_{\kappa}(i,r) = (j,t,k)$,
	we construct another injection $g_{\kappa+1}:\mathcal{R}\rightarrow \mathcal{T}$, $S_{\kappa+1}$ and $\mathcal{R}_{\kappa+1}$ as follows.
	There are two possible cases for $g_{\kappa}$, $S_{\kappa}$ and $\mathcal{R}_{\kappa}$.
	\begin{itemize}
		\item[(1)] there exists $i,r,j,t,k$ such that $(i,r)\in \mathcal{R}_{\kappa}$, $g_{\kappa}(i,r)=(j,t,k)$
		and there is a shortest path between $i$ and $j$ 
		such that no vertex in $S_{\kappa}$ is on the path;
		\item[(2)] For each $g_{\kappa}(i,r)=(j,t,k)$ where $(i,r)\in \mathcal{R}_{\kappa}$
		and each shortest path between $i$ and $j$,
		there is a vertex in $ S_{\kappa}$ on the path.
	\end{itemize}
	
	For case (1), we let $g_{\kappa+1} = g_{\kappa}$, $\mathcal{R}_{\kappa+1} = \mathcal{R}_{\kappa}\setminus \{(i,r)\}$, and 
	$$S_{\kappa+1} = \{j\in S_-: \text{there exists } i,r,t,k \text{ where }(i,r)\in \mathcal{R}_{\kappa+1} 
	\text{ such that } g_{\kappa+1}(i,r) = (j,t,k) \}.$$
	For case (2), 
	there must be $(i_1,r_1,j_1,t_1,k_1),\cdots,(i_n,r_n,j_n,t_n,k_n)$ for some $n\in \mathbb{N}^+$
	such that 
	\begin{itemize}
		\item[-] $(i_{\ell},r_{\ell})\in \mathcal{R}_{\kappa}$, $j_{\ell}\in S_{\kappa}$,
		$g_{\kappa}(i_{\ell},r_{\ell})=(j_{\ell},t_{\ell},k_{\ell})$ for each $\ell \in [n]$,
		\item[-] $j_{\ell+1}$ is on a shortest path between $i_{\ell}$ and $j_{\ell}$ for each $\ell \in [n-1]$,
		\item[-] $j_{1}$ is on a shortest path between $i_{n}$ and $j_{n}$.
	\end{itemize}
	We define the injection $F(g_{\kappa})$ as follows.
	\begin{align*}
	\begin{cases}
	F(g_{\kappa})(i_{n},r_{n})  &= (j_{1},t_{1},k_{1}),\\
	F(g_{\kappa})(i_{\ell},r_{\ell}) &= (j_{\ell+1},t_{\ell+1},k_{\ell+1}) \text{ for each } \ell \in [n-1],\\
	F(g_{\kappa})(i,r) &= g_{\kappa}(i,r)  \text{ for other } (i,r).
	\end{cases}
	\end{align*}
	
	%Let $F(g_{\kappa})(i_{\ell},r_{\ell}) = (j_{\ell+1},t_{\ell+1},k_{\ell+1})$ for each $\ell \in [n-1]$,
	%let $F(g_{\kappa})(i_{n},r_{n}) = (j_{1},t_{1},k_{1})$,
	%and $F(g_{\kappa})(i,r) = g_{\kappa}(i,r)$ for other $(i,r)$.
	One can verify that $\text{dis}(i,j)\leq d$ if $F(g_{\kappa})(i,r) = (j,t,k)$
	and 
	\begin{align*}
	N \triangleq \sum_{\substack{(i,r,j,t,k):\\   g_{\kappa}(i,r) = (j,t,k)}}\text{dis}(i,j) \geq 
	1 + \sum_{\substack{(i,r,j,t,k):\\   F(g_{\kappa})(i,r) = (j,t,k)}}\text{dis}(i,j) .
	\end{align*}
	%\begin{align}\label{eq-dist-dec}
	%\sum_{(i,r,j,t,k):F(h)(i,r) = (j,t,k)}\text{dis}(i,j) \leq \sum_{(i,r,j,t,k):h(i,r) = (j,t,k)}\text{dis}(i,j) - 1.
	%\end{align}
	Since $N$ is bounded, 
	there must be a constant $\ell \leq N$ and $i,r,j,t,k$
	such that $(i,r)\in \mathcal{R}_{\kappa}$, $F^{\ell}(g_{\kappa})(i,r)=(j,t,k)$
	and there is a shortest path between $i$ and $j$ 
	such that no vertex in $S_{\kappa}$ is on the path. 
	Let $g_{\kappa+1} = F^{\ell}(g_{\kappa})$, $\mathcal{R}_{\kappa+1} = \mathcal{R}_{\kappa}\setminus \{(i,r)\}$ and 
	$$S_{\kappa+1} = \{j\in S_-: \text{there exists } i,r,t,k \text{ where }(i,r)\in \mathcal{R}_{\kappa+1} 
	\text{ such that } g_{\kappa+1}(i,r) = (j,t,k) \}.$$
	One can verify that in both cases, 
	$g_{\kappa+1}$ is an injection from $\mathcal{R}$ to $\mathcal{T}$
	and $\text{dis}(i,j)\leq d$ if $g_{\kappa+1}(i,r) = (j,t,k)$.
	
	Let $g'$ be $g_{|\mathcal{R}|}$.
	For each $\ell\in [|\mathcal{R}|]$,
	let $(i_{\ell},r_{\ell})$ be the unique element in $\mathcal{R}_{\ell-1}\setminus \mathcal{R}_{\ell}$.
	Let $(j_{\ell},t_{\ell},k_{\ell})$ denote $g'(i_{\ell},r_{\ell})$.
	Thus, we have 
	\begin{itemize}
		\item[-] $g'$ is an injection from $\mathcal{R}$ to $\mathcal{T}$,
		\item[-] $\text{dis}(i_{\ell},j_{\ell})\leq d$ for each $\ell\in [|\mathcal{R}|]$,
		\item[-] there is a shortest path between $i_{\ell}$ and $j_{\ell}$
		such that $j_{\ell+1},j_{\ell+2},\cdots,j_{|\mathcal{R}|}\in S_{\ell}$ are not on the path.
	\end{itemize}
	For each $j\in S_{-}$, define
	\begin{align*}
	\eta_j = |\{(i,r):g'(i,r) = (j,t,k) \text{ for some } t\in [\tau_j],k\in [K]\}|.
	\end{align*}
	Because $g'$ is an injection, we have $\eta_j \leq \tau_j\cdot K$.
	Let 
	$$\vec{p''} \triangleq \vec{p}' + \sum_{j\in S_-} (K \cdot \tau_j - \eta_j)\cdot\left(\frac{1-p}{p}\right)^{d-1}\cdot \frac{\Delta_p}{p}\cdot \vec{e}_{j}.$$
	By $\vec{p}'$ is beyond Shearer's bound and $\eta_j \leq K\cdot \tau_j $ for each $j\in S_-$,
	we have $\vec{p}''$ is also beyond Shearer's bound.
	For each $\ell \in [0, |\mathcal{R}|]$,
	let 
	\begin{align*}
	\vec{p}_{\ell} \triangleq \vec{p}'' - \Delta_p\cdot\left( \sum_{\kappa \leq \ell - 1}\left( \vec{e}_{i_{\kappa}}-\left(\frac{1-p}{p}\right)^{d-1}\cdot\frac{1}{p}\cdot\vec{e}_{j_{\kappa}}\right) + \vec{e}_{i_{\ell}}-\left(\frac{1-p}{p}\right)^{d-1}\cdot\frac{1}{p + \Delta_p}\cdot\vec{e}_{j_{\ell}}\right).
	\end{align*}
	Then we have the following claim.
	\begin{claim}
		For $\ell \in [0, |\mathcal{R}|]$, $\vec{p}_{\ell}$ is beyond Shearer's bound.
	\end{claim}
	\begin{proof}
		We prove this claim by induction. 
		Obviously, $\vec{p}_{0}$ is beyond Shearer's bound.
		In the following, we prove that if $\vec{p}_{\ell-1}$ is beyond Shearer's bound,
		then $\vec{p}_{\ell}$ is also beyond Shearer's bound.
		
		Let 
		$$\vec{q} \triangleq \vec{p}'' - \Delta_p\cdot\sum_{\kappa \leq \ell-1}\left( \vec{e}_{i_{\kappa}}-\left(\frac{1-p}{p}\right)^{d-1}\cdot\frac{1}{p}\cdot\vec{e}_{j_{\kappa}}\right).$$
		Obviously, $\vec{q}\geq \vec{p}_{\ell-1}$.
		By $\vec{p}_{\ell-1}$ is beyond Shearer's bound,
		we have $\vec{q}$ is also beyond Shearer's bound.
		Note that there is a shortest path $i_{\ell},k_1,k_2,\cdots,k_n,j_{\ell}$
		between $i_{\ell}$ and $j_{\ell}$
		such that $j_{\ell+1},j_{\ell+2},\cdots,j_{|\mathcal{R}|}$ are not on the path.
		Because $\vec{q}$ is beyond Shearer's bound,
		by Lemma~\ref{lem:probtransferonpath},
		we have  
		$$\vec{q}'\triangleq \vec{q} - \Delta_p\cdot\left( \vec{e}_{i_{\ell}}-\left(\prod_{t \in [n]}\frac{1-q_{k_{t}}}{q_{k_{t}}}\right)\cdot\frac{1}{q_i}\cdot\vec{e}_{j_{\ell}}\right)$$
		is also beyond Shearer's bound.
		Meanwhile, by $(i_{\ell},r_{\ell}) \in \mathcal{R}$,
		we have 
		\begin{align*}
		q_{i_\ell} = p'_i - \Delta_p\sum_{\kappa \in \ell - 1}\mathbbm{1}(i_{\kappa} = i_{\ell}) \geq p'_i - (\gamma_i - 1)\Delta_p \geq p_i + \Delta_p.
		\end{align*}
		For each $t\in [n]$, if $k_t \not\in S_-$, we have 
		$q_{k_t} \geq p$.
		Otherwise, $k_t \in S_-$, and $k_t \neq j_{\kappa}$ for each $\kappa\geq\ell$.
		Thus, we have $\sum_{\kappa \in \ell - 1}\mathbbm{1}(j_{\kappa} = k_t) = \eta_{k_t}$.
		Therefore, 
		\begin{align*}
		q_{k_t}&=  p'_{k_t} +  (K \cdot \tau_{k_t} - \eta_{k_t})\cdot\left(\frac{1-p}{p}\right)^{d-1}\cdot \frac{\Delta_p}{p} + \sum_{\kappa \in \ell - 1}\mathbbm{1}(j_{\kappa} = k_t)\cdot\left(\frac{1-p}{p}\right)^{d-1}\cdot \frac{\Delta_p}{p} 
		\\&= p'_{k_t} +  (K \cdot \tau_{k_t} - \eta_{k_t})\cdot\left(\frac{1-p}{p}\right)^{d-1}\cdot \frac{\Delta_p}{p} + \eta_{k_t}\cdot\left(\frac{1-p}{p}\right)^{d-1}\cdot \frac{\Delta_p}{p} =  p.
		\end{align*}
		By $\text{dis}(i,j)\geq d, q_{i_{\ell}} \geq p + \Delta_p$ and $q_{k_t} \geq p$ for each $t \in [n]$,
		we have 
		$$\left(\prod_{t \in [n]}\frac{1-q_{k_{t}}}{q_{k_{t}}}\right)\cdot\frac{1}{q_i}< \left(\frac{1-p}{p}\right)^{d-1}\cdot\frac{1}{p + \Delta_p}.$$
		Thus,  by $\vec{q}'$ is beyond Shearer's bound, we have 
		$$\vec{p}_{\ell} = \vec{q} - \Delta_p\cdot\left( \vec{e}_{i_{\ell}}-\left(\frac{1-p}{p}\right)^{d-1}\cdot\frac{1}{p + \Delta_p}\cdot\vec{e}_{j_{\ell}}\right)$$
		is also beyond Shearer's bound.
	\end{proof}
	
	Thus, we have $\vec{p}_{|\mathcal{R}|} $ is beyond Shearer's bound. 
	It is easy to verify that $\vec{p}_{|\mathcal{R}|} < \vec{p}$,
	which is contradictory with that $\vec{p}$ is on Shearer's boundary.

\section{Missing part in the proof of Theorem \ref{thm:gapbetweenlattices}}
\begin{proof}[Proof of Claim \ref{claim:51}]
Observe that for each $(v,v')\in E_U^k$, if $(v,v')\notin \M$, then one of its neighboring edge $(v_0,v_1)$ is in $T_k$ and satisfies that $\delta_{v,v'}\leq \delta_{v_0,v_1}$. Here, we say two edges neighboring if they share a common vertex. Besides, note that each edge has at most $2\Delta$ neighboring edges. So	
\begin{align}\label{eq:521}
\sum_{(v_0,v_1)\in T_k} \delta_{v_0,v_1}^2 \geq \frac{1}{2\Delta}\sum_{(v,v')\in E_U^k}  \delta_{v,v'}^2.
\end{align}

Moreover, according to Lemma~\ref{lem:general2elementary} and ~\ref{lemma:elementary}, it has that
	\begin{align}\label{eq:522}
	\sum_{(v,v')\in E_U^k}  \delta_{v,v'}^2 \geq \frac{1}{|E_U^k|}\cdot \left(\sum_{(v,v')\in E_U^k }\delta_{v,v'}\right)^2
	\geq \frac{|V_U^k|\cdot\Delta^2}{|E_U^k|}\cdot \left(\lowerBB^+(G_B(G_D),\vec{p})\right)^2,
	\end{align}
	By combining Inequality \ref{eq:521}, \ref{eq:522} and the fact that $2|E_U^k|\leq |V_U^k|\Delta$, we finish the proof.
\end{proof}

\vspace{1ex}
Let $K:=(\Delta+1)|V_U|$, $d:=D+2$, $\mathcal{S}:=\{V_U^1,V_U^2,\cdots\}$, and $f(V_U^k):=T_k$. In the following, we check that all the three conditions in Lemma \ref{lem:probtransfer} hold.

\noindent\textit{\underline{Condition (a).}} That is, we want to show $|\{k:T_k\ni v\}|\leq(\Delta+1)|V_U|$ for each $v\in V_D$. Observe that if $v\in T_k$, then $v\in \N^+(V_U^k)$. So 
\[
|\{k:T_k\ni v\}|\leq |\{k:\N^+(V_U^k)\ni v\}| \leq |\{k:\N^+(v)\cap V_U^k\neq \emptyset\}|\leq \sum_{v'\in \N^+(v)}|\{k: V_U^k\ni v'\}|\leq (\Delta+1)\cdot |V_D|.
\]
The last inequality uses the fact that $h(k',u)\neq h(k,u)$ if $k\neq k'$. 

\vspace{1ex}
\noindent\textit{\underline{Condition (b).}} That is, we want to show $\dist(v,v')\leq D+2$ for any $v\in V_U^k$ and $v'\in T_k$. This is obvious, because if $v'\in T_k$, then $v'\in \N^+(V_U^k)$. 
\vspace{1ex}

\vspace{1ex}
\noindent\textit{\underline{Condition (c).}} We verify that 
\begin{align}\label{eq:conditionc}
\left(\frac{1-p_a}{p_a}\right)^{D+1}\cdot \frac{K}{p_a}\cdot\sum_{i\in S}\max\{p_i-p_a,0\} 
\leq \sum_{i\in T} \max\{p_a-p_i,0\}.
\end{align}

On one hand, noting that $\max\{(1+\eps)p^-_{v}-p_a,0\} \leq \max\{(1+\eps)p_{v}-p_a,0\} \leq q$, we have
\begin{align}\label{eq:c2}
\text{L.H.S of }(\ref{eq:conditionc})\leq \left(\frac{1-p_a}{p_a}\right)^{D+1}\cdot \frac{(\Delta+1)  |V_U|^2}{p_a}\cdot q.
\end{align}

On the other hand, observe that 
\begin{align*}
\max\{p_a - (1+\eps)p^-_{v},0\} &\geq 
p_a - (1+\eps)p^-_{v} 
=(p_a+q-(1+\eps)p^-_{v})-q
\geq (1+\eps)(p_v-p_v^-)-q\\
&\geq (p_v-p_v^-)-q, 
\end{align*}
where the last inequality is due to the assumption that $(1+\eps)\vec{p}\leq (p_a+q,\cdots,p_a+q)$. Then 
\begin{align}
\text{R.H.S of }(\ref{eq:conditionc})\geq &\left(\sum_{v\in V_U^k} (p_v-p_v^-)\right)-|\N^+(V_U^k)|q\geq \frac{2}{17}\left(\sum_{(v_0,v_1)\in T_k}\delta_{v_0,v_1}^2\right)-\Delta|V_U|q\notag\\
\geq & \frac{2}{17}\left(\lowerBB^+(G_B(G_D),\vec{p})\right)^2-\Delta|V_U|q\label{eq:c3}.
\end{align}
Putting Inequality \ref{eq:c2} and \ref{eq:c3} together and noting that $\frac{1-p_a}{p_a}\geq 1$. 